\providecommand{\prlClassLayout}{reprint}

\documentclass[
    aps,
    prl,
    \prlClassLayout,
    superscriptaddress,
    nofootinbib,
    nobibnotes
    ]{revtex4-2}

\usepackage{cmap}
\usepackage[utf8]{inputenc}
\usepackage[english]{babel}
\usepackage[T1]{fontenc}

\usepackage{amsmath}
\usepackage{mathtools}

\usepackage{amssymb}
\usepackage{bm}

\usepackage{braket}

\usepackage{amsfonts}
\usepackage{dsfont}

\usepackage{algorithmicx}
\usepackage{algpseudocode}
\usepackage{enumitem}
\usepackage{booktabs}
\usepackage{array}
\usepackage{longtable}

\usepackage[most]{tcolorbox}

\usepackage{keyval}

\usepackage{graphicx}

\usepackage{xcolor}
\definecolor{blueviolet}{rgb}{0.2, 0.2, 0.6}
\definecolor{webgreen}{rgb}{0,.5,0}
\definecolor{webbrown}{rgb}{.6,0,0}

\usepackage[pdftex,
	bookmarks=false,
	colorlinks=true,
	urlcolor=webbrown,
	linkcolor=blueviolet,
	citecolor=webgreen,
	pdfstartpage=1,
	pdfstartview={FitH},
	bookmarksopen=false,
    linktocpage=true,
	]{hyperref}

\usepackage{xurl}
\hypersetup{breaklinks=true}

\algrenewcommand\algorithmicrequire{\textbf{Input:}}
\algrenewcommand\algorithmicensure{\textbf{Output:}}
\algrenewcommand\algorithmicreturn{\textbf{return}}
\algrenewcommand\alglinenumber[1]{\scriptsize\color{black!55}#1}

\newtcolorbox{modernalgobox}[1][]{
    enhanced,
    breakable,
    colback=blueviolet!2,
    colframe=blueviolet!60!black,
    boxrule=0.55pt,
    arc=1.2mm,
    left=7pt,right=7pt,top=6pt,bottom=6pt,
    borderline west={1.4pt}{0pt}{blueviolet!80!black},
    fonttitle=\bfseries,
    coltitle=black,
    attach boxed title to top left={xshift=2mm,yshift*=-\tcboxedtitleheight/2},
    boxed title style={
        colback=blueviolet!12,
        colframe=blueviolet!60!black,
        arc=0.8mm,
        boxrule=0.55pt,
        left=4pt,right=4pt,top=2pt,bottom=2pt
    },
    #1
}

\usepackage{titletoc}
\contentsuse{section}{app}
\contentsuse{subsection}{app}
\contentsuse{subsubsection}{app}

\usepackage{amsthm}
\theoremstyle{plain}
\newtheorem{theorem}{Theorem}[section]
\newtheorem{lemma}[theorem]{Lemma}
\newtheorem{proposition}[theorem]{Proposition}
\newtheorem{corollary}[theorem]{Corollary}

\theoremstyle{definition}
\newtheorem{definition}[theorem]{Definition}
\newtheorem{assumption}[theorem]{Assumption}
\newtheorem{observation}[theorem]{Observation}
\newtheorem{protocol}[theorem]{Protocol}

\theoremstyle{remark}
\newtheorem{remark}[theorem]{Remark}

\DeclareMathOperator{\Tr}{Tr}
\DeclareMathOperator*{\argmax}{arg\,max}

\begin{document}

\title{Learning Midcircuit Measurement Backaction from Three Repeated Measurements}

\author{Chia-Tung Chu}
    \email[Contact author: ]{ctchu@uchicago.edu}
	\affiliation{Pritzker School of Molecular Engineering, The University of Chicago, Chicago, Illinois 60637, USA}
    \affiliation{Chicago Quantum Institute, Chicago, Illinois 60637, USA}
\author{Su-un Lee}
    \affiliation{Pritzker School of Molecular Engineering, The University of Chicago, Chicago, Illinois 60637, USA}
    \affiliation{Chicago Quantum Institute, Chicago, Illinois 60637, USA}
\author{Han Zheng}
    \affiliation{Pritzker School of Molecular Engineering, The University of Chicago, Chicago, Illinois 60637, USA}
    \affiliation{Chicago Quantum Institute, Chicago, Illinois 60637, USA}
\author{Senrui Chen}
    \altaffiliation[Present address: ]{Institute for Quantum Information and Matter, California Institute of Technology, Pasadena, California 91125, USA}
    \affiliation{Pritzker School of Molecular Engineering, The University of Chicago, Chicago, Illinois 60637, USA}
\author{Bibek~Pokharel}
\affiliation{IBM Quantum, T. J. Watson Research Center, Yorktown Heights, New York 10598, USA}
\author{Alireza~Seif}
\affiliation{IBM Quantum, T. J. Watson Research Center, Yorktown Heights, New York 10598, USA}
\author{Liang Jiang}
	\affiliation{Pritzker School of Molecular Engineering, The University of Chicago, Chicago, Illinois 60637, USA}
    \affiliation{Chicago Quantum Institute, Chicago, Illinois 60637, USA}

\begin{abstract}
    Accurate modeling of midcircuit measurements (MCMs) is essential for dynamic-circuit operations such as syndrome extraction, measurement-based reset, and the separation of state-preparation and measurement error.
    Unlike terminal measurement, a noisy MCM both produces a classical outcome and alters the incoming quantum state, thereby influencing subsequent circuit operations.
    This makes conventional confusion-matrix or fidelity-level characterization insufficient.
    Here, we introduce an efficient, self-consistent protocol for learning a single-qubit Z-twirled MCM instrument, retaining the readout--backaction correlations and excitation--decay asymmetry that are erased in Pauli-error descriptions.
    Remarkably, readout bit strings from only three repeated MCMs on a maximally mixed input determine all learnable parameters of the reduced instrument, up to a single unidentifiable gauge degree of freedom.
    Physicality constraints convert this nonidentifiability into narrow, gauge-aware error intervals.
    Implemented on IBM superconducting processors, the learned instrument improves Pauli-observable prediction by ${\sim}100\times$ over a conventional confusion-matrix model and reveals a $T_1$-decay dominated backaction.
    Our protocol provides a compact characterization layer for state-preparation and measurement error separation, reset optimization, and noise-aware quantum error correction.
\end{abstract}
\maketitle

\AddToHookNext{shipout/foreground}{%
    \put(0,-\paperheight){%
        \raisebox{24pt}[0pt][0pt]{\makebox[\paperwidth][c]{%
            \scriptsize\copyright\ 2026 American Physical Society. \href{https://doi.org/10.1103/y6x3-348l}{DOI: 10.1103/y6x3-348l}.%
        }}%
    }%
}

\newcommand{\prlFigOverviewRef}{Fig.~\ref{fig:overview}}
\newcommand{\prlFigResultsRef}{Fig.~\ref{fig:results}}

\emph{Introduction}---Midcircuit measurements (MCMs), which are measurements performed during circuit execution, are indispensable components in quantum computation.
They provide the syndrome information and enable measurement-based reset needed for quantum error correction~\cite{acharyaQuantumErrorCorrection2025,sundaresanDemonstratingMultiroundSubsystem2023,bluvsteinFaulttolerantNeutralatomArchitecture2026,bluvsteinArchitecturalMechanismsUniversal2025}.
They are also crucial in adaptive quantum circuits~\cite{baumerEfficientLongRangeEntanglement2024,buhrmanStatePreparationShallow2024,piroliApproximatingManybodyQuantum2024,caoMeasurementdrivenQuantumAdvantages2026}, where intermediate measurements can greatly reduce the circuit depth required to prepare long-range correlated or resource states relative to purely unitary circuits without MCMs.
These significant roles have motivated experimental demonstrations of MCMs across multiple platforms~\cite{risteFeedbackControlSolidstate2012,risteInitializationMeasurementSuperconducting2012,corcolesExploitingDynamicQuantum2021,guptaProbabilisticErrorCancellation2024,ivashkovHighFidelityMultiqubitGeneralized2024,lundConstantdepthQuantumImaginary2026}.

The reliable use of MCMs requires accurate characterization of their noise.
Such characterizations not only inform hardware diagnosis~\cite{qasimkhanSeparateEfficientCharacterization2026}, but also enable the codesign of quantum error correction codes and decoders~\cite{bonillaataidesXZZXSurfaceCode2021,sundaresanDemonstratingMultiroundSubsystem2023,tuckettUltrahighErrorThreshold2018, tuckettFaulttolerantThresholdsSurface2020, higgottImprovedDecodingCircuit2023a} and strategies for quantum error mitigation~\cite{caiQuantumErrorMitigation2023,goviaBoundingSystematicError2025,vandenbergProbabilisticErrorCancellation2023}.
However, conventional methods for characterizing terminal measurements, i.e., measurements performed at the end of a quantum circuit, do not faithfully model MCM noise.
These methods typically construct a confusion matrix, which records the probabilities of reporting each measurement outcome conditioned on the ideal input state~\cite{maciejewskiMitigationReadoutNoise2020,nationScalableMitigationMeasurement2021,bravyiMitigatingMeasurementErrors2021,laflammeAlgorithmicCoolingResolving2022,yuEfficientSeparateQuantification2025,qasimkhanSeparateEfficientCharacterization2026}.
Meanwhile, unlike terminal measurements, MCMs are followed by subsequent circuit operations, so their postmeasurement state is equally important and cannot be neglected.
Consequently, realistic MCMs display both readout and backaction errors, including measurement-induced excitation or decay, which cannot be captured by a confusion matrix alone.

\begin{figure*}[!htp]
    \centering
    \includegraphics[width=1.00\textwidth]{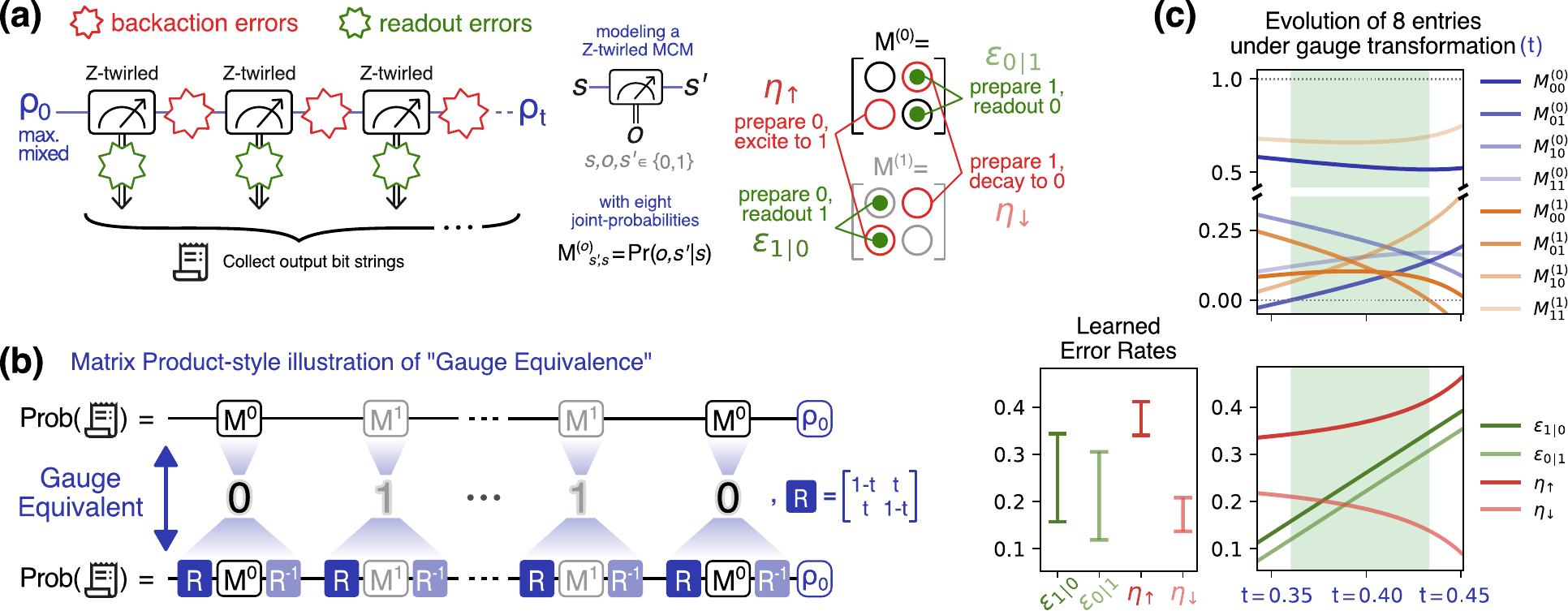}
    \caption{
        \textbf{Learning midcircuit measurements (MCMs): overview and gauge-aware reporting.}
        \textbf{(a)}~Protocol overview: starting from a maximally mixed single-qubit input, repeated $Z$-twirled MCMs generate outcome bit strings used for learning; green and red markers indicate readout and backaction error channels.
        For a single $Z$-twirled MCM, the effective model is a pair of $2\times 2$ nonnegative matrices $\{M^{(0)},M^{(1)}\}$, containing the eight joint probabilities $\Pr(o,s'\! \mid \!s)$ for obtaining the outcome $o$ and postmeasurement state $s'$, conditioned on the premeasurement state $s$.
        The two readout and two backaction error rates are obtained by summing selected entries of $\{M^{(0)},M^{(1)}\}$.
        \textbf{(b)}~Gauge-equivalent matrix-pair representations: the similarity transform $M \mapsto R(t)^{-1} M R(t)$ leaves every readout-string probability unchanged and changes the reported readout and backaction rates across gauge representatives.
        \textbf{(c)}~Simulated example for which the matrix pair used to generate the data has $F_{\mathrm{MCM}}^{\mathrm{pop}}\equiv\tfrac12(M^{(0)}_{0,0}+M^{(1)}_{1,1})=70\%$: evolution of the eight entries of $\{M^{(0)},M^{(1)}\}$ under gauge sweep.
        Green shading marks the physically allowed gauge band in which all matrix entries remain within the $[0,1]$ interval; the lower panel shows the corresponding derived readout and backaction error rates versus $t$.
        The four gauge-dependent error rates are reported as intervals over the physically allowed gauge range, quantifying unavoidable systematic uncertainty from gauge nonidentifiability.
    }
    \label{fig:overview}
\end{figure*}

Therefore, a faithful MCM characterization must capture the correlated statistics of the reported classical outcome and the resulting postmeasurement quantum state. Several approaches have been developed to characterize MCM noise, each with its own trade-offs between efficiency and the level of detail captured in the noise model, and can be broadly categorized into two classes.
Firstly, there are methods based on gate set tomography (GST) that reconstruct a self-consistent description of noisy state preparations, gates, and measurements and can be adapted to quantum instruments such as MCMs~\cite{rudingerCharacterizingMidcircuitMeasurements2022,strickerCharacterizingQuantumInstruments2022,nielsenGateSetTomography2021,jayakumarUniversalFrameworkSimultaneous2024}.
Recently, Wysocki~\emph{et al.}~\cite{wysockiDetailedInterpretableCharacterization2026} proposed and experimentally demonstrated a GST-based method that can exhaustively capture all learnable noise parameters associated with MCMs, which was implemented on a transmon qubit device.
However, GST-based methods typically require a large number of circuits (e.g., 128 GST circuits per qubit in Ref.~\cite{wysockiDetailedInterpretableCharacterization2026}), which can be resource-intensive and may not be practical for larger systems.
A related tomographic approach by Pereira, Garc\'ia-Ripoll, and Ramos reconstructs complete outcome-conditioned measurement processes; this complete reconstruction requires multiple controlled input states and interleaved unitaries from a universal gate set, increasing the circuit and control requirements~\cite{pereiraCompletePhysicalCharacterization2022}.

A complementary family of methods uses Pauli twirling or randomized benchmarking to reduce the generic noise model of MCMs to a more compact Pauli-error or fidelity-level model~\cite{gaeblerSuppressionMidcircuitMeasurement2021,harperEfficientLearningQuantum2020,chenLearnabilityPauliNoise2023,chenEfficientSelfConsistentLearning2026,zhangGeneralizedCycleBenchmarking2025,hinesPauliNoiseLearning2025,hothemMeasuringErrorRates2025,goviaRandomizedBenchmarkingSuite2023,shirizlyRandomizedBenchmarkingProtocol2025,bealeRandomizedCompilingSubsystem2023}.
This compact description supports efficient learning.
It omits the postmeasurement state conditioned on the classical outcome and the distinction between measurement-induced excitation and decay, which are crucial for applications such as repeated syndrome-extraction protocols that reuse the measured qubit.

In this Letter, we demonstrate our protocol on IBM superconducting processors using depth-three circuits, bypassing the many-circuit and nonlinear-fitting overhead of full GST.
The protocol yields closed-form gauge invariants and retains readout--backaction correlations erased by Pauli-level models.
Pauli $Z$-twirling removes coherence-sensitive information and preserves the joint statistics of classical outcomes and postmeasurement computational-basis states.
The resulting MCM noise model can be learned self-consistently from the outputs of three rounds of MCMs, and we can identify all learnable parameters associated with the readout and backaction error rates, except for one remaining gauge degree of freedom that keeps all output probabilities invariant.
Furthermore, completely positive and trace-preserving constraints significantly limit the range of this gauge parameter, resulting in tight relative precision for all learned error rates.

On independent bit-flip-interleaved circuits, our protocol predicts Pauli observables 1 to 2 orders of magnitude more accurately than a conventional confusion-matrix model.
Its accurate modeling of MCM noise and minimal circuit depth make the protocol readily applicable to other platforms supporting MCMs, including neutral-atom and trapped-ion systems.

\emph{Setup and learning protocol}---A general two-outcome MCM instrument assigns a completely positive state-transition map to each classical outcome, where the sum over outcomes remains trace-preserving.
We learn its $Z$-twirled reduction shown in Fig.~\ref{fig:overview}(a) by averaging each outcome-conditioned map with the same random $I$ or $Z$ before and after the MCM. This separates the evolution of computational-basis populations from coherences while preserving transitions between computational-basis states and correlations between the readout outcome and the postmeasurement state.
For consecutive repeated-MCM learning circuits, adjacent twirling layers combine, so the same readout-string probabilities are obtained when the hardware independently inserts a Pauli-$Z$ gate immediately before each MCM with probability $1/2$.
We therefore describe this reduced MCM instrument by its premeasurement computational-basis state $s\in\{0,1\}$, classical outcome $o\in\{0,1\}$, and postmeasurement computational-basis state $s'\in\{0,1\}$.
The effective model is a joint conditional distribution $\Pr(o,s'\! \mid \!s) \equiv p_s^{(o,s')}$ of classical outcome and postmeasurement state conditioned on the premeasurement state, arranged as a pair of $2\times 2$ nonnegative matrices:
\begin{equation}
\label{eq:mcm-matrices}
    M^{(0)} = \begin{pmatrix}
        p_0^{(0,0)} & p_1^{(0,0)} \\
        p_0^{(0,1)} & p_1^{(0,1)}
    \end{pmatrix},\quad
    M^{(1)} = \begin{pmatrix}
        p_0^{(1,0)} & p_1^{(1,0)} \\
        p_0^{(1,1)} & p_1^{(1,1)}
    \end{pmatrix},
\end{equation}
where the trace-preserving constraint requires each column of $(M^{(0)}+M^{(1)})$ to sum to one: $\sum_{o,s'} M^{(o)}_{s',s}=1$, under the assumption that the noise model is stationary, Markovian, and leakage-free.

Following this setup, all relevant error rates can be straightforwardly obtained by summing matrix elements. First, the readout error rates $\epsilon_{1|0} = \sum_{s'} M^{(1)}_{s',0}$ and $\epsilon_{0|1} = \sum_{s'}M^{(0)}_{s',1}$ represent the false-$1$ and false-$0$ readout probabilities for premeasurement states $\ket{0}$ and $\ket{1}$, respectively. Second, the backaction error rates $\eta_{\uparrow} = \sum_o M^{(o)}_{1,0}$ and $\eta_{\downarrow} = \sum_o M^{(o)}_{0,1}$ correspond to the MCM-induced excitation ($0 \to 1$) and decay ($1 \to 0$) probabilities, respectively.

Now, if we apply $L$ copies of this $Z$-twirled MCM to a single qubit initialized with computational-basis probabilities $\bm{\pi}=(\pi_0,\pi_1)^{\mathsf{T}}$, the probability of sampling a particular readout bit string $w=w_1\cdots w_L\in\{0,1\}^L$ over all bit strings of length $L$ can be expressed in a compact matrix-product form:
\begin{equation}
\label{eq:string-prob}
    p(w) = \bm{1}^{\mathsf{T}}\, M^{(w_L)} \cdots M^{(w_1)}\, \bm{\pi},
\end{equation}
where $\bm{1}=(1,1)^{\mathsf{T}}$ marginalizes over the final postmeasurement basis states.

Next, we show that sampling the readout-string probabilities up to length $3$ is enough to estimate the traces and determinants of $M^{(0)}$ and $M^{(1)}$ from a second-order recurrence relation.
Specifically, we start from the Cayley--Hamilton identity for any $2\times2$ real matrix $A$, $A^2 - \Tr(A)\,A + \det(A)\,I = 0$.
For notational convenience, we define the probability of measuring $n$-repeated zeros and ones by $S_n \coloneqq p(0^n)$ and $T_n \coloneqq p(1^n)$, where digits in the argument of $p(\cdot)$ denote readout bit strings.
Then we obtain
\begin{equation}
\label{eq:recurrence}
\begin{aligned}
    S_{n+2} &= \Tr(M^{(0)})\,S_{n+1} - \det(M^{(0)})\, S_n, \\
    T_{n+2} &= \Tr(M^{(1)})\,T_{n+1} - \det(M^{(1)})\, T_n.
\end{aligned}
\end{equation}
Setting $n=0,1$ and letting $S_0=T_0=1$ give us the closed-form estimator pairs
\begin{equation}
\label{eq:tr-det-m0-main}
\begin{aligned}
    \Tr(M^{(0)}) &=\frac{S_1S_2-S_3}{S_1^2-S_2},
    &
    \det(M^{(0)}) &=\frac{S_2^2-S_1S_3}{S_1^2-S_2}, \\
    \Tr(M^{(1)}) &=\frac{T_1T_2-T_3}{T_1^2-T_2},
    &
    \det(M^{(1)}) &=\frac{T_2^2-T_1T_3}{T_1^2-T_2}.
\end{aligned}
\end{equation}
outside the singular cases $S_1^2=S_2$ and $T_1^2=T_2$.

Ultimately, extracting the defined error rates requires determining the eight entries of the matrix pair $(M^{(0)},M^{(1)})$. As shown above, these can be inferred by sampling the probability distribution of bit strings up to length $3$, provided the initial input state $\bm{\pi}$ to the sampling circuit is known.
For the hardware implementation, when the target Pauli gates $\{I,X,Y,Z\}$ are treated as ideal, their uniform average gives the maximally mixed input $I/2$, with computational-basis probability vector $\bm{\pi}=\tfrac12\bm{1}$.
This choice introduces a degree of freedom in the parametrization of the matrix pair.
This is clear from parameter counting: the eight matrix entries are restricted by two independent trace-preserving constraints, alongside five independent constraints extracted from the sampled probability distribution---specifically, the scalars $\{p(0),\,\Tr(M^{(0)}),\,\det(M^{(0)}),\,\Tr(M^{(1)}),\,\det(M^{(1)})\}$.
This yields exactly $8 - 2 - 5 = 1$ unidentifiable degree of freedom.
To better understand the physical nature of this remaining parameter freedom, we next examine it in terms of a gauge transformation parametrized by a continuous variable $t$.

\begin{figure*}[!t]
    \centering
    \includegraphics[width=1.0\textwidth]{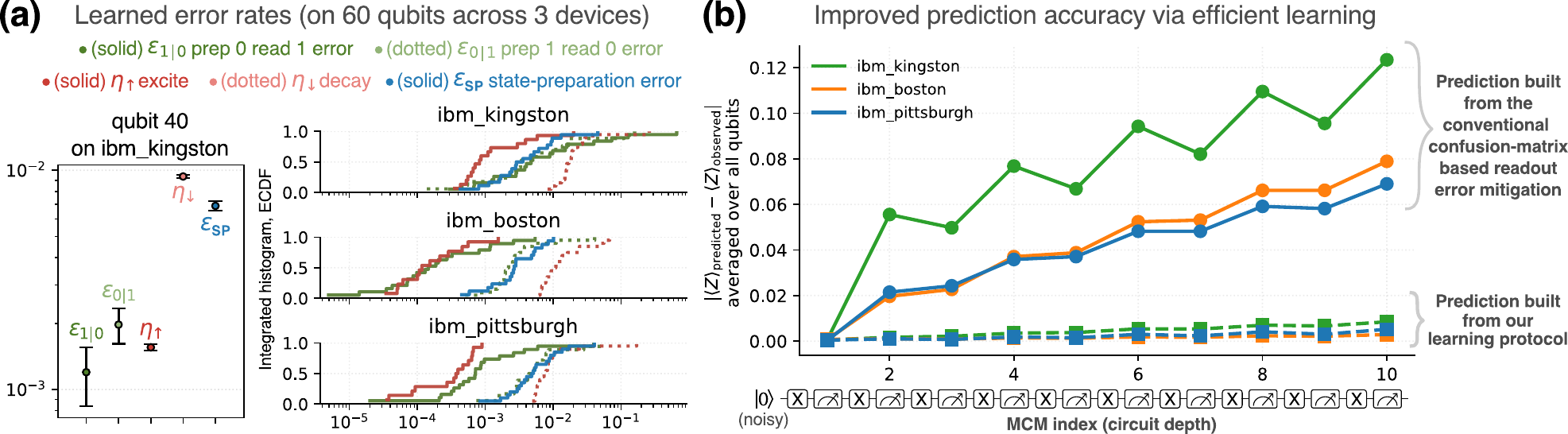}
    \caption{
        \textbf{IBM hardware results on superconducting processors.}
        \textbf{(a)}~Five inferred gauge-dependent error rates on 60 selected qubits across \texttt{ibm\_boston}, \texttt{ibm\_kingston}, and \texttt{ibm\_pittsburgh} (20 qubits each).
        \textbf{Left panel:} for qubit 40 on \texttt{ibm\_kingston}, four MCM rates are derived from the learned instrument.
        The preparation error is inferred afterward using noisy-$\ket{0}$ data and the learned readout matrix.
        Capped bars show the gauge bands; circles mark their individual midpoints.
        \textbf{Right panel:} integrated histograms (empirical cumulative distribution functions) of the five inferred error rates, reported as the midpoints of their associated gauge bands, across all 20 qubits on each device.
        For all three devices, we see the decay rates $\eta_{\downarrow}$ ($1\!\to\!0$) exceed excitation rates $\eta_{\uparrow}$ ($0\!\to\!1$) by 1 to 2 orders of magnitude, consistent with spontaneous-emission ($T_1$)-dominated physics.
        \textbf{(b)}~Absolute deviation of the predicted Pauli-$Z$ observable from the hardware data in independent $X$-interleaved validation circuits reserved for prediction tests (averaged over all 20 qubits on each device), comparing \textit{(i)} the full learned-instrument predictor (dashed) and \textit{(ii)} a confusion-matrix-based predictor with no MCM state-update dynamics (solid).
        The full-instrument predictor remains accurate at larger circuit depths.
        Over the same range, the confusion-matrix-based predictor's deviation grows to 1 or 2 orders of magnitude above it.
    }
    \label{fig:results}
\end{figure*}

\emph{Gauge structure and physical constraints}---We can explicitly construct this transformation using the invertible gauge matrix
\begin{equation}
\label{eq:gauge-R}
      R(t) = \begin{pmatrix} 1-t & t \\ t & 1-t \end{pmatrix},\quad \det R(t) = 1-2t \neq 0,
\end{equation}
for $t\in\mathbb{R}\setminus\{\tfrac{1}{2}\}$.
Since $R(t)\bm{1}=\bm{1}$ and $\bm{1}^{\mathsf{T}}R(t)^{-1}=\bm{1}^{\mathsf{T}}$, applying a simultaneous similarity transformation $M^{(o)}\mapsto \widetilde{M}^{(o)} \coloneqq R(t)^{-1}M^{(o)}R(t)$ to the matrix-product model in Eq.~\eqref{eq:string-prob} leaves every bit string's sampling probability unchanged:
\begin{equation}
\label{eq:gauge-invariance}
    \bm{1}^{\mathsf{T}} \widetilde{M}^{(w_L)}\!\cdots\widetilde{M}^{(w_1)} \tfrac12\bm{1}
    = \bm{1}^{\mathsf{T}} M^{(w_L)}\!\cdots M^{(w_1)} \tfrac12\bm{1}.
\end{equation}
As illustrated in Fig.~\ref{fig:overview}(b), the sampled readout-string probability distribution cannot distinguish between matrix parametrizations related by this similarity transform. Conversely, under the conditions stated in the \hyperref[sec:supplement]{Supplemental Material}, any two parameterizations of $\{M^{(0)}, M^{(1)}\}$ yielding identical sampling distributions from $\bm{\pi}=\tfrac{1}{2}\bm{1}$ must be related by a unique $R(t)$. The variable $t$ therefore acts as a residual gauge, explicitly absorbing the single unidentifiable degree of freedom left by our five Cayley--Hamilton estimators.

Because gauge-dependent quantities such as the readout ($\epsilon_{1|0}, \epsilon_{0|1}$) and backaction ($\eta_{\uparrow}, \eta_{\downarrow}$) error rates are evaluated using the entries of the transformed matrices $\widetilde{M}^{(o)}(t)$, choosing an arbitrary $t$ would assign them arbitrary values. We strictly bound this ambiguity by enforcing the physical constraints of the reduced MCM instrument: every entry of $\widetilde{M}^{(0)}(t)$ and $\widetilde{M}^{(1)}(t)$ must remain in $[0,1]$. Similar physicality constraints are also used to bound unlearnable degrees of freedom in Clifford-gate Pauli-noise learning~\cite{chenLearnabilityPauliNoise2023} and estimate stochastic Pauli-error rates in randomized MCM layers~\cite{hinesPauliNoiseLearning2025}.
These entrywise conditions restrict $t$ to a computable \emph{physically allowed gauge set} $\mathcal{T}_\mathrm{phys}\subset\mathbb{R}\setminus\{\tfrac{1}{2}\}$, which forms a finite union of intervals. We retain the original computational-basis labeling by restricting to $\mathcal{T}_\mathrm{phys}^{(+)}\equiv\mathcal{T}_\mathrm{phys}\cap\{t<\tfrac{1}{2}\}$ and report any gauge-dependent derived quantity $q(t)$ over this set as the \emph{gauge band}
\begin{equation}
\label{eq:gauge-band}
      [q]_{\mathcal{T}_\mathrm{phys}^{(+)}} \coloneqq [\min_{t\in\mathcal{T}_\mathrm{phys}^{(+)}} q(t),\; \max_{t\in\mathcal{T}_\mathrm{phys}^{(+)}} q(t)],
\end{equation}
quantifying unavoidable systematic uncertainty from gauge nonidentifiability in the readout-string data [Fig.~\ref{fig:overview}(c)].

\emph{IBM hardware results}---We apply the protocol to 20 selected qubits on each of three 156-qubit IBM Heron processors: \texttt{ibm\_kingston} (r2), \texttt{ibm\_pittsburgh} (r3), and \texttt{ibm\_boston} (r3).
For each qubit, depth-$3$ learning circuits are preceded by uniform input randomization over target Pauli gates $\{I,X,Y,Z\}$ applied to noisy $\ket{0}$.
Each Pauli setting pools independently randomized pre-MCM $Z$ patterns and uses $N=10^6$ shots, totaling $4\times10^6$ learning shots per qubit.
Optimization-level-0 transpilation uses no dynamical decoupling or error suppression.
The learned effective instrument includes physical errors associated with the compiled pre-MCM $Z$ layers.
The input Pauli gates are treated as ideal.
Equation~\eqref{eq:tr-det-m0-main}, together with $p(0)$, gives the five gauge invariants used to reconstruct the gauge family and the bands in Eq.~\eqref{eq:gauge-band}.
Further implementation and validation details are provided in the \hyperref[sec:supplement]{Supplemental Material}.

Independent validation circuits contain $L_{\mathrm{val}}=10$ repeated $[X\text{-gate}\to\text{MCM}]$ blocks, with each $X$ modeled by a separately characterized noisy bit-flip transition matrix.
Figure~\ref{fig:results}(b) reports $|z^{\mathrm{pred}}_\ell-z^{\mathrm{obs}}_\ell|$, where $z_\ell=2p_\ell(0)-1$ is the Pauli-$Z$ observable, averaged over each device's 20 qubits.
Both predictors use the same preparation error $\epsilon_{\mathrm{SP}}$, inferred after instrument learning from noisy-$\ket{0}$ data and the learned readout matrix, and the same separately characterized noisy $X$ gate.
The solid baseline collapses the learned errors into a confusion matrix with no MCM state-update dynamics.
The dashed predictor updates the population state using the full learned instrument.

The two predictors initially reproduce the same observable. At later indices, the full instrument has a Pauli-$Z$ deviation that is 1 to 2 orders of magnitude smaller because it accounts for accumulated MCM-induced backaction through the sequence.

Across qubits and devices, backaction decay rates $\eta_{\downarrow}$ ($1\!\to\!0$) exceed excitation rates $\eta_{\uparrow}$ ($0\!\to\!1$) by 1 to 2 orders of magnitude, consistent with spontaneous-emission ($T_1$-decay)-dominated physics within the reduced model [Fig.~\ref{fig:results}(a)].
Moreover, the gauge bands for these error rates are typically narrow ($<10\%$ relative width), indicating that the physical constraints leave less ambiguity than the observed decay-excitation separation for all qubits.
This asymmetry is useful beyond the validation circuits because it reveals hardware-diagnostic information that would be averaged out in a Pauli error model with a single Pauli-X error rate.

During the learning stage, appending more MCMs to the sampling circuit ($L>3$) allows a self-consistency check using the Cayley--Hamilton recurrence~\eqref{eq:recurrence}.
For the length-$10$ diagnostic, the empirical residuals $|\hat{S}_{n+2} - \Tr(M^{(0)})\hat{S}_{n+1} + \det(M^{(0)})\hat{S}_n|$ for $n=0,\ldots,8$ are consistent with the stationary model at the tested sensitivity.

Independent lower-shot IBM runs show that, at $N=10000$ shots per Pauli input setting, 99.4\% of the tested reconstructions succeeded; among successful reconstructions, the results remained close to the million-shot reference.
The numerical simulations show how the required shot count changes with the instrument: when the estimator denominators in Eq.~\eqref{eq:tr-det-m0-main} approach zero, finite-shot noise may be amplified.
They also show that larger departures from the stationary two-state model produce clearer signals in longer-sequence checks.
Detailed results are provided in the \hyperref[sec:supplement]{Supplemental Material}.

\emph{Applications and discussion}---Once $(M^{(0)},M^{(1)})$ is learned, the same instrument can be used for several subsequent characterization and control tasks.
First, the learned instrument supports gauge-aware state-preparation and measurement error separation.
By preceding a single MCM with a Pauli-axis dephasing map and an appropriate single-qubit rotation, one can estimate the Bloch-vector components along different axes of the noisily prepared initial state.
Second, it supports quantitative comparison of reset strategies.
Treating the conditional $X$ as ideal, the effective state-transition map of a measurement-based reset block can be expressed as
\begin{equation}
\label{eq:reset-rmb-main}
(R_{\mathrm{MB}})_{s',s}=M^{(0)}_{s',s}+M^{(1)}_{1-s',s},
\end{equation}
so its $k$-step fidelity follows directly from powers of $R_{\mathrm{MB}}$.
This deterministic reset can be compared with all-zero heralded postselection.
Postselection can yield higher conditional fidelity, accompanied by a lower success probability and greater retry overhead~\cite{johnsonHeraldedStatePreparation2012,risteInitializationMeasurementSuperconducting2012}.

These results demonstrate that accurate MCM characterization must retain state-transition information beyond mere readout error in classical outcomes.
By capturing the cumulative disturbance induced by MCMs, our reduced-instrument model successfully resolves the nonunital decay--excitation imbalance consistent with relaxation-dominated superconducting hardware.
More broadly, the same self-consistent strategy can be applied to characterizing other nonunital processes, including relaxation and autonomous cooling mechanisms~\cite{boykinAlgorithmicCoolingScalable2002,aamirThermallyDrivenQuantum2025}, where excitation-decay asymmetry can qualitatively alter noisy-circuit behavior, error-correction performance, and simulation complexity~\cite{shtankoComplexityLocalQuantum2025,feffermanEffectNonunitalNoise2024}.

\begin{acknowledgments}
    \emph{Acknowledgments}---The authors thank Argyris Giannisis Manes, Oles Shtanko, Abhinav Deshpande, and Maika Takita for helpful discussions.
\end{acknowledgments}

The results presented in this work are the subject of a patent application (Application No. 19/449634), filed on 15 Jan. 2026.

\section*{Data Availability}
There are no publicly available research data or software supporting this manuscript. Requests for further information or data should be sent to the authors.

\nocite{chenEnhancingQuantumNoise2025,bravyiHighthresholdLowoverheadFaulttolerant2024}
\bibliography{references}

@article{acharyaQuantumErrorCorrection2025,
  title = {Quantum Error Correction below the Surface Code Threshold},
  year = 2025,
  month = feb,
  journal = {Nature},
  volume = {638},
  number = {8052},
  pages = {920--926},
  publisher = {Nature Publishing Group},
  issn = {1476-4687},
  doi = {10.1038/s41586-024-08449-y},
  url = {https://www.nature.com/articles/s41586-024-08449-y},
  copyright = {2024 The Author(s)},
  author = {Acharya, Rajeev and Abanin, Dmitry A. and Aghababaie-Beni, Laleh and others}
}

@article{sundaresanDemonstratingMultiroundSubsystem2023,
  title = {Demonstrating Multi-Round Subsystem Quantum Error Correction Using Matching and Maximum Likelihood Decoders},
  year = 2023,
  month = may,
  journal = {Nature Communications},
  volume = {14},
  number = {1},
  pages = {2852},
  publisher = {Nature Publishing Group},
  issn = {2041-1723},
  doi = {10.1038/s41467-023-38247-5},
  url = {https://www.nature.com/articles/s41467-023-38247-5},
  copyright = {2023 The Author(s)},
  author = {Sundaresan, Neereja and Yoder, Theodore J. and Kim, Youngseok and others}
}

@article{bluvsteinFaulttolerantNeutralatomArchitecture2026,
  title = {A Fault-Tolerant Neutral-Atom Architecture for Universal Quantum Computation},
  year = 2026,
  month = jan,
  journal = {Nature},
  volume = {649},
  number = {8095},
  pages = {39--46},
  publisher = {Nature Publishing Group},
  issn = {1476-4687},
  doi = {10.1038/s41586-025-09848-5},
  url = {https://www.nature.com/articles/s41586-025-09848-5},
  copyright = {2025 The Author(s)},
  author = {Bluvstein, Dolev and Geim, Alexandra A. and Li, Sophie H. and others}
}

@misc{bluvsteinArchitecturalMechanismsUniversal2025,
  title = {Architectural Mechanisms of a Universal Fault-Tolerant Quantum Computer},
  year = 2025,
  month = jun,
  number = {arXiv:2506.20661},
  eprint = {2506.20661},
  primaryclass = {quant-ph},
  publisher = {arXiv},
  doi = {10.48550/arXiv.2506.20661},
  url = {http://arxiv.org/abs/2506.20661},
  archiveprefix = {arXiv},
  author = {Bluvstein, Dolev and Geim, Alexandra A. and Li, Sophie H. and others}
}

@article{baumerEfficientLongRangeEntanglement2024,
  title = {Efficient {{Long-Range Entanglement Using Dynamic Circuits}}},
  year = 2024,
  month = aug,
  journal = {PRX Quantum},
  volume = {5},
  number = {3},
  pages = {030339},
  publisher = {American Physical Society},
  doi = {10.1103/PRXQuantum.5.030339},
  url = {https://link.aps.org/doi/10.1103/PRXQuantum.5.030339},
  author = {B\"aumer, Elisa and Tripathi, Vinay and Wang, Derek S. and others}
}

@article{buhrmanStatePreparationShallow2024,
  title = {State Preparation by Shallow Circuits Using Feed Forward},
  year = 2024,
  month = dec,
  journal = {Quantum},
  volume = {8},
  pages = {1552},
  publisher = {Verein zur F\"orderung des Open Access Publizierens in den Quantenwissenschaften},
  doi = {10.22331/q-2024-12-09-1552},
  url = {https://quantum-journal.org/papers/q-2024-12-09-1552/},
  author = {Buhrman, Harry and Folkertsma, Marten and Loff, Bruno and others}
}

@article{piroliApproximatingManybodyQuantum2024,
  title = {Approximating Many-Body Quantum States with Quantum Circuits and Measurements},
  author = {Piroli, Lorenzo and Styliaris, Georgios and Cirac, J. Ignacio},
  year = 2024,
  month = dec,
  journal = {Physical Review Letters},
  volume = {133},
  number = {23},
  pages = {230401},
  publisher = {American Physical Society},
  doi = {10.1103/PhysRevLett.133.230401},
  url = {https://link.aps.org/doi/10.1103/PhysRevLett.133.230401}
}

@article{caoMeasurementdrivenQuantumAdvantages2026,
  title = {Measurement-Driven Quantum Advantages in Shallow Circuits},
  author = {Cao, Chenfeng and Eisert, Jens},
  year = 2026,
  month = feb,
  journal = {Physical Review Letters},
  volume = {136},
  number = {8},
  pages = {080601},
  publisher = {American Physical Society},
  doi = {10.1103/4b99-xmqn},
  url = {https://link.aps.org/doi/10.1103/4b99-xmqn}
}

@article{risteFeedbackControlSolidstate2012,
  title = {Feedback Control of a Solid-State Qubit Using High-Fidelity Projective Measurement},
  year = 2012,
  month = dec,
  journal = {Physical Review Letters},
  volume = {109},
  number = {24},
  pages = {240502},
  publisher = {American Physical Society},
  doi = {10.1103/PhysRevLett.109.240502},
  url = {https://link.aps.org/doi/10.1103/PhysRevLett.109.240502},
  author = {Rist{\`e}, D. and Bultink, C. C. and Lehnert, K. W. and DiCarlo, L.}
}

@article{risteInitializationMeasurementSuperconducting2012,
  title = {Initialization by Measurement of a Superconducting Quantum Bit Circuit},
  year = 2012,
  month = aug,
  journal = {Physical Review Letters},
  volume = {109},
  number = {5},
  pages = {050507},
  publisher = {American Physical Society},
  doi = {10.1103/PhysRevLett.109.050507},
  url = {https://link.aps.org/doi/10.1103/PhysRevLett.109.050507},
  author = {Rist{\`e}, D. and van Leeuwen, J. G. and Ku, H.-S. and Lehnert, K. W. and DiCarlo, L.}
}

@article{corcolesExploitingDynamicQuantum2021,
  title = {Exploiting {{Dynamic Quantum Circuits}} in a {{Quantum Algorithm}} with {{Superconducting Qubits}}},
  year = 2021,
  month = aug,
  journal = {Physical Review Letters},
  volume = {127},
  number = {10},
  pages = {100501},
  publisher = {American Physical Society},
  doi = {10.1103/PhysRevLett.127.100501},
  url = {https://link.aps.org/doi/10.1103/PhysRevLett.127.100501},
  author = {C{\'o}rcoles, A. D. and Takita, Maika and Inoue, Ken and Lekuch, Scott and Minev, Zlatko K. and Chow, Jerry M. and Gambetta, Jay M.}
}

@article{guptaProbabilisticErrorCancellation2024,
  title = {Probabilistic Error Cancellation for Dynamic Quantum Circuits},
  year = 2024,
  month = jun,
  journal = {Physical Review A},
  volume = {109},
  number = {6},
  pages = {062617},
  publisher = {American Physical Society},
  doi = {10.1103/PhysRevA.109.062617},
  url = {https://link.aps.org/doi/10.1103/PhysRevA.109.062617},
  author = {Gupta, Riddhi S. and van den Berg, Ewout and Takita, Maika and Rist{\`e}, Diego and Temme, Kristan and Kandala, Abhinav}
}

@article{ivashkovHighFidelityMultiqubitGeneralized2024,
  title = {High-{{Fidelity}}, {{Multiqubit Generalized Measurements}} with {{Dynamic Circuits}}},
  author = {Ivashkov, Petr and Uchehara, Gideon and Jiang, Liang and Wang, Derek S. and Seif, Alireza},
  year = 2024,
  month = jul,
  journal = {PRX Quantum},
  volume = {5},
  number = {3},
  pages = {030315},
  publisher = {American Physical Society},
  doi = {10.1103/PRXQuantum.5.030315},
  url = {https://link.aps.org/doi/10.1103/PRXQuantum.5.030315}
}

@misc{lundConstantdepthQuantumImaginary2026,
  title = {Constant-Depth Quantum Imaginary Time Evolution Using Dynamic Fan-out Circuits},
  year = 2026,
  month = mar,
  number = {arXiv:2603.05156},
  eprint = {2603.05156},
  primaryclass = {quant-ph},
  publisher = {arXiv},
  doi = {10.48550/arXiv.2603.05156},
  url = {http://arxiv.org/abs/2603.05156},
  archiveprefix = {arXiv},
  author = {Lund, Albert and Magnusson, Erika and Dobrautz, Werner and others}
}

@article{qasimkhanSeparateEfficientCharacterization2026,
  title = {Separate and Efficient Characterization of State-Preparation and Measurement Errors Using Single-Qubit Operations},
  author = {Qasim Khan, Muhammad and Norris, Leigh M and Viola, Lorenza},
  year = 2026,
  month = may,
  journal = {Quantum Science and Technology},
  volume = {11},
  number = {2},
  pages = {025051},
  publisher = {IOP Publishing},
  issn = {2058-9565},
  doi = {10.1088/2058-9565/ae65d7},
  url = {https://doi.org/10.1088/2058-9565/ae65d7}
}

@article{bonillaataidesXZZXSurfaceCode2021,
  title = {The {{XZZX}} Surface Code},
  year = 2021,
  month = apr,
  journal = {Nature Communications},
  volume = {12},
  number = {1},
  pages = {2172},
  publisher = {Nature Publishing Group},
  issn = {2041-1723},
  doi = {10.1038/s41467-021-22274-1},
  url = {https://www.nature.com/articles/s41467-021-22274-1},
  copyright = {2021 The Author(s)},
  author = {Bonilla Ataides, J. Pablo and Tuckett, David K. and Bartlett, Stephen D. and others}
}

@article{tuckettUltrahighErrorThreshold2018,
  title = {Ultrahigh Error Threshold for Surface Codes with Biased Noise},
  author = {Tuckett, David K. and Bartlett, Stephen D. and Flammia, Steven T.},
  year = 2018,
  month = jan,
  journal = {Physical Review Letters},
  volume = {120},
  number = {5},
  pages = {050505},
  publisher = {American Physical Society},
  doi = {10.1103/PhysRevLett.120.050505},
  url = {https://link.aps.org/doi/10.1103/PhysRevLett.120.050505}
}

@article{tuckettFaulttolerantThresholdsSurface2020,
  title = {Fault-Tolerant Thresholds for the Surface Code in Excess of \$5\%\$ under Biased Noise},
  year = 2020,
  month = mar,
  journal = {Physical Review Letters},
  volume = {124},
  number = {13},
  pages = {130501},
  publisher = {American Physical Society},
  doi = {10.1103/PhysRevLett.124.130501},
  url = {https://link.aps.org/doi/10.1103/PhysRevLett.124.130501},
  author = {Tuckett, David K. and Bartlett, Stephen D. and Flammia, Steven T. and others}
}

@article{higgottImprovedDecodingCircuit2023a,
  title = {Improved Decoding of Circuit Noise and Fragile Boundaries of Tailored Surface Codes},
  year = 2023,
  month = jul,
  journal = {Physical Review X},
  volume = {13},
  number = {3},
  pages = {031007},
  publisher = {American Physical Society},
  doi = {10.1103/PhysRevX.13.031007},
  url = {https://link.aps.org/doi/10.1103/PhysRevX.13.031007},
  author = {Higgott, Oscar and Bohdanowicz, Thomas C. and Kubica, Aleksander and others}
}

@article{caiQuantumErrorMitigation2023,
  title = {Quantum Error Mitigation},
  year = 2023,
  month = dec,
  journal = {Reviews of Modern Physics},
  volume = {95},
  number = {4},
  eprint = {2210.00921},
  pages = {045005},
  doi = {10.1103/RevModPhys.95.045005},
  url = {https://link.aps.org/doi/10.1103/RevModPhys.95.045005},
  archiveprefix = {arXiv},
  author = {Cai, Zhenyu and Babbush, Ryan and Benjamin, Simon C. and others}
}

@article{goviaBoundingSystematicError2025,
  title = {Bounding the {{Systematic Error}} in {{Quantum Error Mitigation}} Due to {{Model Violation}}},
  year = 2025,
  month = mar,
  journal = {PRX Quantum},
  volume = {6},
  number = {1},
  pages = {010354},
  publisher = {American Physical Society},
  doi = {10.1103/PRXQuantum.6.010354},
  url = {https://link.aps.org/doi/10.1103/PRXQuantum.6.010354},
  author = {Govia, L. C. G. and Majumder, S. and Barron, S. V. and Mitchell, B. and Seif, A. and Kim, Y. and Wood, C. J. and Pritchett, E. J. and Merkel, S. T. and McKay, D. C.}
}

@article{vandenbergProbabilisticErrorCancellation2023,
  title = {Probabilistic Error Cancellation with Sparse {{Pauli}}--{{Lindblad}} Models on Noisy Quantum Processors},
  year = 2023,
  month = aug,
  journal = {Nature Physics},
  volume = {19},
  number = {8},
  pages = {1116--1121},
  issn = {1745-2473, 1745-2481},
  doi = {10.1038/s41567-023-02042-2},
  url = {https://www.nature.com/articles/s41567-023-02042-2},
  author = {van den Berg, Ewout and Minev, Zlatko K. and Kandala, Abhinav and others}
}

@article{maciejewskiMitigationReadoutNoise2020,
  title = {Mitigation of Readout Noise in Near-Term Quantum Devices by Classical Post-Processing Based on Detector Tomography},
  author = {Maciejewski, Filip B. and Zimbor{\'a}s, Zolt{\'a}n and Oszmaniec, Micha{\l}},
  year = 2020,
  month = apr,
  journal = {Quantum},
  volume = {4},
  pages = {257},
  publisher = {Verein zur F\"orderung des Open Access Publizierens in den Quantenwissenschaften},
  doi = {10.22331/q-2020-04-24-257},
  url = {https://quantum-journal.org/papers/q-2020-04-24-257/}
}

@article{nationScalableMitigationMeasurement2021,
  title = {Scalable {{Mitigation}} of {{Measurement Errors}} on {{Quantum Computers}}},
  author = {Nation, Paul D. and Kang, Hwajung and Sundaresan, Neereja and Gambetta, Jay M.},
  year = 2021,
  journal = {PRX Quantum},
  volume = {2},
  number = {4},
  doi = {10.1103/PRXQuantum.2.040326},
  pages = {040326}
}

@article{bravyiMitigatingMeasurementErrors2021,
  title = {Mitigating Measurement Errors in Multiqubit Experiments},
  author = {Bravyi, Sergey and Sheldon, Sarah and Kandala, Abhinav and McKay, David C. and Gambetta, Jay M.},
  year = 2021,
  journal = {Physical Review A},
  volume = {103},
  number = {4},
  doi = {10.1103/PhysRevA.103.042605},
  pages = {042605}
}

@article{laflammeAlgorithmicCoolingResolving2022,
  title = {Algorithmic Cooling for Resolving State Preparation and Measurement Errors in Quantum Computing},
  author = {Laflamme, Raymond and Lin, Junan and Mor, Tal},
  year = 2022,
  month = jul,
  journal = {Physical Review A},
  volume = {106},
  number = {1},
  eprint = {2203.08114},
  pages = {012439},
  doi = {10.1103/PhysRevA.106.012439},
  url = {https://link.aps.org/doi/10.1103/PhysRevA.106.012439},
  archiveprefix = {arXiv}
}

@article{yuEfficientSeparateQuantification2025,
  title = {Efficient Separate Quantification of State Preparation Errors and Measurement Errors on Quantum Computers and Their Mitigation},
  author = {Yu, Hongye and Wei, Tzu-Chieh},
  year = 2025,
  month = may,
  journal = {Quantum},
  volume = {9},
  eprint = {2310.18881},
  primaryclass = {quant-ph},
  pages = {1724},
  issn = {2521-327X},
  doi = {10.22331/q-2025-05-05-1724},
  url = {http://arxiv.org/abs/2310.18881},
  archiveprefix = {arXiv}
}

@article{rudingerCharacterizingMidcircuitMeasurements2022,
  title = {Characterizing {{Midcircuit Measurements}} on a {{Superconducting Qubit Using Gate Set Tomography}}},
  year = 2022,
  month = jan,
  journal = {Physical Review Applied},
  volume = {17},
  number = {1},
  pages = {014014},
  publisher = {American Physical Society},
  doi = {10.1103/PhysRevApplied.17.014014},
  url = {https://link.aps.org/doi/10.1103/PhysRevApplied.17.014014},
  author = {Rudinger, Kenneth and Ribeill, Guilhem J. and Govia, Luke C. G. and Ware, Matthew and Nielsen, Erik and Young, Kevin and Ohki, Thomas A. and Blume-Kohout, Robin and Proctor, Timothy}
}

@article{strickerCharacterizingQuantumInstruments2022,
  title = {Characterizing Quantum Instruments: From Nondemolition Measurements to Quantum Error Correction},
  shorttitle = {Characterizing Quantum Instruments},
  year = 2022,
  month = aug,
  journal = {PRX Quantum},
  volume = {3},
  number = {3},
  pages = {030318},
  publisher = {American Physical Society},
  doi = {10.1103/PRXQuantum.3.030318},
  url = {https://link.aps.org/doi/10.1103/PRXQuantum.3.030318},
  author = {Stricker, Roman and Vodola, Davide and Erhard, Alexander and others}
}

@article{nielsenGateSetTomography2021,
  title = {Gate Set Tomography},
  year = 2021,
  month = oct,
  journal = {Quantum},
  volume = {5},
  pages = {557},
  publisher = {Verein zur F\"orderung des Open Access Publizierens in den Quantenwissenschaften},
  doi = {10.22331/q-2021-10-05-557},
  url = {https://quantum-journal.org/papers/q-2021-10-05-557/},
  author = {Nielsen, Erik and Gamble, John King and Rudinger, Kenneth and others}
}

@article{jayakumarUniversalFrameworkSimultaneous2024,
  title = {Universal Framework for Simultaneous Tomography of Quantum States and {{SPAM}} Noise},
  year = 2024,
  month = jul,
  journal = {Quantum},
  volume = {8},
  eprint = {2308.15648},
  primaryclass = {quant-ph},
  pages = {1426},
  issn = {2521-327X},
  doi = {10.22331/q-2024-07-30-1426},
  url = {http://arxiv.org/abs/2308.15648},
  archiveprefix = {arXiv},
  author = {Jayakumar, Abhijith and Chessa, Stefano and Coffrin, Carleton and others}
}

@misc{wysockiDetailedInterpretableCharacterization2026,
  title = {Detailed, Interpretable Characterization of Mid-Circuit Measurement on a Transmon Qubit},
  year = 2026,
  month = feb,
  number = {arXiv:2602.03938},
  eprint = {2602.03938},
  primaryclass = {quant-ph},
  publisher = {arXiv},
  doi = {10.48550/arXiv.2602.03938},
  url = {http://arxiv.org/abs/2602.03938},
  archiveprefix = {arXiv},
  author = {Wysocki, Piper C. and Burkhart, Luke D. and Morocco, Madeline H. and others}
}

@article{gaeblerSuppressionMidcircuitMeasurement2021,
  title = {Suppression of Midcircuit Measurement Crosstalk Errors with Micromotion},
  year = 2021,
  month = dec,
  journal = {Physical Review A},
  volume = {104},
  number = {6},
  pages = {062440},
  publisher = {American Physical Society},
  doi = {10.1103/PhysRevA.104.062440},
  url = {https://link.aps.org/doi/10.1103/PhysRevA.104.062440},
  author = {Gaebler, J. P. and Baldwin, C. H. and Moses, S. A. and Dreiling, J. M. and Figgatt, C. and Foss-Feig, M. and Hayes, D. and Pino, J. M.}
}

@article{harperEfficientLearningQuantum2020,
  title = {Efficient Learning of Quantum Noise},
  author = {Harper, Robin and Flammia, Steven T. and Wallman, Joel J.},
  year = 2020,
  month = dec,
  journal = {Nature Physics},
  volume = {16},
  number = {12},
  pages = {1184--1188},
  publisher = {Nature Publishing Group},
  issn = {1745-2481},
  doi = {10.1038/s41567-020-0992-8},
  url = {https://www.nature.com/articles/s41567-020-0992-8},
  copyright = {2020 The Author(s), under exclusive licence to Springer Nature Limited}
}

@article{chenLearnabilityPauliNoise2023,
  title = {The Learnability of {{Pauli}} Noise},
  year = 2023,
  month = jan,
  journal = {Nature Communications},
  volume = {14},
  number = {1},
  pages = {52},
  publisher = {Nature Publishing Group},
  issn = {2041-1723},
  doi = {10.1038/s41467-022-35759-4},
  url = {https://www.nature.com/articles/s41467-022-35759-4},
  copyright = {2023 The Author(s)},
  author = {Chen, Senrui and Liu, Yunchao and Otten, Matthew and others}
}

@article{chenEfficientSelfConsistentLearning2026,
  title = {Efficient {{Self-Consistent Learning}} of {{Gate Set Pauli Noise}}},
  year = 2026,
  month = jan,
  journal = {PRX Quantum},
  volume = {7},
  number = {1},
  pages = {010305},
  publisher = {American Physical Society},
  doi = {10.1103/1pnv-t9px},
  url = {https://link.aps.org/doi/10.1103/1pnv-t9px},
  author = {Chen, Senrui and Zhang, Zhihan and Jiang, Liang and others}
}

@article{zhangGeneralizedCycleBenchmarking2025,
  title = {Generalized {{Cycle Benchmarking Algorithm}} for {{Characterizing Midcircuit Measurements}}},
  year = 2025,
  month = jan,
  journal = {PRX Quantum},
  volume = {6},
  number = {1},
  pages = {010310},
  publisher = {American Physical Society},
  doi = {10.1103/PRXQuantum.6.010310},
  url = {https://link.aps.org/doi/10.1103/PRXQuantum.6.010310},
  author = {Zhang, Zhihan and Chen, Senrui and Liu, Yunchao and others}
}

@article{hinesPauliNoiseLearning2025,
  title = {Pauli {{Noise Learning}} for {{Mid-Circuit Measurements}}},
  author = {Hines, Jordan and Proctor, Timothy},
  year = 2025,
  month = jan,
  journal = {Physical Review Letters},
  volume = {134},
  number = {2},
  pages = {020602},
  publisher = {American Physical Society},
  doi = {10.1103/PhysRevLett.134.020602},
  url = {https://link.aps.org/doi/10.1103/PhysRevLett.134.020602}
}

@article{hothemMeasuringErrorRates2025,
  title = {Measuring Error Rates of Mid-Circuit Measurements},
  year = 2025,
  month = jul,
  journal = {Nature Communications},
  volume = {16},
  number = {1},
  pages = {5761},
  publisher = {Nature Publishing Group},
  issn = {2041-1723},
  doi = {10.1038/s41467-025-60923-x},
  url = {https://www.nature.com/articles/s41467-025-60923-x},
  copyright = {2025 The Author(s)},
  author = {Hothem, Daniel and Hines, Jordan and Baldwin, Charles and others}
}

@article{goviaRandomizedBenchmarkingSuite2023,
  title = {A Randomized Benchmarking Suite for Mid-Circuit Measurements},
  year = 2023,
  month = dec,
  journal = {New Journal of Physics},
  volume = {25},
  number = {12},
  eprint = {2207.04836},
  primaryclass = {quant-ph},
  pages = {123016},
  issn = {1367-2630},
  doi = {10.1088/1367-2630/ad0e19},
  url = {http://arxiv.org/abs/2207.04836},
  archiveprefix = {arXiv},
  author = {Govia, L. C. G. and Jurcevic, P. and Wood, C. J. and others}
}

@article{shirizlyRandomizedBenchmarkingProtocol2025,
  title = {Randomized Benchmarking Protocol for Dynamic Circuits},
  author = {Shirizly, Liran and Govia, Luke C. G. and McKay, David C.},
  year = 2025,
  month = jan,
  journal = {Physical Review A},
  volume = {111},
  number = {1},
  pages = {012611},
  publisher = {American Physical Society},
  doi = {10.1103/PhysRevA.111.012611},
  url = {https://link.aps.org/doi/10.1103/PhysRevA.111.012611}
}

@misc{bealeRandomizedCompilingSubsystem2023,
  title = {Randomized Compiling for Subsystem Measurements},
  author = {Beale, Stefanie J. and Wallman, Joel J.},
  year = 2023,
  month = apr,
  number = {arXiv:2304.06599},
  eprint = {2304.06599},
  primaryclass = {quant-ph},
  publisher = {arXiv},
  doi = {10.48550/arXiv.2304.06599},
  url = {http://arxiv.org/abs/2304.06599},
  archiveprefix = {arXiv}
}

@article{johnsonHeraldedStatePreparation2012,
  title = {Heralded State Preparation in a Superconducting Qubit},
  year = 2012,
  month = aug,
  journal = {Physical Review Letters},
  volume = {109},
  number = {5},
  pages = {050506},
  publisher = {American Physical Society},
  doi = {10.1103/PhysRevLett.109.050506},
  url = {https://link.aps.org/doi/10.1103/PhysRevLett.109.050506},
  author = {Johnson, J. E. and Macklin, C. and Slichter, D. H. and others}
}

@article{boykinAlgorithmicCoolingScalable2002,
  title = {Algorithmic {{Cooling}} and {{Scalable NMR Quantum Computers}}},
  year = 2002,
  month = mar,
  journal = {Proceedings of the National Academy of Sciences},
  volume = {99},
  number = {6},
  eprint = {quant-ph/0106093},
  pages = {3388--3393},
  issn = {0027-8424, 1091-6490},
  doi = {10.1073/pnas.241641898},
  url = {http://arxiv.org/abs/quant-ph/0106093},
  archiveprefix = {arXiv},
  author = {Boykin, P. Oscar and Mor, Tal and Roychowdhury, Vwani and others}
}

@article{aamirThermallyDrivenQuantum2025,
  title = {Thermally Driven Quantum Refrigerator Autonomously Resets a Superconducting Qubit},
  year = 2025,
  month = feb,
  journal = {Nature Physics},
  volume = {21},
  number = {2},
  pages = {318--323},
  publisher = {Nature Publishing Group},
  issn = {1745-2481},
  doi = {10.1038/s41567-024-02708-5},
  url = {https://www.nature.com/articles/s41567-024-02708-5},
  copyright = {2025 The Author(s)},
  author = {Aamir, Mohammed Ali and Jamet Suria, Paul and Mar\'in Guzm\'an, Jos\'e Antonio and others}
}

@article{shtankoComplexityLocalQuantum2025,
  title = {Complexity of Local Quantum Circuits under Nonunital Noise},
  author = {Shtanko, Oles and Sharma, Kunal},
  year = 2025,
  month = sep,
  journal = {PRX Quantum},
  volume = {6},
  number = {3},
  pages = {030347},
  publisher = {American Physical Society},
  doi = {10.1103/xwz2-wfr4},
  url = {https://link.aps.org/doi/10.1103/xwz2-wfr4}
}

@article{feffermanEffectNonunitalNoise2024,
  title = {Effect of {{Nonunital Noise}} on {{Random-Circuit Sampling}}},
  year = 2024,
  month = jul,
  journal = {PRX Quantum},
  volume = {5},
  number = {3},
  pages = {030317},
  publisher = {American Physical Society},
  doi = {10.1103/PRXQuantum.5.030317},
  url = {https://link.aps.org/doi/10.1103/PRXQuantum.5.030317},
  author = {Fefferman, Bill and Ghosh, Soumik and Gullans, Michael and others}
}

@misc{chenEnhancingQuantumNoise2025,
  title = {Enhancing Quantum Noise Characterization via Extra Energy Levels},
  year = 2025,
  month = jul,
  number = {arXiv:2506.09131},
  eprint = {2506.09131},
  primaryclass = {quant-ph},
  publisher = {arXiv},
  doi = {10.48550/arXiv.2506.09131},
  url = {http://arxiv.org/abs/2506.09131},
  archiveprefix = {arXiv},
  author = {Chen, Senrui and Hashim, Akel and Goss, Noah and others}
}

@article{bravyiHighthresholdLowoverheadFaulttolerant2024,
  title = {High-Threshold and Low-Overhead Fault-Tolerant Quantum Memory},
  year = 2024,
  month = mar,
  journal = {Nature},
  volume = {627},
  number = {8005},
  pages = {778--782},
  publisher = {Nature Publishing Group},
  issn = {1476-4687},
  doi = {10.1038/s41586-024-07107-7},
  url = {https://www.nature.com/articles/s41586-024-07107-7},
  copyright = {2024 The Author(s)},
  author = {Bravyi, Sergey and Cross, Andrew W. and Gambetta, Jay M. and others}
}

@article{pereiraCompletePhysicalCharacterization2022,
  title = {Complete Physical Characterization of Quantum Nondemolition Measurements via Tomography},
  author = {Pereira, L. and {Garc{\'i}a-Ripoll}, J. J. and Ramos, T.},
  year = 2022,
  month = jun,
  journal = {Physical Review Letters},
  volume = {129},
  number = {1},
  pages = {010402},
  publisher = {American Physical Society},
  doi = {10.1103/PhysRevLett.129.010402},
  url = {https://link.aps.org/doi/10.1103/PhysRevLett.129.010402}
}

\clearpage
\onecolumngrid
\typeout{ARXIV-SUPPLEMENT-START-PAGE=\thepage}
\begingroup
\makeatletter
\let\addcontentsline\@gobblethree
\section*{Supplemental Material}
\label{sec:supplement}
\makeatother
\endgroup

\newcommand{\prlSupplementContentsFont}{\normalsize\linespread{1.05}\selectfont}
\newcommand{\prlSupplementTocSectionSkip}{0.86em}
\newcommand{\prlSupplementTocSubsectionSkip}{0.18em}
\newcommand{\prlSupplementTocSubsubsectionSkip}{0.12em}
\newcommand{\prlSupplementContentsHeading}{%
    \vspace{-0.2em}
    \begin{center}
    \textbf{CONTENTS}
    \end{center}
    \vspace{-0.1em}
}
\newcommand{\prlSupplementAfterContents}{\clearpage}

\raggedbottom

\setcounter{section}{0}
\setcounter{subsection}{0}
\setcounter{subsubsection}{0}
\renewcommand{\thesection}{S\arabic{section}}
\renewcommand{\thesubsection}{\thesection.\arabic{subsection}}
\renewcommand{\thesubsubsection}{\thesubsection.\alph{subsubsection}}
\setcounter{secnumdepth}{3}

\makeatletter
\@removefromreset{equation}{section}
\@removefromreset{theorem}{section}
\renewcommand{\p@section}{}
\renewcommand{\p@subsection}{}
\renewcommand{\p@subsubsection}{}
\makeatother
\renewcommand{\theequation}{S\arabic{equation}}
\renewcommand{\thetheorem}{S\arabic{theorem}}
\renewcommand{\theHequation}{supp.equation.\arabic{equation}}
\renewcommand{\theHtheorem}{supp.theorem.\arabic{theorem}}
\renewcommand{\theHfigure}{supp.figure.\arabic{figure}}
\renewcommand{\theHtable}{supp.table.\arabic{table}}
\renewcommand{\theHsection}{supp.section.\arabic{section}}
\renewcommand{\theHsubsection}{supp.subsection.\arabic{section}.\arabic{subsection}}
\renewcommand{\theHsubsubsection}{supp.subsubsection.\arabic{section}.\arabic{subsection}.\arabic{subsubsection}}
\setcounter{equation}{0}
\setcounter{theorem}{0}

\renewcommand{\thefigure}{S\arabic{figure}}
\renewcommand{\thetable}{S\arabic{table}}
\setcounter{figure}{0}
\setcounter{table}{0}
\newcounter{suppalgorithm}
\renewcommand{\thesuppalgorithm}{S\arabic{suppalgorithm}}
\setcounter{suppalgorithm}{1}

\makeatletter
\let\supp@origaddcontentsline\addcontentsline
\providecommand{\prlSupplementContentsHeading}{}
\providecommand{\prlSupplementContentsFont}{\normalsize}
\providecommand{\prlSupplementAfterContents}{}
\providecommand{\prlSupplementTocSectionSkip}{0.62em}
\providecommand{\prlSupplementTocSubsectionSkip}{0.04em}
\providecommand{\prlSupplementTocSubsubsectionSkip}{0.02em}
\def\supp@addcontentsline#1#2#3{
    \def\supp@level{#2}
    \def\supp@section{section}
    \def\supp@subsection{subsection}
    \def\supp@subsubsection{subsubsection}
    \ifx\supp@level\supp@section
        \addtocontents{supp}{\protect\contentsline{section}{#3}{\thepage}{\@currentHref}}
    \else\ifx\supp@level\supp@subsection
        \addtocontents{supp}{\protect\contentsline{subsection}{#3}{\thepage}{\@currentHref}}
    \else\ifx\supp@level\supp@subsubsection
        \addtocontents{supp}{\protect\contentsline{subsubsection}{#3}{\thepage}{\@currentHref}}
    \fi\fi
    \fi
}
\let\addcontentsline\supp@addcontentsline
\makeatother

\begingroup
\makeatletter
\let\addcontentsline\@gobblethree
\setcounter{tocdepth}{3}
\prlSupplementContentsHeading
\begingroup
\prlSupplementContentsFont
\parskip\z@skip
\def\@pnumwidth{2.4em}
\def\@tocrmarg{3em}
\def\supp@tocline#1#2#3#4#5#6{%
    \ifnum #1>\c@tocdepth
    \else
        \vskip #6\relax
        {%
        \leftskip #2\relax
        \rightskip \@tocrmarg
        \parfillskip -\rightskip
        \parindent #2\relax
        \@afterindenttrue
        \interlinepenalty\@M
        \leavevmode
        \@tempdima #3\relax
        \advance\leftskip \@tempdima
        \null\nobreak\hskip -\@tempdima
        \def\numberline##1{\hb@xt@\@tempdima{\if\relax\detokenize{##1}\relax\else ##1.\fi\hfil}}%
        #4\nobreak\hfill\nobreak\hb@xt@\@pnumwidth{\hss #5}\par
        }%
    \fi
}
\def\l@section#1#2{%
    \supp@tocline{1}{0pt}{3.0em}{#1}{#2}{\prlSupplementTocSectionSkip}%
}
\def\l@subsection#1#2{%
    \supp@tocline{2}{1.8em}{4.1em}{#1}{#2}{\prlSupplementTocSubsectionSkip}%
}
\def\l@subsubsection#1#2{%
    \supp@tocline{3}{3.6em}{4.8em}{#1}{#2}{\prlSupplementTocSubsubsectionSkip}%
}
\@starttoc{supp}%
\endgroup
\makeatother
\endgroup
\prlSupplementAfterContents
\vspace{0.4em}

\section{Model, conventions, and forward probabilities}
\label{app:model}
\label{app:prelim}
\label{app:notation-table}
\begingroup
\small
\setlength{\tabcolsep}{4pt}
\begin{center}
\refstepcounter{table}\label{tab:notation-glossary}
Table~\thetable: Notation table and glossary.

\vspace{0.4em}
\begin{tabular}{@{}llll@{}}
    \toprule
    Symbol                                    & Type          & Meaning & First use \\
    \midrule

    $\Lambda=\{\Lambda_0,\Lambda_1\}$         & instrument    &
    \parbox[t]{0.46\columnwidth}{\raggedright
        Two-outcome single-qubit quantum instrument; each $\Lambda_o$ is CP and $\Lambda_0+\Lambda_1$ is TP.
    }                                         &
    \parbox[t]{0.16\columnwidth}{\raggedright
        Def.~\ref{def:z-twirled-mcm-instrument}
    }                                                                               \\

    $\mathcal{T}_Z$                           & map           &
    \parbox[t]{0.46\columnwidth}{\raggedright
        $Z$-twirl of an instrument (applied outcome-wise), Eq.~\eqref{eq:z-twirl-instrument}.
    }                                         &
    \parbox[t]{0.16\columnwidth}{\raggedright
        Def.~\ref{def:z-twirled-mcm-instrument}
    }                                                                               \\

    $p_s^{(o,s')}$                            & scalar        &
    \parbox[t]{0.46\columnwidth}{\raggedright
    Joint probability $p_s^{(o,s')}=\Pr(o,\ s'\mid s)$ in the $Z$-twirled MCM model.
    }                                         &
    \parbox[t]{0.16\columnwidth}{\raggedright
        Def.~\ref{def:z-twirled-mcm-instrument}
    }                                                                               \\

    $M^{(o)}$                                 & matrix        &
    \parbox[t]{0.46\columnwidth}{\raggedright
    $2\times 2$ nonnegative matrix in the MCM matrix pair, with entries $M^{(o)}_{s',s}=p_s^{(o,s')}$; columns satisfy
    $\sum_{o,s'} M^{(o)}_{s',s}=1$.
    }                                         &
    \parbox[t]{0.16\columnwidth}{\raggedright
        Def.~\ref{def:z-twirled-mcm-instrument}
    }                                                                               \\

    $A_{o\mid s}$                             & scalar        &
    \parbox[t]{0.46\columnwidth}{\raggedright
    Readout-confusion probability $\Pr(o\mid s)=A_{o\mid s}=\sum_{s'} M^{(o)}_{s',s}$.
    }                                         &
    \parbox[t]{0.16\columnwidth}{\raggedright
        Def.~\ref{def:readout-backaction-error-rates}
    }                                                                               \\

    $B_{s'\mid s}$                            & scalar        &
    \parbox[t]{0.46\columnwidth}{\raggedright
    Unconditional backaction (population-transfer) probability
    $\Pr(s'\mid s)=B_{s'\mid s}=\sum_o M^{(o)}_{s',s}$.
    }                                         &
    \parbox[t]{0.16\columnwidth}{\raggedright
        Def.~\ref{def:readout-backaction-error-rates}
    }                                                                               \\

    $\epsilon_{1\mid 0},\,\epsilon_{0\mid 1}$ & scalars       &
    \parbox[t]{0.46\columnwidth}{\raggedright
    Asymmetric readout error rates $\epsilon_{1\mid 0}=A_{1\mid 0}$ and $\epsilon_{0\mid 1}=A_{0\mid 1}$.
    }                                         &
    \parbox[t]{0.16\columnwidth}{\raggedright
        Def.~\ref{def:readout-backaction-error-rates}
    }                                                                               \\

    $\eta_{\uparrow},\,\eta_{\downarrow}$     & scalars       &
    \parbox[t]{0.46\columnwidth}{\raggedright
    Backaction excitation/decay probabilities $\eta_{\uparrow}=B_{1\mid 0}$ and $\eta_{\downarrow}=B_{0\mid 1}$.
    }                                         &
    \parbox[t]{0.16\columnwidth}{\raggedright
        Def.~\ref{def:readout-backaction-error-rates}
    }                                                                               \\

    \shortstack[l]{$w,\ p(w),$\\$M^{(w)}$}   & mixed &
    \parbox[t]{0.46\columnwidth}{\raggedright
    Readout bit string $w=w_1\cdots w_L\in\{0,1\}^L$, its probability $p(w)$, and the ordered product $M^{(w)}=M^{(w_L)}\cdots M^{(w_1)}$ used in the forward model.
    }                                         &
    \parbox[t]{0.16\columnwidth}{\raggedright
        SM Sec.~\ref{app:prob-dict}
    }                                                                               \\

    $N,\ N(w),\ \widehat{p}(w)$               & scalars       &
    \parbox[t]{0.46\columnwidth}{\raggedright
        Total shots $N$, count $N(w)$, and empirical frequency $\widehat{p}(w)=N(w)/N$.
    }                                         &
    \parbox[t]{0.16\columnwidth}{\raggedright
        SM Sec.~\ref{app:prob-dict}
    }                                                                               \\

    $\bm{\pi}$                                & vector        &
    \parbox[t]{0.46\columnwidth}{\raggedright
        Initial computational-basis population vector $\bm{\pi}=(\pi_0,\pi_1)^{\mathsf{T}}$ with $\pi_0+\pi_1=1$.
    }                                         &
    \parbox[t]{0.16\columnwidth}{\raggedright
        Lem.~\ref{lem:matrix-product-string-probabilities}
    }                                                                               \\

    \shortstack[l]{$0^n,\,1^n,$\\$S_n,\,T_n$} & mixed &
    \parbox[t]{0.46\columnwidth}{\raggedright
        Length-$n$ all-zero/all-one strings ($0^0=1^0=\emptyset$) and the associated all-zero/all-one readout-string probabilities $S_n=p(0^n)$ and $T_n=p(1^n)$.
    }                                         &
    \parbox[t]{0.16\columnwidth}{\raggedright
        Eq.~\eqref{eq:SnTn-def}
    }                                                                               \\

    $R(t),\,t$                                & matrix/scalar &
    \parbox[t]{0.46\columnwidth}{\raggedright
    Gauge matrix $R(t)=\begin{psmallmatrix}1-t&t\\ t&1-t\end{psmallmatrix}$ (invertible iff $t\neq\tfrac12$) generating the similarity action $M^{(o)}\mapsto R(t)^{-1}M^{(o)}R(t)$.
    }                                         &
    \parbox[t]{0.16\columnwidth}{\raggedright
        SM Sec.~\ref{app:gauge}
    }                                                                               \\

    $\mathcal{T}_\mathrm{phys}$               & set           &
    \parbox[t]{0.46\columnwidth}{\raggedright
        Physically allowed gauge set consisting of all $t\neq\tfrac12$ such that the gauged representative remains entrywise within $[0,1]$.
    }                                         &
    \parbox[t]{0.16\columnwidth}{\raggedright
        SM Sec.~\ref{app:physicality}
    }                                                                               \\

    $t_{\mathrm{rep}}$                        & scalar        &
    \parbox[t]{0.46\columnwidth}{\raggedright
        Representative gauge choice selected from the $t<\tfrac12$ branch of $\mathcal{T}_\mathrm{phys}$ for reporting and plotting.
    }                                         &
    \parbox[t]{0.16\columnwidth}{\raggedright
        Def.~\ref{def:max-margin-gauge-selection}
    }                                                                               \\

    $[q]_{\mathcal{T}_\mathrm{phys}^{(+)}}$   & interval      &
    \parbox[t]{0.46\columnwidth}{\raggedright
        Gauge band of a derived quantity $q$ over the $t<\tfrac12$ reporting branch, Eq.~\eqref{eq:q-band-def}.
    }                                         &
    \parbox[t]{0.16\columnwidth}{\raggedright
        SM Sec.~\ref{app:physicality}
    }                                                                               \\
\bottomrule
\end{tabular}
\end{center}
\endgroup

\subsection{Notation, readout-string distributions, and randomization/twirling conventions}
\label{app:prob-dict}

We represent a recorded sequence of mid-circuit measurement (MCM) readouts as a \emph{bit string}
\begin{equation}
    \label{eq:string-def}
    w \;=\; w_1 w_2 \cdots w_L \in \{0,1\}^L,
\end{equation}
where $L$ denotes the number of MCM readouts in the sequence.
Our convention is \emph{chronological left-to-right ordering}.
$w_1$ denotes the earliest readout in time (the first MCM applied in the circuit), and $w_L$ denotes the latest.
For $1\le \ell\le L$, we write the first $\ell$ readouts as
\begin{equation}
    \label{eq:prefix-def}
    w_{1:\ell} \;\coloneqq\; w_1 w_2 \cdots w_\ell.
\end{equation}
For concatenation, if $u\in\{0,1\}^\ell$ and $o\in\{0,1\}$, we write $u o \in \{0,1\}^{\ell+1}$ for appending $o$ to the \emph{right} end of $u$.

A \emph{readout-string distribution} is the collection of readout-string probabilities $p(w)$ for strings $w$ of one or more lengths, where for each fixed length $L$ the restriction $p(\cdot)\big|_{\{0,1\}^L}$ is a normalized distribution.
\begin{equation}
    \label{eq:string-normalization}
    \sum_{w\in\{0,1\}^L} p(w) \;=\; 1.
\end{equation}
In experiments, $p(w)$ is estimated from $N$ repeated shots of the same length-$L$ circuit.
If $N(w)$ is the observed count of the string $w$, then the empirical frequency is
\begin{equation}
    \label{eq:empirical-frequency}
    \widehat{p}(w) \;\coloneqq\; \frac{N(w)}{N},
    \qquad \sum_{w\in\{0,1\}^L} \widehat{p}(w)=1.
\end{equation}

When a single dataset contains readout-string probabilities for multiple lengths (e.g.\ $L=1,2,\dots$), we will implicitly interpret $p(w)$ and $\widehat{p}(w)$ as referring to the experiment with length $L$.
When probabilities for a shorter initial readout string are required, they are defined by summing over later readouts within a fixed length-$L$ distribution.
Concretely, for any initial readout string $u\in\{0,1\}^\ell$ and any $L\ge \ell$,
\begin{equation}
    \label{eq:prefix-marginal}
    p(u)
    \;\coloneqq\;
    \sum_{v\in\{0,1\}^{L-\ell}} p(u v),
    \qquad
    \widehat{p}(u)
    \;\coloneqq\;
    \sum_{v\in\{0,1\}^{L-\ell}} \widehat{p}(u v),
\end{equation}
where $u v$ denotes concatenation.

For the learned $Z$-twirled model with matrices $M^{(0)},M^{(1)}$ (Definition~\ref{def:z-twirled-mcm-instrument}), define the string-ordered product
\begin{equation}
    \label{eq:mw-def}
    M^{(w)} \;\coloneqq\; M^{(w_L)} M^{(w_{L-1})}\cdots M^{(w_1)} \in \mathbb{R}^{2\times 2}.
\end{equation}
The reversal relative to the left-to-right string order is purely notational.
The bit $w_1$ acts first on a population column vector, and matrix products compose right-to-left.

A two-outcome instrument is written $\Lambda=\{\Lambda_0,\Lambda_1\}$, where each $\Lambda_o$ is completely positive (CP) and $\Lambda_0+\Lambda_1$ is trace preserving (TP).
In our $Z$-twirled setting, the induced dynamics on computational-basis populations is classical.
We use \emph{twirling} for channel or instrument averaging by conjugation; when a random unitary is applied before the MCM chain and the resulting data are averaged, we call this \emph{input randomization}.

\paragraph{Input randomization and channel twirling.}\label{app:randomization:input-vs-channel}

Uniform input randomization maps an input state or operator to $\mathbb{E}_\omega[U_\omega A U_\omega^\dagger]$ for sampled unitaries $U_\omega$.
A channel twirl instead averages a noisy channel by conjugation,
$\mathcal{E}\mapsto \mathbb{E}_\omega[\mathcal{U}_\omega^\dagger\circ \mathcal{E}\circ \mathcal{U}_\omega]$,
where $\mathcal{U}_\omega(\cdot)=U_\omega(\cdot)U_\omega^\dagger$.
These are conceptually distinct operations; in this manuscript, the learned MCM model uses a $Z$-twirled instrument, while the learning circuits use single-qubit Pauli input randomization to prepare a maximally mixed effective input.

\begin{definition}[Complete depolarizing channel from Pauli input randomization]
    \label{def:pauli-input-randomization-depolarizing-map}
    Uniform randomization over the single-qubit Paulis $\{I,X,Y,Z\}$ implements the complete depolarizing channel
    \begin{equation}
        \label{eq:pauli-randomization-depolarizing-def}
        \Delta_{\mathrm{dep}}(A)\;\coloneqq\;\frac{1}{4}\sum_{P\in\{I,X,Y,Z\}} P\,A\,P,
    \end{equation}
    defined for any single-qubit operator $A$.
\end{definition}

\begin{lemma}[Pauli input randomization outputs the maximally mixed state]
    \label{lem:pauli-randomization-maximally-mixed}
    For any single-qubit operator $A$,
    \begin{equation}
        \label{eq:pauli-randomization-maxmix}
        \Delta_{\mathrm{dep}}(A)\;=\;\frac{\Tr(A)}{2}\,I.
    \end{equation}
    In particular, for any single-qubit density operator $\rho$, $\Delta_{\mathrm{dep}}(\rho)=I/2$.
\end{lemma}

\begin{proof}
    Expand $A$ in the Pauli basis as $A=a_0 I+a_x X+a_y Y+a_z Z$ with $a_0=\Tr(A)/2$.
    For any fixed traceless Pauli $P_j\in\{X,Y,Z\}$, conjugation by $P\in\{I,X,Y,Z\}$ yields $P P_j P=\pm P_j$, and exactly two of the four single-qubit Paulis commute with $P_j$ while the other two anticommute, so the average $\frac14\sum_{P\in\{I,X,Y,Z\}} P P_j P$ vanishes.
    The identity component is invariant under conjugation, hence the randomization keeps only $a_0 I$, giving \eqref{eq:pauli-randomization-maxmix}.
\end{proof}

\paragraph{Pauli-axis dephasing map for Bloch-component isolation.}\label{app:randomization:axis-dephasing}

A recurring operation in SM Sec.~\ref{app:spam} is a \emph{Pauli-axis dephasing map} that removes two Bloch components of an input state, leaving only the component aligned with a chosen Pauli axis.
It is implemented by averaging over a random choice of applying either $I$ or that Pauli.

\begin{definition}[Pauli-axis dephasing map]
    \label{def:pauli-axis-dephasing-map}
    Fix $a\in\{x,y,z\}$ and write $(P_x,P_y,P_z)=(X,Y,Z)$.
    Define the CPTP map
    \begin{equation}
        \label{eq:axis-dephasing-def}
        \mathcal{D}_a(A)\;\coloneqq\;\frac{1}{2}\bigl(A+P_a A P_a\bigr),
    \end{equation}
    for any single-qubit operator $A$.
\end{definition}

\begin{lemma}[Pauli-axis dephasing keeps exactly one Bloch component]
    \label{lem:axis-dephasing-keeps-one-component}
    Write any single-qubit operator as $A=a_0 I+a_x X+a_y Y+a_z Z$ with $a_0=\Tr(A)/2$.
    Then
    \begin{equation}
        \label{eq:axis-dephasing-projection}
        \mathcal{D}_x(A)=a_0 I+a_x X,\qquad
        \mathcal{D}_y(A)=a_0 I+a_y Y,\qquad
        \mathcal{D}_z(A)=a_0 I+a_z Z.
    \end{equation}
    In particular, for $\rho=\frac12(I+r_x X+r_y Y+r_z Z)$,
    \begin{equation}
        \label{eq:axis-dephasing-bloch}
        \mathcal{D}_a(\rho)=\frac12\bigl(I+r_a P_a\bigr),\qquad a\in\{x,y,z\}.
    \end{equation}
\end{lemma}

\begin{proof}
    For $b\in\{x,y,z\}$, Pauli conjugation acts as $P_a P_b P_a = (-1)^{\mathds{1}[a\neq b]} P_b$ and fixes $I$.
    Applying \eqref{eq:axis-dephasing-def} to the Pauli expansion of $A$ therefore cancels the two components orthogonal to $a$ and keeps only the $I$ and $P_a$ components, yielding \eqref{eq:axis-dephasing-projection} and \eqref{eq:axis-dephasing-bloch}.
\end{proof}

\noindent The map $\mathcal{D}_a$ is idempotent, $\mathcal{D}_a\circ\mathcal{D}_a=\mathcal{D}_a$, with image $\mathrm{span}\{I,P_a\}$.

\subsection{Z-twirled MCM instrument and readout/backaction metrics}
\begin{definition}[Z-twirled single-qubit MCM instrument]
    \label{def:z-twirled-mcm-instrument}
    Let $\mathcal{H}\cong \mathbb{C}^2$ with computational basis $\{\ket{0},\ket{1}\}$ and Pauli operator $Z=\ket{0}\!\bra{0}-\ket{1}\!\bra{1}$.
    A \emph{two-outcome single-qubit instrument} is a pair $\Lambda=\{\Lambda_0,\Lambda_1\}$ of linear maps acting on single-qubit operators such that each $\Lambda_o$ is completely positive (CP) and the sum $\Lambda_0+\Lambda_1$ is trace preserving (TP).

    The \emph{$Z$-twirl} of $\Lambda$ is the instrument $\mathcal{T}_Z(\Lambda)$ defined outcome-wise by
    \begin{equation}
        \label{eq:z-twirl-instrument}
        (\mathcal{T}_Z(\Lambda))_o(\rho)
        \;\coloneqq\;
        \frac{1}{2}\sum_{b\in\{0,1\}} Z^b\,\Lambda_o\!\bigl(Z^b \rho Z^b\bigr)\,Z^b,
        \qquad o\in\{0,1\},
    \end{equation}
    where $Z^0\equiv I$ and $Z^1\equiv Z$.

    The associated \emph{MCM matrix pair} $M^{(0)},M^{(1)}\in\mathbb{R}^{2\times 2}$ is defined by
    \begin{equation}
        \label{eq:mcm-matrix-def}
        M^{(o)}_{s',s}
        \;\coloneqq\;
        \bra{s'}\,(\mathcal{T}_Z(\Lambda))_o\!\bigl(\ket{s}\!\bra{s}\bigr)\,\ket{s'},
        \qquad s,s',o\in\{0,1\}.
    \end{equation}
    Equivalently, the entry $M^{(o)}_{s',s}$ is the joint probability
    \begin{equation}
        \label{eq:joint-prob-def}
        M^{(o)}_{s',s} \;=\; p_s^{(o,s')}
        \;\coloneqq\;
        \Pr(\text{readout }o,\ \text{post-measurement state }s' \mid \text{pre-measurement state }s),
    \end{equation}
    in the $Z$-twirled MCM model.

    A \emph{$Z$-twirled single-qubit MCM instrument model} is any pair $(M^{(0)},M^{(1)})$ satisfying
    \begin{equation}
        \label{eq:z-twirled-mcm-constraints}
        M^{(o)}_{s',s}\ge 0,
        \qquad
        \sum_{o\in\{0,1\}}\sum_{s'\in\{0,1\}} M^{(o)}_{s',s}=1
        \quad\text{for each } s\in\{0,1\}.
    \end{equation}
    The constraint \eqref{eq:z-twirled-mcm-constraints} is necessary for physical realizability, and Proposition~\ref{prop:mcm-equivalence-learning} below shows it is also sufficient (via an explicit canonical Kraus realization).
\end{definition}

\begin{proposition}[Equivalence in learning: $(M^{(0)},M^{(1)})$ vs $Z$-twirled POVM + canonical Kraus operators]
    \label{prop:mcm-equivalence-learning}
    Let $\Lambda=\{\Lambda_0,\Lambda_1\}$ be any two-outcome single-qubit instrument on $\mathcal{H}$, and let $M^{(0)},M^{(1)}$ be defined from its $Z$-twirl as in Definition~\ref{def:z-twirled-mcm-instrument}.
    Then the matrices satisfy the physical constraints \eqref{eq:z-twirled-mcm-constraints}.

    Conversely, given any pair of matrices $(M^{(0)},M^{(1)})$ satisfying \eqref{eq:z-twirled-mcm-constraints}, define Kraus operators
    \begin{equation}
        \label{eq:canonical-kraus}
        K_{o,s',s}\;\coloneqq\;\sqrt{M^{(o)}_{s',s}}\;\ket{s'}\!\bra{s},
        \qquad s,s',o\in\{0,1\},
    \end{equation}
    and completely positive maps (one per outcome)
    \begin{equation}
        \label{eq:canonical-instrument}
        \Lambda^{\mathrm{can}}_o(\rho)
        \;\coloneqq\;
        \sum_{s,s'\in\{0,1\}} K_{o,s',s}\,\rho\,K_{o,s',s}^\dagger,
        \qquad o\in\{0,1\}.
    \end{equation}
    Then $\Lambda^{\mathrm{can}}=\{\Lambda^{\mathrm{can}}_0,\Lambda^{\mathrm{can}}_1\}$ is a valid two-outcome instrument, it is $Z$-incoherent (it discards coherences and maps populations classically), and it reproduces the prescribed matrices on computational-basis inputs.
    \begin{equation}
        \label{eq:canonical-reproduces-m}
        \Lambda^{\mathrm{can}}_o\!\bigl(\ket{s}\!\bra{s}\bigr)
        \;=\;
        \sum_{s'\in\{0,1\}} M^{(o)}_{s',s}\,\ket{s'}\!\bra{s'}.
    \end{equation}
    The associated (diagonal) two-outcome POVM effects are
    \begin{equation}
        \label{eq:z-twirled-povm-from-m}
        F_o^{(Z)}
        \;\coloneqq\;
        \sum_{s,s'\in\{0,1\}} K_{o,s',s}^\dagger K_{o,s',s}
        \;=\;
        \sum_{s\in\{0,1\}}\Bigl(\sum_{s'\in\{0,1\}} M^{(o)}_{s',s}\Bigr)\ket{s}\!\bra{s},
        \qquad o\in\{0,1\}.
    \end{equation}
    In particular, the data $(M^{(0)},M^{(1)})$ determines (and is determined by) the diagonal $Z$-twirled POVM $\{F_0^{(Z)},F_1^{(Z)}\}$ together with the canonical $Z$-incoherent Kraus operators $\{K_{o,s',s}\}_{o,s',s\in\{0,1\}}$ in \eqref{eq:canonical-kraus}.
\end{proposition}

\begin{proof}
    \emph{Step 1.
    From an instrument to the constraints on $M^{(0)},M^{(1)}$.}
    Fix $o\in\{0,1\}$.
    The map $(\mathcal{T}_Z(\Lambda))_o$ in \eqref{eq:z-twirl-instrument} is an average of CP maps, hence CP.
    For a basis projector $\ket{s}\!\bra{s}$ we have $Z^b \ket{s}\!\bra{s} Z^b=\ket{s}\!\bra{s}$, so
    \begin{equation}
        \label{eq:z-twirl-on-basis-projector}
        (\mathcal{T}_Z(\Lambda))_o\!\bigl(\ket{s}\!\bra{s}\bigr)
        \;=\;
        \frac{1}{2}\Bigl(\Lambda_o(\ket{s}\!\bra{s}) + Z\,\Lambda_o(\ket{s}\!\bra{s})\,Z\Bigr).
    \end{equation}
    The right-hand side commutes with $Z$, hence is diagonal in the computational basis.
    Because it is also positive semidefinite, its diagonal entries are nonnegative.
    \[
        M^{(o)}_{s',s}
        =\bra{s'}(\mathcal{T}_Z(\Lambda))_o(\ket{s}\!\bra{s})\ket{s'}
        \ge 0.
    \]

    Next, we show the column-stochastic normalization.
    Since $\Lambda$ is an instrument, the sum $\Lambda_0+\Lambda_1$ is TP; conjugation by $Z^b$ is TP; hence averaging preserves TP.
    Therefore $\sum_o (\mathcal{T}_Z(\Lambda))_o$ is TP.
    For each $s$,
    \begin{align}
        1
        \;=\;
        \Tr (\ket{s}\!\bra{s})
        \;=\;
        \Tr \!\left(\sum_{o\in\{0,1\}} (\mathcal{T}_Z(\Lambda))_o(\ket{s}\!\bra{s})\right)
        \;=\;
        \sum_{o\in\{0,1\}} \Tr \!\left((\mathcal{T}_Z(\Lambda))_o(\ket{s}\!\bra{s})\right)
        \;=\;
        \sum_{o\in\{0,1\}}\sum_{s'\in\{0,1\}} M^{(o)}_{s',s},
    \end{align}
    where the last equality uses diagonality of \eqref{eq:z-twirl-on-basis-projector}.
    This proves \eqref{eq:z-twirled-mcm-constraints}.

    \emph{Step 2.
    From $(M^{(0)},M^{(1)})$ to a canonical $Z$-incoherent instrument.}
    Assume $(M^{(0)},M^{(1)})$ satisfies \eqref{eq:z-twirled-mcm-constraints}.
    Define $K_{o,s',s}$ as in \eqref{eq:canonical-kraus} and $\Lambda_o^{\mathrm{can}}$ as in \eqref{eq:canonical-instrument}.
    Each $\Lambda_o^{\mathrm{can}}$ is CP by construction.

    To check trace preservation of the total map, compute the Kraus completeness relation.
    \begin{align}
        \sum_{o\in\{0,1\}} \sum_{s,s'\in\{0,1\}} K_{o,s',s}^\dagger K_{o,s',s}
         & =
        \sum_{o,s,s'} M^{(o)}_{s',s}\,\ket{s}\!\bra{s}
        \\
         & =
        \sum_{s\in\{0,1\}}\Bigl(\sum_{o\in\{0,1\}}\sum_{s'\in\{0,1\}} M^{(o)}_{s',s}\Bigr)\ket{s}\!\bra{s}
        \;=\;
        \sum_{s\in\{0,1\}} \ket{s}\!\bra{s}
        \;=\;
        I,
    \end{align}
    using \eqref{eq:z-twirled-mcm-constraints}.
    Hence $\sum_o \Lambda_o^{\mathrm{can}}$ is TP, so $\Lambda^{\mathrm{can}}$ is a valid instrument.

    \emph{Step 3.
    Reproduction of $M$ on basis inputs and $Z$-incoherence.}
    For a basis projector $\ket{s}\!\bra{s}$,
    \begin{align}
        \Lambda_o^{\mathrm{can}}(\ket{s}\!\bra{s})
         & =
        \sum_{t,t'\in\{0,1\}} K_{o,t',t}\,\ket{s}\!\bra{s}\,K_{o,t',t}^\dagger
        \\
         & =
        \sum_{t'\in\{0,1\}} K_{o,t',s}\,\ket{s}\!\bra{s}\,K_{o,t',s}^\dagger
        \;=\;
        \sum_{t'\in\{0,1\}} M^{(o)}_{t',s}\,\ket{t'}\!\bra{t'},
    \end{align}
    which is exactly \eqref{eq:canonical-reproduces-m} (rename $t'\mapsto s'$).

    More generally, for an arbitrary operator $\rho=\sum_{u,v}\rho_{uv}\ket{u}\!\bra{v}$,
    \begin{align}
        \Lambda_o^{\mathrm{can}}(\rho)
         & =
        \sum_{s,s'} M^{(o)}_{s',s}\,\ket{s'}\!\bra{s}\,
        \left(\sum_{u,v}\rho_{uv}\ket{u}\!\bra{v}\right)
        \ket{s}\!\bra{s'}
        \\
         & =
        \sum_{s,s'} M^{(o)}_{s',s}\,\rho_{ss}\,\ket{s'}\!\bra{s'}.
        \label{eq:canonical-action-general-rho}
    \end{align}
    Equation \eqref{eq:canonical-action-general-rho} shows that (i) $\Lambda_o^{\mathrm{can}}(\rho)$ is always diagonal in the computational basis, and (ii) it depends only on the diagonal entries of $\rho$; hence coherences never feed into populations (i.e., the instrument is $Z$-incoherent).

    \emph{Step 4.
    POVM effects.}
    The POVM effect associated with outcome $o$ is the operator $F_o^{(Z)}\coloneqq \sum_{s,s'} K_{o,s',s}^\dagger K_{o,s',s}$, yielding \eqref{eq:z-twirled-povm-from-m}.
    Summing \eqref{eq:z-twirled-povm-from-m} over $o$ and using \eqref{eq:z-twirled-mcm-constraints} gives $\sum_o F_o^{(Z)}=I$, so $\{F_0^{(Z)},F_1^{(Z)}\}$ is a valid (diagonal) POVM.

    The final “equivalence in learning” statement follows because \eqref{eq:canonical-kraus} and \eqref{eq:z-twirled-povm-from-m} are explicit deterministic maps from $(M^{(0)},M^{(1)})$ to $\{F_o^{(Z)}\}$ and $\{K_{o,s',s}\}$, while $M^{(o)}_{s',s}$ is recovered from \eqref{eq:canonical-kraus} via $M^{(o)}_{s',s}=\|K_{o,s',s}\ket{s}\|^2$.
\end{proof}

\paragraph{Derived error metrics.}
\begin{definition}[Readout and backaction error rates]
    \label{def:readout-backaction-error-rates}
    Consider the $Z$-twirled single-qubit MCM model of Definition~\ref{def:z-twirled-mcm-instrument}, i.e. a pair of nonnegative matrices $(M^{(0)},M^{(1)})$ with entries $M^{(o)}_{s',s}=\Pr(o,\;s'\mid s)$ and column normalization $\sum_{o,s'} M^{(o)}_{s',s}=1$ for each $s\in\{0,1\}$.

    \emph{Readout-confusion probabilities.}
    Define the classical readout confusion matrix $A\in\mathbb{R}^{2\times 2}$ by
    \begin{equation}
        \label{eq:assignment-matrix-def}
        A_{o\mid s}
        \;\coloneqq\;
        \Pr(o\mid s)
        \;=\;
        \sum_{s'\in\{0,1\}} M^{(o)}_{s',s},
        \qquad o,s\in\{0,1\}.
    \end{equation}
    The \emph{asymmetric readout error rates} are the two off-diagonal readout-confusion probabilities
    \begin{equation}
        \label{eq:readout-error-rates-def}
        \epsilon_{1\mid 0}
        \;\coloneqq\;
        A_{1\mid 0}
        \;=\;
        \sum_{s'\in\{0,1\}} M^{(1)}_{s',0},
        \qquad
        \epsilon_{0\mid 1}
        \;\coloneqq\;
        A_{0\mid 1}
        \;=\;
        \sum_{s'\in\{0,1\}} M^{(0)}_{s',1}.
    \end{equation}
    Equivalently, $\epsilon_{1\mid 0}=\Pr(\text{readout }1\mid \text{pre-state }0)$ and $\epsilon_{0\mid 1}=\Pr(\text{readout }0\mid \text{pre-state }1)$, while the correct-readout probabilities are $A_{0\mid 0}=1-\epsilon_{1\mid 0}$ and $A_{1\mid 1}=1-\epsilon_{0\mid 1}$.

    \emph{Backaction (population-transfer) probabilities.}
    Define the unconditional post-measurement population-transfer matrix $B\in\mathbb{R}^{2\times 2}$ by
    \begin{equation}
        \label{eq:backaction-matrix-def}
        B_{s'\mid s}
        \;\coloneqq\;
        \Pr(s'\mid s)
        \;=\;
        \sum_{o\in\{0,1\}} M^{(o)}_{s',s},
        \qquad s,s'\in\{0,1\}.
    \end{equation}
    The \emph{asymmetric backaction error rates} are the two off-diagonal transition probabilities
    \begin{equation}
        \label{eq:backaction-error-rates-def}
        \eta_{\uparrow}
        \;\coloneqq\;
        B_{1\mid 0}
        \;=\;
        \sum_{o\in\{0,1\}} M^{(o)}_{1,0},
        \qquad
        \eta_{\downarrow}
        \;\coloneqq\;
        B_{0\mid 1}
        \;=\;
        \sum_{o\in\{0,1\}} M^{(o)}_{0,1}.
    \end{equation}
    Here $\eta_{\uparrow}$ is the \emph{excitation} probability (population $0\!\to\!1$ induced by the MCM), and $\eta_{\downarrow}$ is the \emph{decay} probability (population $1\!\to\!0$ induced by the MCM).
\end{definition}

\begin{definition}[Basis-averaged joint success probability used for illustrative examples]
    \label{def:population-instrument-fidelity}
    In the reduced two-state model, an ideal computational-basis MCM maps a pre-measurement state $s$ to readout $o=s$ and post-measurement state $s'=s$.
    For a $Z$-twirled MCM matrix pair $(M^{(0)},M^{(1)})$, we define the basis-averaged joint success probability
    \begin{equation}
        \label{eq:population-instrument-fidelity}
        F_{\mathrm{MCM}}^{\mathrm{pop}}
        \;\coloneqq\;
        \frac{1}{2}\Bigl(M^{(0)}_{0,0}+M^{(1)}_{1,1}\Bigr).
    \end{equation}
    Equivalently, $F_{\mathrm{MCM}}^{\mathrm{pop}}$ is the average, over equally weighted inputs $\ket{0}$ and $\ket{1}$, of the probability that the MCM both reports the correct computational-basis outcome and leaves the qubit in the corresponding basis state.
    For the matrix pair used to generate the simulated curves in \prlFigOverviewRef{}(c), this probability is $70\%$.
    Its value changes as $t$ varies over the physically allowed gauge set.
    Every matrix pair in that family gives the same readout-string probabilities.
\end{definition}

\noindent The matrix $A$ in \eqref{eq:assignment-matrix-def} contains only the probabilities of reporting each readout $o$ for each input basis state $s$.
It is the usual readout confusion matrix used in measurement-error mitigation.
If $A$ is estimated by preparing nominal $\ket{0}$ and $\ket{1}$ states, preparation errors and readout errors can be mixed together.
Such an estimate cannot by itself separate the two sources of error.

The $Z$-twirled instrument retains the fuller joint probabilities $M^{(o)}_{s',s}=\Pr(o,s'\mid s)$.
Summing over the post-measurement state $s'$ gives the readout matrix $A$ in \eqref{eq:assignment-matrix-def}.
Summing over the readout $o$ gives the state-transition matrix $B$ in \eqref{eq:backaction-matrix-def}.
After $(M^{(0)},M^{(1)})$ is learned, the relevant noisy-preparation readout probabilities are combined with its readout matrix $A$ to infer preparation parameters in SM Sec.~\ref{app:spam}.
This procedure estimates the preparation error from $A$ and preserves the state-transition information in $B$.
Two instruments can have the same $A$ but different $B$ matrices and different correlations between the readout and the post-measurement state.

\subsection{Forward model for readout-string probabilities}
\label{app:model:forward}

Repeated applications of the $Z$-twirled MCM matrix pair $(M^{(0)},M^{(1)})$ (Definition~\ref{def:z-twirled-mcm-instrument}) induce the following matrix-product expression for readout-string probabilities.

\begin{lemma}[Matrix-product form for readout-string probabilities]
    \label{lem:matrix-product-string-probabilities}
    Fix any $Z$-twirled single-qubit MCM model $(M^{(0)},M^{(1)})$ satisfying \eqref{eq:z-twirled-mcm-constraints}.
    Run this same MCM model $L\ge 1$ times in sequence, and let $w=w_1\cdots w_L\in\{0,1\}^L$ be the recorded readout bit string, ordered chronologically as in SM Sec.~\ref{app:prob-dict}.
    Let $\bm{\pi}=(\pi_0,\pi_1)^{\mathsf{T}}$ be the initial population vector in the computational basis, and let $\bm{1}=(1,1)^{\mathsf{T}}$.

    After the first $\ell$ readouts $w_{1:\ell}$, define the unnormalized population vector
    \begin{equation}
        \label{eq:prefix-unnormalized-state}
        \bm{\pi}^{(w_{1:\ell})}
        \;\coloneqq\;
        M^{(w_\ell)} M^{(w_{\ell-1})}\cdots M^{(w_1)}\,\bm{\pi}
        \;\in\; \mathbb{R}_{\ge 0}^2.
    \end{equation}
    This vector contains both the probability weight of the observed partial readout string and the remaining computational-basis population information after the $\ell$-th MCM.
    Therefore the probability of observing the full readout string $w$ is
    \begin{equation}
        \label{eq:string-prob-matrix-product}
        p(w)
        \;=\;
        \bm{1}^{\mathsf{T}}\,\bm{\pi}^{(w)}
        \;=\;
        \bm{1}^{\mathsf{T}}\,M^{(w)}\,\bm{\pi},
        \qquad
        M^{(w)} \coloneqq M^{(w_L)}\cdots M^{(w_1)}.
    \end{equation}
    Whenever $p(w)>0$, the normalized post-measurement population vector after the final MCM, conditioned on the observed readout string, is
    \begin{equation}
        \label{eq:conditional-poststate}
        \Pr(s_L=\cdot \mid w)
        \;=\;
        \frac{\bm{\pi}^{(w)}}{p(w)}.
    \end{equation}
\end{lemma}

\begin{proof}
    The defining interpretation \eqref{eq:joint-prob-def} of the MCM matrix pair is
    \begin{equation}
        \label{eq:kernel-interpretation}
        M^{(o)}_{s',s} \;=\; \Pr(o,s'\mid s),
        \qquad s,s',o\in\{0,1\}.
    \end{equation}
    For a length-$L$ run, let $s_0$ denote the computational-basis state before the first MCM and let $s_\ell$ denote the computational-basis state after the $\ell$-th MCM:
    \begin{equation}
        \label{eq:hidden-state-seq}
        s_0 \;\text{(pre first MCM)},\qquad
        s_\ell \;\text{(post $\ell$-th MCM)} \ \text{for }\ell=1,\dots,L.
    \end{equation}
    Conditioned on the pre-MCM basis state $s_{\ell-1}$, the $\ell$-th MCM produces the readout bit $w_\ell$ and the post-MCM basis state $s_\ell$ with probability $M^{(w_\ell)}_{s_\ell,s_{\ell-1}}$ by \eqref{eq:kernel-interpretation}.
    Therefore, for any fixed computational-basis state sequence $(s_0,\dots,s_L)$, the joint probability of that sequence and the observed readout string $w$ is
    \begin{equation}
        \label{eq:path-prob}
        \Pr(w,s_0,\dots,s_L)
        \;=\;
        \pi_{s_0}\,\prod_{\ell=1}^L M^{(w_\ell)}_{s_\ell,s_{\ell-1}}.
    \end{equation}
    Summing \eqref{eq:path-prob} over all computational-basis state sequences yields
    \begin{equation}
        \label{eq:string-prob-sum-over-paths}
        p(w)
        \;=\;
        \sum_{s_0,\dots,s_L\in\{0,1\}} \pi_{s_0}\,\prod_{\ell=1}^L M^{(w_\ell)}_{s_\ell,s_{\ell-1}}.
    \end{equation}

    To identify the matrix-product structure, define recursively the unnormalized population vectors
    \begin{equation}
        \label{eq:forward-recursion}
        \bm{\pi}^{(w_{1:0})}\coloneqq \bm{\pi},
        \qquad
        \bm{\pi}^{(w_{1:\ell})}\coloneqq M^{(w_\ell)}\,\bm{\pi}^{(w_{1:\ell-1})}
        \ \text{ for }\ell=1,\dots,L.
    \end{equation}
    Expanding \eqref{eq:forward-recursion} entry-wise shows that, for each $\ell$ and each $s_\ell\in\{0,1\}$,
    \begin{equation}
        \label{eq:forward-recursion-entrywise}
        \bigl[\bm{\pi}^{(w_{1:\ell})}\bigr]_{s_\ell}
        \;=\;
        \sum_{s_0,\dots,s_{\ell-1}\in\{0,1\}} \pi_{s_0}\,\prod_{j=1}^\ell M^{(w_j)}_{s_j,s_{j-1}}.
    \end{equation}
    Comparing \eqref{eq:forward-recursion-entrywise} at $\ell=L$ with \eqref{eq:string-prob-sum-over-paths} and then summing over the final computational-basis state $s_L$ yields
    \begin{equation}
        \label{eq:marginalize-final-state}
        p(w)
        \;=\;
        \sum_{s_L\in\{0,1\}}\bigl[\bm{\pi}^{(w)}\bigr]_{s_L}
        \;=\;
        \bm{1}^{\mathsf{T}}\,\bm{\pi}^{(w)}
        \;=\;
        \bm{1}^{\mathsf{T}}\,M^{(w)}\,\bm{\pi},
    \end{equation}
    which is \eqref{eq:string-prob-matrix-product}.

    Finally, \eqref{eq:forward-recursion-entrywise} shows that $\bm{\pi}^{(w)}$ is the joint distribution of the observed readout string and the final computational-basis state, namely $\bigl[\bm{\pi}^{(w)}\bigr]_{s_L}=\Pr(w,s_L)$.
    When $p(w)=\sum_{s_L}\Pr(w,s_L)>0$, Bayes' rule gives $\Pr(s_L=\cdot\mid w)=\Pr(w,\cdot)/p(w)=\bm{\pi}^{(w)}/p(w)$, proving \eqref{eq:conditional-poststate}.
\end{proof}

\subsection{Learning assumptions and diagnostic scope}
\label{app:model:noise-assumptions}

The learning model in SM Sec.~\ref{app:model} treats each mid-circuit measurement (MCM) as a \emph{fixed} two-outcome single-qubit instrument, reduced by the $Z$-twirl to the classical kernel $M^{(o)}_{s',s}=\Pr(o,s'\mid s)$ of Definition~\ref{def:z-twirled-mcm-instrument}.
When the same instrument is applied repeatedly, the resulting readout-string probabilities admit the matrix-product form recorded in Lemma~\ref{lem:matrix-product-string-probabilities}.
The hardware data are well approximated by that model when the following physical assumptions hold.

\begin{assumption}[Stationarity and Markovianity]
    \label{ass:stationarity-markovianity}
    Throughout the learning stage (i.e., across all circuits and shots used to build the measured readout-string distribution), we assume the following.

    \emph{(i) Time-homogeneous MCM within a run.}
    Within a single circuit run containing $L$ consecutive MCM applications, each application is described by the \emph{same} pair of matrices $(M^{(0)},M^{(1)})$ (equivalently, the same $Z$-twirled instrument), independent of the step index $\ell=1,\dots,L$.

    \emph{(ii) Markov property at the level of the reduced single-qubit model.}
    Conditioned on the computational-basis state $s_{\ell-1}\in\{0,1\}$ immediately before the $\ell$-th MCM, the joint distribution of the next observed readout bit and the next post-measurement computational-basis state, $(w_\ell,s_\ell)$, depends only on $s_{\ell-1}$ and not on earlier computational-basis states or earlier readouts.
    Equivalently, the transition kernel is precisely $M^{(w_\ell)}_{s_\ell,s_{\ell-1}}$ as used in Eqs.~\eqref{eq:kernel-interpretation}--\eqref{eq:string-prob-matrix-product}.

    \emph{(iii) Shot-to-shot stationarity (i.i.d.\ sampling).}
    Across repeated shots of any fixed circuit, the underlying model parameters (and hence the string distribution $p(w)$) are stable, so empirical frequencies concentrate around the probabilities generated by the fixed $(M^{(0)},M^{(1)})$.

    \emph{(iv) Effective single-qubit reduction.}
    Any additional degrees of freedom (spectator qubits, classical control electronics, etc.) that influence the measured qubit are assumed to be either negligible or effectively absorbed into a stationary single-qubit instrument description over the timescale of data acquisition.
\end{assumption}

\begin{assumption}[No leakage (base model)]
    \label{ass:no-leakage-base}
    In the base learning model, the measured system is assumed to remain within the computational subspace $\mathrm{span}\{\ket{0},\ket{1}\}$ for the duration of each circuit run and across the learning dataset.
    Consequently, the post-measurement computational-basis state $s'\in\{0,1\}$ appearing in $M^{(o)}_{s',s}=\Pr(o,s'\mid s)$ is sufficient to describe measurement-induced transitions between the two basis states, and the normalization constraints \eqref{eq:z-twirled-mcm-constraints} are appropriate.

    If leakage is appreciable, population can leave the computational basis and the two-state description can fail.
    This work keeps a two-state fitted model.
    The longer-string checks in SM Sec.~\ref{app:validation} test whether the measured probabilities support that description and can expose leakage.
\end{assumption}

\paragraph{Model-mismatch interpretation.}
Assumption~\ref{ass:stationarity-markovianity} is an idealized description of the hardware behavior.
A real device can drift slowly, retain memory of earlier operations, or couple to degrees of freedom that are absent from a stationary single-qubit model.
The longer-string checks in SM Sec.~\ref{app:validation} compare the predictions of the learned model with additional measured probabilities.
A disagreement larger than the finite-shot fluctuations signals an incomplete model.
Further tests are needed to identify the physical cause.

\section{Gauge structure and short readout-string invariants}
\label{app:gauge}

\subsection{Gauge transformation and invariance}
\label{sec:gauge-transformation-invariance}

The similarity-gauge freedom of the two-state representation is intrinsic to learning a $Z$-twirled single-qubit MCM instrument from \emph{readout-string probabilities only} when the effective input is the maximally-mixed state.
Concretely, in our $2\times 2$ classical representation of the $Z$-incoherent post-measurement state, a readout-string probability is a scalar of the form
\begin{equation}
    p(o_1,\dots,o_L)
    \;=\;
    \bm{1}^{\mathsf{T}}\, M^{(o_L)} \cdots M^{(o_1)}\, \bm{\pi}_{\mathrm{mix}},
    \qquad
    \bm{1}:=\begin{pmatrix}1\\1\end{pmatrix},
    \quad
    \bm{\pi}_{\mathrm{mix}}:=\tfrac12\,\bm{1},
\end{equation}
(and similarly for protocols that insert additional trace-preserving linear maps between MCM calls and treat those maps as part of the same gauge-dependent representation).
The key point is that $\bm{\pi}_{\mathrm{mix}}$ and $\bm{1}^{\mathsf{T}}$ are simultaneously invariant under a nontrivial one-parameter family of invertible similarity transformations, implying an \emph{irreducible} one-degree-of-freedom non-identifiability unless external information breaks this symmetry (e.g., a trusted preparation of a non-maximally-mixed input).

\begin{definition}[Gauge matrices and gauge action]
    \label{def:gauge-action}
    For any real parameter $t\in\mathbb{R}\setminus\{\tfrac12\}$, define
    \begin{equation}
        R(t)
        \;:=\;
        \begin{pmatrix}
            1-t & t   \\
            t   & 1-t
        \end{pmatrix},
        \qquad
        \det R(t)=1-2t\neq 0.
    \end{equation}
    We define the associated \emph{gauge action} $\Phi_t$ on any $2\times 2$ matrix $A$ by
    \begin{equation}
        \Phi_t(A)
        \;:=\;
        R(t)^{-1} A\, R(t).
    \end{equation}
    In particular, the gauge-transformed MCM instrument is
    \begin{equation}
        M^{(o)} \longmapsto \widetilde M^{(o)} := \Phi_t\!\big(M^{(o)}\big)=R(t)^{-1} M^{(o)} R(t),
        \qquad o\in\{0,1\}.
    \end{equation}
    The matrix $R(t)$ has unit row and column sums for every real $t$ and is invertible exactly when $t\neq\tfrac12$.
    It is entrywise nonnegative, and hence doubly stochastic, exactly for $t\in[0,1]$.
    A representation gauge only requires invertibility of $R(t)$; physicality is imposed on the transformed instrument matrices $\widetilde M^{(o)}$.
\end{definition}

Two elementary invariances will be repeatedly used.
\begin{equation}
    R(t)\,\bm{\pi}_{\mathrm{mix}} = \bm{\pi}_{\mathrm{mix}},
    \qquad
    \bm{1}^{\mathsf{T}} R(t)^{-1}=\bm{1}^{\mathsf{T}}.
    \label{eq:R-invariants-v0-1}
\end{equation}
The first identity follows from the unit row sums.
Symmetry gives unit column sums and $\bm{1}^{\mathsf T}R(t)=\bm{1}^{\mathsf T}$; multiplying by $R(t)^{-1}$ gives the second identity.

\paragraph{Composition law.}
The matrices $R(t)$ form an Abelian one-parameter group under multiplication.
Direct multiplication gives, for $t_1,t_2\neq \tfrac12$,
\begin{equation}
    R(t_1)R(t_2)
    \;=\;
    \begin{pmatrix}
        1-t_1-t_2+2t_1t_2 & t_1+t_2-2t_1t_2   \\
        t_1+t_2-2t_1t_2   & 1-t_1-t_2+2t_1t_2
    \end{pmatrix}
    \;=\;
    R(t_1\odot t_2),
\end{equation}
with induced parameter composition law
\begin{equation}
    t_1\odot t_2
    \;:=\;
    t_1+t_2-2t_1t_2.
    \label{eq:gauge-composition-law}
\end{equation}
Equivalently, in the basis $\{\bm{1},(1,-1)^\top\}$ the matrix $R(t)$ acts with eigenvalues $1$ and $u:=1-2t$, so the gauge group is isomorphic to multiplication on $\mathbb{R}\setminus\{0\}$.
For $t\in[0,1]$, $R(t)$ is the action of a qubit bit-flip channel on computational-basis probabilities; outside this interval it is the invertible affine continuation used as a representation gauge.
Closely related depolarizing gauges appear in incoherent SPAM characterization, where extra energy levels can reduce the allowed ambiguity through stronger positivity constraints~\cite{chenEnhancingQuantumNoise2025}.
In our setting, the same gauge transformation also acts on outcome-conditioned MCM transition matrices, so it changes the inferred measurement-induced state transitions in addition to the state-preparation and readout parameters.
The identity is $R(0)=I$, and the inverse is
\begin{equation}
    t^{-1}
    \;=\;
    \frac{-t}{1-2t},
    \qquad\text{so that}\qquad
    R(t)^{-1}=R(t^{-1}).
\end{equation}

\begin{theorem}[Gauge invariance of readout-string probability data]
    \label{thm:gauge-invariance-prob-only}
    Fix $\bm{\pi}_{\mathrm{mix}}=\tfrac12(1,1)^\top$.
    Consider any protocol family whose reported data are \emph{only} readout-string probabilities obtained by composing the two-outcome instrument $\{M^{(0)},M^{(1)}\}$ and marginalizing over the final computational-basis state with $\bm{1}^{\mathsf{T}}$.
    More generally, the same statement holds if the protocol includes additional trace-preserving linear maps on $\mathbb{R}^2$ that are treated as unknown parts of the same gauge-dependent representation.
    Then, for every $t\neq \tfrac12$, simultaneously applying the gauge action \( M^{(o)}\mapsto \widetilde M^{(o)}:=R(t)^{-1}M^{(o)}R(t) \) (and conjugating any such interleaved maps by the same $R(t)$) leaves \emph{all} readout-string probabilities invariant.

    Externally characterized interleaves or trusted non-maximally mixed preparations are different.
    They are external reference operations, are not conjugated by this representation change, and can make the gauge parameter observable.
    Consequently, any quantity computed solely from the measured readout-string distribution is gauge-invariant, whereas quantities that use externally characterized operations are conditional on those external characterizations.
\end{theorem}

\begin{proof}
    It suffices to check a single classical-control path corresponding to a fixed readout string and fixed classical-control choices.
    Such a path probability can be written as
    \begin{equation}
        p
        \;=\;
        \bm{1}^{\mathsf{T}} F\, \bm{\pi}_{\mathrm{mix}},
    \end{equation}
    where $F$ is the ordered product of the matrices executed along that path (each factor is either an $M^{(o)}$ or an interleaved trace-preserving map treated as part of the same fitted model).
    Under the gauge action, every factor $A$ is replaced by $\widetilde A=R^{-1}AR$, hence the full path product becomes $\widetilde F = R^{-1} F R$.
    Therefore
    \begin{equation}
        \widetilde p
        \;=\;
        \bm{1}^{\mathsf{T}} \widetilde F\, \bm{\pi}_{\mathrm{mix}}
        \;=\;
        \bm{1}^{\mathsf{T}} (R^{-1} F R)\, \bm{\pi}_{\mathrm{mix}}
        \;=\;
        (\bm{1}^{\mathsf{T}} R^{-1})\,F\,(R \bm{\pi}_{\mathrm{mix}})
        \;=\;
        \bm{1}^{\mathsf{T}} F \bm{\pi}_{\mathrm{mix}}
        \;=\;
        p,
    \end{equation}
    where we used \eqref{eq:R-invariants-v0-1}.
    Summing over classical-control paths (adaptive control) preserves equality, and any post-processing that depends only on the resulting readout-string distribution is therefore gauge-invariant as well.
    An externally characterized map would instead remain fixed under the representation change, which is why it can make the gauge parameter observable rather than becoming part of the gauge-equivalent family.
\end{proof}

\begin{theorem}[The gauge is the only non-identifiability in the full-span $2\times 2$ case]
    \label{thm:one-dof-nonidentifiability}
    Let $\{M^{(0)},M^{(1)}\}$ and $\{\widehat M^{(0)},\widehat M^{(1)}\}$ be two $2\times 2$ realizations with nonnegative entries and satisfying the (aggregate) trace-preservation / normalization constraint
    \begin{equation}
        \bm{1}^{\mathsf{T}}\!\big(M^{(0)}+M^{(1)}\big)=\bm{1}^{\mathsf{T}},
        \qquad
        \bm{1}^{\mathsf{T}}\!\big(\widehat M^{(0)}+\widehat M^{(1)}\big)=\bm{1}^{\mathsf{T}}.
    \end{equation}
    Assume that for every length $L\ge 0$ and every readout string $(o_1,\dots,o_L)\in\{0,1\}^L$,
    \begin{equation}
        \bm{1}^{\mathsf{T}} M^{(o_L)}\cdots M^{(o_1)} \bm{\pi}_{\mathrm{mix}}
        \;=\;
        \bm{1}^{\mathsf{T}} \widehat M^{(o_L)}\cdots \widehat M^{(o_1)} \bm{\pi}_{\mathrm{mix}},
        \qquad \bm{\pi}_{\mathrm{mix}}=\tfrac12(1,1)^\top.
        \label{eq:prob-equality-all-strings}
    \end{equation}
    Further assume the following two full-span conditions for both two-dimensional realizations.
    The first says that the unnormalized state vectors generated from $\bm{\pi}_{\mathrm{mix}}$ by all possible readout strings span the full two-dimensional state space.
    \begin{equation}
        \begin{aligned}
        \mathrm{span}\{ M^{(o_L)}\cdots M^{(o_1)} \bm{\pi}_{\mathrm{mix}} : L\ge 0,\ o_i\in\{0,1\} \} &= \mathbb{R}^2,\\
        \mathrm{span}\{ \widehat M^{(o_L)}\cdots \widehat M^{(o_1)} \bm{\pi}_{\mathrm{mix}} : L\ge 0,\ o_i\in\{0,1\} \} &= \mathbb{R}^2.
        \end{aligned}
        \label{eq:full-span-populations}
    \end{equation}
    The second says that the row vectors obtained by propagating the final marginalization functional $\bm{1}^{\mathsf{T}}$ backward through all possible readout strings span the full dual space.
    \begin{equation}
        \begin{aligned}
        \mathrm{span}\{ \bm{1}^{\mathsf{T}} M^{(o_L)}\cdots M^{(o_1)} : L\ge 0,\ o_i\in\{0,1\} \} &= \mathbb{R}^{1\times 2},\\
        \mathrm{span}\{ \bm{1}^{\mathsf{T}} \widehat M^{(o_L)}\cdots \widehat M^{(o_1)} : L\ge 0,\ o_i\in\{0,1\} \} &= \mathbb{R}^{1\times 2}.
        \end{aligned}
        \label{eq:full-span-dual}
    \end{equation}
    Together, these two span assumptions are the usual ``minimal realization'' condition in linear systems; operationally, they rule out a redundant one-dimensional representation.
    Then there exists a unique parameter $t\in\mathbb{R}\setminus\{\tfrac12\}$ such that
    \begin{equation}
        \widehat M^{(o)} \;=\; R(t)^{-1} M^{(o)} R(t),
        \qquad o\in\{0,1\}.
    \end{equation}
    Equivalently, the observational equivalence class of $\{M^{(0)},M^{(1)}\}$ under readout-string probability access from $\bm{\pi}_{\mathrm{mix}}$ is exactly the one-parameter gauge-equivalent family $\{\Phi_t\}$.
\end{theorem}

\begin{proof}
    Define the set $\mathcal{W}$ consisting of finite readout bit strings over the alphabet $\{0,1\}$, including the empty string $\emptyset$.
    For each readout bit string $w=(o_1,\dots,o_L)$, write \( F_w:=M^{(o_L)}\cdots M^{(o_1)} \) and \( \widehat F_w:=\widehat M^{(o_L)}\cdots \widehat M^{(o_1)}. \)
    By the first full-span condition \eqref{eq:full-span-populations}, the vectors $\{F_w \bm{\pi}_{\mathrm{mix}}\}_{w\in\mathcal{W}}$ span $\mathbb{R}^2$.

    We define a linear map $T:\mathbb{R}^2\to\mathbb{R}^2$ on the spanning set by
    \begin{equation}
        T(F_w \bm{\pi}_{\mathrm{mix}}) \;:=\; \widehat F_w \bm{\pi}_{\mathrm{mix}},
        \qquad w\in\mathcal{W}.
        \label{eq:define-T}
    \end{equation}
    We must check that $T$ is well-defined.
    Suppose $\sum_i c_i F_{w_i} \bm{\pi}_{\mathrm{mix}}=0$ for some finite collection.
    For every string $v$, left multiplication by $\bm{1}^{\mathsf T}F_v$ gives
    \begin{equation}
        0
        =\sum_i c_i\,\bm{1}^{\mathsf T}F_vF_{w_i}\bm{\pi}_{\mathrm{mix}}
        =\sum_i c_i\,\bm{1}^{\mathsf T}\widehat F_v\widehat F_{w_i}\bm{\pi}_{\mathrm{mix}},
    \end{equation}
    where the second equality follows from \eqref{eq:prob-equality-all-strings} for the corresponding concatenated strings.
    The hatted row vectors in \eqref{eq:full-span-dual} span the full dual space, which forces $\sum_i c_i \widehat F_{w_i} \bm{\pi}_{\mathrm{mix}}=0$.
    Hence \eqref{eq:define-T} is consistent and extends linearly to all of $\mathbb{R}^2$.
    Its image contains every vector $\widehat F_w\bm{\pi}_{\mathrm{mix}}$ generated by the hatted realization.
    The hatted full-span condition in Eq.~\eqref{eq:full-span-populations} makes $T$ surjective and therefore invertible.

    Next, for any readout bit string $w$ and any outcome $o\in\{0,1\}$,
    \begin{equation}
        T\!\big(M^{(o)}F_w \bm{\pi}_{\mathrm{mix}}\big)
        \;=\;
        T(F_{(w,o)} \bm{\pi}_{\mathrm{mix}})
        \;=\;
        \widehat F_{(w,o)} \bm{\pi}_{\mathrm{mix}}
        \;=\;
        \widehat M^{(o)} \widehat F_w \bm{\pi}_{\mathrm{mix}}
        \;=\;
        \widehat M^{(o)}\,T(F_w \bm{\pi}_{\mathrm{mix}}).
    \end{equation}
    Since $\{F_w \bm{\pi}_{\mathrm{mix}}\}$ spans $\mathbb{R}^2$, we obtain the intertwining relation
    \begin{equation}
        T\,M^{(o)} \;=\; \widehat M^{(o)}\,T,
        \qquad\text{hence}\qquad
        \widehat M^{(o)} \;=\; T\,M^{(o)}\,T^{-1}.
        \label{eq:similarity-T}
    \end{equation}

    Finally, the empty string gives $F_\emptyset=I$ and $\widehat F_\emptyset=I$, hence
    \begin{equation}
        T \bm{\pi}_{\mathrm{mix}} = \bm{\pi}_{\mathrm{mix}}.
        \label{eq:Tv0}
    \end{equation}
    Moreover, for any $x$ in the span generated by the vectors $F_w\bm{\pi}_{\mathrm{mix}}$, write $x=\sum_i c_i F_{w_i} \bm{\pi}_{\mathrm{mix}}$; then
    \begin{equation}
        \bm{1}^{\mathsf{T}} T x
        \;=\;
        \sum_i c_i\,\bm{1}^{\mathsf{T}} \widehat F_{w_i} \bm{\pi}_{\mathrm{mix}}
        \;=\;
        \sum_i c_i\,\bm{1}^{\mathsf{T}} F_{w_i} \bm{\pi}_{\mathrm{mix}}
        \;=\;
        \bm{1}^{\mathsf{T}} x,
    \end{equation}
    using \eqref{eq:prob-equality-all-strings}.
    By the first full-span condition, this holds for all $x\in\mathbb{R}^2$, so
    \begin{equation}
        \bm{1}^{\mathsf{T}} T=\bm{1}^{\mathsf{T}}.
        \label{eq:1T-T}
    \end{equation}

    Now specialize to $2\times 2$.
    Write $T=\begin{pmatrix}a&b\\c&d\end{pmatrix}$.
    The constraints \eqref{eq:Tv0}--\eqref{eq:1T-T} with $\bm{\pi}_{\mathrm{mix}}=\tfrac12(1,1)^\top$ imply
    \begin{equation}
        a+b=1,\quad c+d=1,\quad a+c=1,\quad b+d=1,
    \end{equation}
    hence $c=b$ and $d=a$, i.e.
    \begin{equation}
        T=\begin{pmatrix}a&1-a\\1-a&a\end{pmatrix}=R(s),\qquad s:=1-a.
    \end{equation}
    Since $T$ is invertible, $\det T = 2a-1\neq 0$, equivalently $s\neq\tfrac12$.
    The inverse law gives a unique $t=-s/(1-2s)$ with $T^{-1}=R(t)$ and $T=R(t)^{-1}$.
    Substitution into \eqref{eq:similarity-T} yields $\widehat M^{(o)}=R(t)^{-1}M^{(o)}R(t)$.
    The map $T$ is fixed by \eqref{eq:define-T} on a spanning set, and $t\mapsto R(t)$ is injective on $\mathbb{R}\setminus\{\tfrac12\}$, proving uniqueness.
\end{proof}

\paragraph{Breaking the gauge requires external information.}
The invariance mechanism in Theorem~\ref{thm:gauge-invariance-prob-only} uses only the two facts in \eqref{eq:R-invariants-v0-1}.
Therefore the gauge disappears as soon as either of these facts is violated by an externally characterized capability.
A trusted non-maximally-mixed input $v_\star$ that is not fixed by $R(t)$, a trusted final-state marginalization different from $\bm{1}^{\mathsf{T}}$, or an externally characterized interleaved map held fixed rather than conjugated by $R(t)$ makes the gauge parameter observable.
Within the reduced model, such external information turns SPAM error separation into an externally referenced inversion problem.
In our setting, the readout-string sampling construction deliberately enforces $\bm{\pi}_{\mathrm{mix}}=\tfrac12(1,1)^\top$ via uniform single-qubit Pauli input randomization, and we do not assume any additional externally characterized operation or state that breaks this symmetry; hence one gauge degree of freedom remains irreducible.

\subsection{Identifiable invariants and short readout-string extraction}
\label{app:invariants}
\label{app:invariants:ch-recurrence}
\label{sec:identifiable-min-gauge}

Theorems~\ref{thm:gauge-invariance-prob-only} and~\ref{thm:one-dof-nonidentifiability} imply that, in the $2\times 2$ $Z$-twirled representation with $\bm{\pi}_{\mathrm{mix}}=\tfrac12(1,1)^\top$ and final-state marginalization $\bm{1}^{\mathsf{T}}$, readout-string probability data determine the MCM instrument $\{M^{(0)},M^{(1)}\}$ only up to the one-parameter gauge-equivalent family $M^{(o)}\sim R(t)^{-1}M^{(o)}R(t)$.
Accordingly, the natural targets of learning are \emph{gauge-invariant} scalars (quantities that are constant along this family), together with the physically allowed gauge choices compatible with physical constraints (nonnegativity, etc.).

\begin{proposition}[A complete five-scalar set of identifiable gauge invariants]
    \label{prop:identifiable-invariants}
    Assume $\{M^{(0)},M^{(1)}\}\subset\mathbb{R}_{\ge 0}^{2\times 2}$ and the normalization constraint $\bm{1}^{\mathsf{T}}\!\big(M^{(0)}+M^{(1)}\big)=\bm{1}^{\mathsf{T}}$.
    Then, under the gauge action $M^{(o)}\mapsto \widetilde M^{(o)}:=R(t)^{-1}M^{(o)}R(t)$ with $t\neq\tfrac12$, the following five scalars are invariant and hence identifiable from readout-string probability data.
    \begin{align}
         & p(0)=\bm{1}^{\mathsf{T}} M^{(0)} \bm{\pi}_{\mathrm{mix}}
        \qquad\text{(equivalently, } \bm{1}^{\mathsf{T}} M^{(0)} \bm{1}=2p(0) \text{ since } \bm{\pi}_{\mathrm{mix}}=\tfrac12\bm{1}\text{)}, \\
         & \Tr \!\big(M^{(0)}\big),\quad \det\!\big(M^{(0)}\big),\quad
        \Tr \!\big(M^{(1)}\big),\quad \det\!\big(M^{(1)}\big).
    \end{align}
    Moreover, in the interior of the parameter region (away from degeneracies where the algebraic reconstruction becomes singular), these five invariants form a locally complete scalar description of the observational equivalence class.
    The instrument has $6$ free real parameters after imposing $\bm{1}^{\mathsf{T}}(M^{(0)}+M^{(1)})=\bm{1}^{\mathsf{T}}$, and exactly one real degree of freedom is unidentifiable due to the gauge, leaving $5$ identifiable degrees of freedom.
\end{proposition}

\begin{proof}
    Gauge invariance is immediate.
    \begin{enumerate}[label=\textup{(\roman*)}, leftmargin=2.2em]
        \item For the single-shot probability of readout bit $0$ from $\bm{\pi}_{\mathrm{mix}}$,
              \begin{equation}
                  \bm{1}^{\mathsf{T}} \widetilde M^{(0)} \bm{\pi}_{\mathrm{mix}}
                  = \bm{1}^{\mathsf{T}} (R^{-1} M^{(0)} R) \bm{\pi}_{\mathrm{mix}}
                  = (\bm{1}^{\mathsf{T}} R^{-1})\, M^{(0)}\, (R \bm{\pi}_{\mathrm{mix}})
                  = \bm{1}^{\mathsf{T}} M^{(0)} \bm{\pi}_{\mathrm{mix}},
              \end{equation}
              using $\bm{1}^{\mathsf{T}} R^{-1}=\bm{1}^{\mathsf{T}}$ and $R \bm{\pi}_{\mathrm{mix}}=\bm{\pi}_{\mathrm{mix}}$ from \eqref{eq:R-invariants-v0-1}.
        \item $\Tr (\cdot)$ and $\det(\cdot)$ are similarity invariants, hence unchanged under conjugation by any invertible $R(t)$.
    \end{enumerate}

    For the dimension count, note that the pair $(M^{(0)},M^{(1)})$ has $8$ real parameters, and the constraint $\bm{1}^{\mathsf{T}}(M^{(0)}+M^{(1)})=\bm{1}^{\mathsf{T}}$ imposes $2$ independent linear constraints, leaving $6$ degrees of freedom (inequality constraints such as nonnegativity restrict the feasible region but do not change the local dimension in the interior).
    Theorem~\ref{thm:one-dof-nonidentifiability} shows that readout-string probability access from $\bm{\pi}_{\mathrm{mix}}$ leaves exactly one continuous gauge degree of freedom unidentifiable, and no more under the full-span assumptions stated there, hence the observational quotient has dimension $6-1=5$.
    Therefore, any complete identifiable scalar description must contain at least five independent scalars; the displayed list is a convenient choice for our subsequent reconstruction pipeline.
\end{proof}

The results in this subsection clarify the appropriate learning targets once gauge freedom is taken into account.
In particular, after establishing that readout-string probability access from $\bm{\pi}_{\mathrm{mix}}=\tfrac12(1,1)^\top$ identifies the instrument only up to the one-parameter family generated by $R(t)$, it is natural to focus on scalar quantities that are invariant under the similarity action $M^{(o)}\mapsto R(t)^{-1}M^{(o)}R(t)$.
The invariants listed in Proposition~\ref{prop:identifiable-invariants} provide a convenient set of parameters for the observational equivalence class in the interior of the physically admissible region.

The following recurrence gives practical estimators from short readout-bit-string distributions, and SM Sec.~\ref{app:reconstruction} then uses the resulting invariant data (together with a fixed gauge choice and physicality constraints) to construct explicit matrix representatives $M^{(0)},M^{(1)}$.

\paragraph{Cayley--Hamilton recurrence for repeated outcomes.}
The recurrence follows from the $2\times 2$ Cayley--Hamilton identity.
It yields closed-form estimators for $\Tr (M^{(0)}),\det(M^{(0)})$ and $\Tr (M^{(1)}),\det(M^{(1)})$ from measured readout-string probabilities of length at most $3$.

For $n\ge 0$, write $0^n$ (resp.\ $1^n$) for the length-$n$ all-zero (resp.\ all-one) readout bit string, with $0^0=1^0=\emptyset$ (the empty string).
Define the all-zero/all-one readout-string probabilities
\begin{equation}
    \label{eq:SnTn-def}
    S_n \;\coloneqq\; p(0^n),
    \qquad
    T_n \;\coloneqq\; p(1^n),
    \qquad n\ge 0.
\end{equation}
Here digits in the argument of $p(\cdot)$ denote readout bit strings; for example, $p(0^n)$ is the probability of the all-zero readout string of length $n$.
By definition, $S_0=T_0=p(\emptyset)=\bm{1}^{\mathsf{T}}\bm{\pi}=1$, and by Lemma~\ref{lem:matrix-product-string-probabilities},
\begin{equation}
    \label{eq:SnTn-matrix-product}
    S_n \;=\; \bm{1}^{\mathsf{T}} (M^{(0)})^{n}\bm{\pi},
    \qquad
    T_n \;=\; \bm{1}^{\mathsf{T}} (M^{(1)})^{n}\bm{\pi}.
\end{equation}
In a single length-$3$ experiment one directly estimates $p(w)$ for $w\in\{0,1\}^3$; the required values $S_1,S_2,S_3$ (and similarly $T_1,T_2,T_3$) are then obtained by summing over later readouts via \eqref{eq:prefix-marginal}.

\begin{proposition}[Second-order recurrence for $S_n$ and $T_n$]
    \label{prop:ch-second-order-recurrence}
    For all $n\ge 0$,
    \begin{align}
        \label{eq:Sn-recurrence}
        S_{n+2}
        \; & =\;
        \Tr \!\big(M^{(0)}\big)\,S_{n+1}
        \;-\;
        \det\!\big(M^{(0)}\big)\,S_{n}, \\[2pt]
        \label{eq:Tn-recurrence}
        T_{n+2}
        \; & =\;
        \Tr \!\big(M^{(1)}\big)\,T_{n+1}
        \;-\;
        \det\!\big(M^{(1)}\big)\,T_{n}.
    \end{align}
\end{proposition}

\begin{proof}
    We prove \eqref{eq:Sn-recurrence}; the proof for \eqref{eq:Tn-recurrence} is identical with $M^{(1)}$ in place of $M^{(0)}$.

    Let $A\coloneqq M^{(0)}\in\mathbb{R}^{2\times 2}$.
    By the Cayley--Hamilton theorem,
    \begin{equation}
        \label{eq:CH-2x2}
        A^2 \;-\; \Tr (A)\,A \;+\; \det(A)\,I \;=\; 0.
    \end{equation}
    Multiplying \eqref{eq:CH-2x2} on the right by $A^{n}$ yields
    \begin{equation}
        \label{eq:CH-shifted}
        A^{n+2} \;=\; \Tr (A)\,A^{n+1} \;-\; \det(A)\,A^{n}.
    \end{equation}
    Left-multiplying \eqref{eq:CH-shifted} by $\bm{1}^{\mathsf{T}}$ and right-multiplying by $\bm{\pi}$ gives
    \begin{equation}
        \bm{1}^{\mathsf{T}}A^{n+2}\bm{\pi}
        \;=\;
        \Tr (A)\,\bm{1}^{\mathsf{T}}A^{n+1}\bm{\pi}
        \;-\;
        \det(A)\,\bm{1}^{\mathsf{T}}A^{n}\bm{\pi}.
    \end{equation}
    Using \eqref{eq:SnTn-matrix-product} with $A=M^{(0)}$ identifies the left-hand side as $S_{n+2}$ and the two terms on the right as $S_{n+1}$ and $S_n$, proving \eqref{eq:Sn-recurrence}.
\end{proof}

\begin{corollary}[Closed-form formulas for $\Tr $ and $\det$ from length-3 readout strings]
    \label{cor:tr-det-from-length3}
    Suppose $S_1^2\neq S_2$ and $T_1^2\neq T_2$.
    Then
    \begin{align}
        \label{eq:trM0-from-S}
        \Tr \!\big(M^{(0)}\big)
        \; & =\;
        \frac{S_1S_2 - S_3}{S_1^2 - S_2},
           &
        \det\!\big(M^{(0)}\big)
        \; & =\;
        \frac{S_2^2 - S_1S_3}{S_1^2 - S_2},
        \\[4pt]
        \label{eq:trM1-from-T}
        \Tr \!\big(M^{(1)}\big)
        \; & =\;
        \frac{T_1T_2 - T_3}{T_1^2 - T_2},
           &
        \det\!\big(M^{(1)}\big)
        \; & =\;
        \frac{T_2^2 - T_1T_3}{T_1^2 - T_2}.
    \end{align}
\end{corollary}

\begin{proof}
    We again prove the statements for $M^{(0)}$; the $M^{(1)}$ formulas follow identically.

    Set $A=M^{(0)}$ and use the recurrence \eqref{eq:Sn-recurrence} at $n=0,1$, together with $S_0=1$:
    \begin{align}
        S_2 & = \Tr (A)\,S_1 - \det(A)\,S_0
        \;=\; \Tr (A)\,S_1 - \det(A),
        \label{eq:S2-eq}                     \\
        S_3 & = \Tr (A)\,S_2 - \det(A)\,S_1.
        \label{eq:S3-eq}
    \end{align}
    From \eqref{eq:S2-eq},
    \begin{equation}
        \label{eq:det-in-terms-tr}
        \det(A) \;=\; \Tr (A)\,S_1 - S_2.
    \end{equation}
    Substituting \eqref{eq:det-in-terms-tr} into \eqref{eq:S3-eq} yields
    \begin{equation}
        S_3
        \;=\;
        \Tr (A)\,S_2
        \;-\;
        \bigl(\Tr (A)\,S_1 - S_2\bigr)S_1
        \;=\;
        \Tr (A)\,(S_2 - S_1^2) + S_1S_2.
    \end{equation}
    Rearranging gives
    \begin{equation}
        \Tr (A)\,(S_1^2 - S_2) \;=\; S_1S_2 - S_3.
    \end{equation}
    Under the stated nondegeneracy condition $S_1^2\neq S_2$, this implies \( \Tr (A)=\frac{S_1S_2-S_3}{S_1^2-S_2}, \) establishing \eqref{eq:trM0-from-S}; plugging this into the intermediate identity \eqref{eq:det-in-terms-tr} gives
    \[
        \det(A)
        =
        \frac{(S_1S_2-S_3)S_1}{S_1^2-S_2}-S_2
        =
        \frac{S_2^2-S_1S_3}{S_1^2-S_2},
    \]
    which is the determinant formula in \eqref{eq:trM0-from-S}.
\end{proof}

\subsection{Long-readout-string consistency checks and failure modes}
\label{subsec:invariants-overcomplete}
\label{subsec:invariants-failure-modes}

Under the maximally mixed-input design and away from the denominator degeneracies of Corollary~\ref{cor:tr-det-from-length3}, the length-$3$ readout-string distribution already determines the gauge-invariant scalars needed for representative reconstruction (SM Sec.~\ref{app:reconstruction}); under the full-span assumptions of Theorem~\ref{thm:one-dof-nonidentifiability}, these scalars then determine the instrument up to the residual one-parameter gauge.
Longer readout strings do not introduce new independent degrees of freedom for the base $2\times 2$, time-homogeneous instrument model (Definition~\ref{def:z-twirled-mcm-instrument}); instead, they provide \emph{overcomplete constraints} that can be used for (i) self-consistency diagnostics and (ii) variance reduction by overdetermined fitting.

\begin{proposition}[Longer-readout-string constraints for self-consistency]
    \label{prop:overcomplete-long-string-constraints}
    Assume the $Z$-twirled, time-homogeneous instrument model of Definition~\ref{def:z-twirled-mcm-instrument}, and define $S_n$ and $T_n$ as in Eq.~\eqref{eq:SnTn-def}.

    Then the second-order recurrences in Proposition~\ref{prop:ch-second-order-recurrence} hold for \emph{all} $n\ge 0$.
    Consequently, once $\Tr(M^{(0)}),\det(M^{(0)})$ are fixed, the entire sequence $\{S_n\}_{n\ge 0}$ is determined by $(S_0,S_1)$; similarly, $\{T_n\}_{n\ge 0}$ is determined by $(T_0,T_1)$ once $\Tr(M^{(1)}),\det(M^{(1)})$ are fixed.

    In particular, let $L\ge 3$ and suppose one measures a length-$L$ readout-string distribution (so that the all-zero/all-one readout-string probabilities $S_1,\dots,S_L$ and $T_1,\dots,T_L$ are available by summing over later readouts via Eq.~\eqref{eq:prefix-marginal}).
    Then for every $n=0,1,\dots,L-2$ one has the \emph{overcomplete constraints}
    \begin{align}
        \label{eq:overcomplete-constraints}
        S_{n+2}
        \;-\;
        \Tr\!\big(M^{(0)}\big)\,S_{n+1}
        \;+\;
        \det\!\big(M^{(0)}\big)\,S_{n}
        \; & =\; 0,
        \\[2pt]
        T_{n+2}
        \;-\;
        \Tr\!\big(M^{(1)}\big)\,T_{n+1}
        \;+\;
        \det\!\big(M^{(1)}\big)\,T_{n}
        \; & =\; 0.
    \end{align}
    Under finite-shot sampling, the empirical left-hand sides obtained by replacing each $p(\cdot)$ with $\widehat{p}(\cdot)$ provide a direct self-consistency check of the model assumptions (stationarity, Markovianity, and absence of leakage).
\end{proposition}

\begin{proof}
    The first two claims are immediate from Proposition~\ref{prop:ch-second-order-recurrence}: the recurrence is valid for all $n\ge 0$, hence it determines $S_{n+2}$ from $(S_{n+1},S_n)$ once $\Tr(M^{(0)})$ and $\det(M^{(0)})$ are fixed; iterating from $(S_0,S_1)$ determines the full sequence.
    The same argument applies to $T_n$.

    The constraints \eqref{eq:overcomplete-constraints} are a direct restatement of Eqs.~\eqref{eq:Sn-recurrence}--\eqref{eq:Tn-recurrence}.
\end{proof}

\begin{figure*}[!t]
    \centering
    \includegraphics[width=0.75\textwidth]{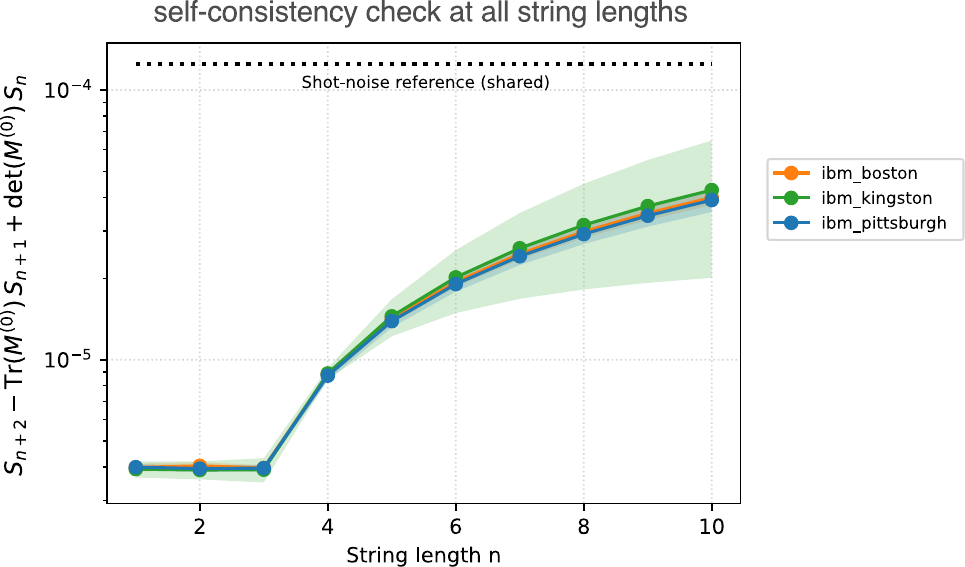}
    \caption{
        \textbf{Self-consistency checks from long strings.}
        The second-order recurrences in Proposition~\ref{prop:ch-second-order-recurrence} hold for all $n\ge 0$ under the model assumptions, hence they provide overcomplete constraints that can be used for self-consistency checks.
        Here we plot the empirical residuals obtained by evaluating the left-hand sides of Eqs.~\eqref{eq:overcomplete-constraints} with empirical probabilities $\widehat{p}(\cdot)$ from a length-$10$ readout-string distribution.
        The residuals are close to zero for small $n$ (where the probabilities are large and hence have small relative error), while finite-shot amplification and possible residual model mismatch can make the longer-readout-string residuals larger.
        For the length-$10$ data, the recurrence indices are $n=0,\ldots,8$.
        Shaded regions indicate $\pm 2$ standard errors averaged over all 20 qubits for each device; the shot-noise scale is computed from the empirical probabilities using the formulas in SM Sec.~\ref{app:stats}.
        }
    \label{fig:consistency-checks}
\end{figure*}

\begin{remark}[Overdetermined fitting for $\Tr/\det$ from long strings]
    \label{rem:overcomplete-wls-trdet}
    When $L>3$, one can exploit the redundant equations \eqref{eq:overcomplete-constraints} to refine the length-$3$ closed-form estimators in Corollary~\ref{cor:tr-det-from-length3}.

    For example, for $M^{(0)}$ one may determine $(\Tr(M^{(0)}),\det(M^{(0)}))$ by minimizing a least-squares loss built from empirical residuals,
    \begin{equation}
        \label{eq:wls-trdet}
        \min_{x,y\in\mathbb{R}}
        \sum_{n=0}^{L-2}
        \Bigl(
        \widehat{p}(0^{n+2}) - x\,\widehat{p}(0^{n+1}) + y\,\widehat{p}(0^{n})
        \Bigr)^2.
    \end{equation}
    Optionally, one may incorporate inverse-variance weights using the shot-noise covariance formulas in SM Sec.~\ref{app:stats}.
    This estimator is consistent under the base model, typically has reduced variance compared to purely algebraic ratios, and is substantially more stable in near-degenerate regimes where denominators such as $\widehat{p}(0)^2-\widehat{p}(00)$ become small, as discussed below.
    An identical refinement applies to $(\Tr(M^{(1)}),\det(M^{(1)}))$ using the $\widehat{p}(1^n)$ sequence.
\end{remark}

\paragraph{Failure modes and remedies.}
\begin{observation}[Instabilities, diagnostics, and practical fixes]
    \label{obs:invariants-failure-modes}
    The algebraic extraction formulas in Corollary~\ref{cor:tr-det-from-length3} are exact in the noiseless model, but can be numerically unstable for finite-shot data or for parameters near singular regimes.
    The dominant failure modes and corresponding remedies are the following.

    \begin{enumerate}[label=(\roman*), leftmargin=2.2em]

        \item \textbf{Exact degeneracy $S_1^2=S_2$ (and/or $T_1^2=T_2$).}
        In this case the length-$3$ ratio formulas \eqref{eq:trM0-from-S} (and/or \eqref{eq:trM1-from-T}) are undefined.
        A sufficient condition is that the initial state $\bm{\pi}_{\mathrm{mix}}$ used in the learning protocol is a right-eigenvector of $M^{(0)}$ (resp.\ $M^{(1)}$), which implies $S_n=(S_1)^n$ (resp.\ $T_n=(T_1)^n$) and hence $S_2=S_1^2$.
        This situation is non-generic but can occur at highly symmetric parameter points.

              \emph{Remedies include} (a) use a different effective initial state (or an additional known pre-rotation) so that the input is not aligned with an eigenvector; (b) incorporate longer readout strings and solve an overdetermined fit as in Remark~\ref{rem:overcomplete-wls-trdet}; (c) augment the invariant set using mixed readout strings (e.g., $010$, $101$) when available in the measured readout-string distribution, which typically breaks the degeneracy even if the all-zero/all-one statistics do not.

        \item \textbf{Small denominators and amplified shot noise.}
        Even when $S_1^2\neq S_2$, the denominator $S_1^2-S_2$ in \eqref{eq:trM0-from-S} can be small compared to the finite-shot fluctuations of $(\widehat{p}(0),\widehat{p}(00),\widehat{p}(000))$, producing large variance and occasional outliers.

              \emph{Remedies include} (a) prefer overdetermined estimators \eqref{eq:wls-trdet} when longer readout strings are available; (b) report uncertainty using the bootstrap workflow in SM Sec.~\ref{app:stats} rather than relying on a single point estimate; (c) when only length-$3$ data are available, treat near-zero denominators as a flag that the learning run is information-poor for that estimator and increase the shot budget or adjust the circuit family.

        \item \textbf{Empty physically allowed gauge set due to statistical fluctuations.}
        The physically allowed gauge set $\mathcal{T}_\mathrm{phys}$ (SM Sec.~\ref{app:physicality}) is defined by nonnegativity inequalities on a gauge-transformed representative.
        Finite-shot errors can push the inferred invariants outside the physical region, yielding an empty inferred feasibility set or disconnected intervals that are not stable across repetitions.

              \emph{Remedies include} (a) increase the shot budget for the specific probabilities that dominate the relevant inequalities (typically those entering small denominators); (b) perform a constrained projection step by replacing the raw invariant estimates by the closest point (in a suitable metric) that satisfies the physicality constraints; (c) use bootstrap resampling to quantify the probability of feasibility failure and to produce conservative gauge bands (SM Sec.~\ref{app:stats}).

        \item \textbf{Algebraic-candidate ambiguity under noise.}
        The representative reconstruction map in SM Sec.~\ref{app:reconstruction} can yield multiple algebraic candidates even in the noiseless setting; finite-shot noise can further produce spurious candidates that fit short readout strings but violate physicality or badly mispredict longer-readout-string statistics.

              \emph{Remedies include} (a) retain all algebraic candidates through reconstruction and then apply physicality constraints (SM Sec.~\ref{app:physicality}); (b) use longer-readout-string overcomplete constraints (Proposition~\ref{prop:overcomplete-long-string-constraints}) as an additional candidate-selection criterion; (c) for reporting, follow the gauge-band convention of Remark~\ref{rem:recommended-reporting-convention} rather than committing to a single representative without justification.
    \end{enumerate}
\end{observation}

\section{Representative reconstruction and gauge-band reporting}
\label{app:reconstruction}
\label{app:physicality}

\subsection{Algebraic reconstruction and gauge representatives}
\label{app:reconstruction:algebraic}

The \emph{learnable} gauge-invariant scalars in Proposition~\ref{prop:identifiable-invariants}, together with the TP constraints~\eqref{eq:z-twirled-mcm-constraints}, determine an explicit algebraic reconstruction of a \emph{representative} pair $(M^{(0)},M^{(1)})$.
As emphasized in SM Sec.~\ref{app:gauge}, the representative is not unique: it depends on a gauge choice, and different choices are related by the similarity action $\Phi_t$.
The reconstruction map is closed-form, exact in the noiseless model, and returns \emph{all algebraic candidates} compatible with the invariants.

\begin{proposition}[Representative reconstruction from invariants and constraints]
    \label{prop:reconstruction-from-invariants-branch}
    Let $(M^{(0)},M^{(1)})\in(\mathbb{R}^{2\times 2})^2$ be a $Z$-twirled MCM model as in Definition~\ref{def:z-twirled-mcm-instrument}, and assume the learning protocol starts from $\bm{\pi}_{\mathrm{mix}}=\tfrac12\bm{1}$ and returns only readout-string probabilities.

    Suppose we are given the following \emph{gauge-invariant} scalars.
    \begin{equation}
        \label{eq:recon-invariant-inputs}
        S_1 = p(0)=\bm{1}^{\mathsf{T}}M^{(0)}\bm{\pi}_{\mathrm{mix}},
        \qquad
        \Tr \!\big(M^{(0)}\big),\ \det\!\big(M^{(0)}\big),
        \qquad
        \Tr \!\big(M^{(1)}\big),\ \det\!\big(M^{(1)}\big),
    \end{equation}
    where $\Tr (\cdot)$ and $\det(\cdot)$ can be computed from short readout-string data via Corollary~\ref{cor:tr-det-from-length3}.

    Fix a \emph{gauge condition} by choosing a scalar
    \begin{equation}
        \label{eq:gauge-choice-a}
        a \;\coloneqq\; M^{(0)}_{0,0},
    \end{equation}
    and write the unknown entries as
    \begin{equation}
        \label{eq:entry-parameterization}
        M^{(0)}=\begin{pmatrix} a & b \\ c & d \end{pmatrix},
        \qquad
        M^{(1)}=\begin{pmatrix} e & f \\ g & h \end{pmatrix}.
    \end{equation}
    Define the shorthand invariants (to avoid collision with $T_n=p(1^n)$ from Eq.~\eqref{eq:SnTn-def} and $D=1-2t$ from Eq.~\eqref{eq:D-def}, we use Greek letters here)
    \begin{equation}
        \label{eq:recon-shorthands}
        \tau_0 \coloneqq \Tr \!\big(M^{(0)}\big),\quad
        \delta_0 \coloneqq \det\!\big(M^{(0)}\big),\quad
        \tau_1 \coloneqq \Tr \!\big(M^{(1)}\big),\quad
        \delta_1 \coloneqq \det\!\big(M^{(1)}\big).
    \end{equation}
    The following reconstruction formulas then hold.

    \begin{enumerate}[label=\textup{(\roman*)}, leftmargin=2.1em]
        \item \textbf{Reconstructing $M^{(0)}$ (two algebraic candidates).}
        The diagonal entry $d$ and the symmetric sums/products of the off-diagonals are fixed by
              \begin{equation}
                  \label{eq:recon-M0-constraints}
                  d \;=\; \tau_0-a,
                  \qquad
                  b+c \;=\; 2S_1 - \tau_0,
                  \qquad
                  bc \;=\; a d - \delta_0.
              \end{equation}
              Consequently, $(b,c)$ are the (at most two) solutions of the quadratic equation
              \begin{equation}
                  \label{eq:recon-quadratic}
                  x^2 - (2S_1-\tau_0)\,x + (a(\tau_0-a)-\delta_0)=0,
              \end{equation}
              i.e.,
              \begin{equation}
                  \label{eq:recon-bc-roots}
                  \{b,c\}
                  \;=\;
                  \left\{
                  \frac{(2S_1-\tau_0)\pm\sqrt{\Delta}}{2},
                  \frac{(2S_1-\tau_0)\mp\sqrt{\Delta}}{2}
                  \right\},
                  \qquad
                  \Delta \coloneqq (2S_1-\tau_0)^2 - 4\big(a(\tau_0-a)-\delta_0\big).
              \end{equation}

        \item \textbf{Reconstructing $M^{(1)}$ (generic case $u\neq v$).}
        Using the TP constraints \eqref{eq:z-twirled-mcm-constraints} in the matrix form $\bm{1}^{\mathsf{T}}(M^{(0)}+M^{(1)})=\bm{1}^{\mathsf{T}}$ (equivalently, each column of $M^{(0)}+M^{(1)}$ sums to $1$), define
              \begin{equation}
                  \label{eq:uv-def}
                  u \;\coloneqq\; e+g \;=\; 1-(a+c),
                  \qquad
                  v \;\coloneqq\; f+h \;=\; 1-(b+d).
              \end{equation}
              If $u\neq v$, then the determinant constraint $eh-fg=\delta_1$ becomes \emph{linear} in $h$ and yields
              \begin{equation}
                  \label{eq:h-solution-generic}
                  h \;=\; \frac{\delta_1 + v(u-\tau_1)}{u-v},
                  \qquad
                  e=\tau_1-h,
                  \qquad
                  g=u-e,
                  \qquad
                  f=v-h.
              \end{equation}

        \item \textbf{Reconstructing $M^{(1)}$ (degenerate case $u=v$).}
        If $u=v$, the formula \eqref{eq:h-solution-generic} is ill-conditioned.
        In that case one can instead use the gauge-invariant identity
              \begin{equation}
                  \label{eq:trM0M1-identity}
                  \Tr \!\big(M^{(0)}M^{(1)}\big)
                  \;=\;
                  \delta_0 + \delta_1 + (\tau_0-1)(\tau_1-1),
              \end{equation}
              which follows from the $2\times 2$ determinant expansion $\det(A+B)=\det(A)+\det(B)+\Tr (A)\Tr (B)-\Tr (AB)$ applied to $A=M^{(0)}$, $B=M^{(1)}$, together with $\det(M^{(0)}+M^{(1)})=\Tr (M^{(0)}+M^{(1)})-1$ for a column-stochastic $2\times 2$ matrix.
              Writing $X\coloneqq \Tr (M^{(0)}M^{(1)})$, the four linear constraints
              \begin{equation}
                  \label{eq:M1-linear-system}
                  e+g=u,\quad f+h=v,\quad e+h=\tau_1,\quad a e + c f + b g + d h = X
              \end{equation}
              determine $(e,f,g,h)$ uniquely whenever the associated $4\times 4$ linear system is nonsingular.

        \item \textbf{Algebraic candidates and admissibility constraints.}
        For each fixed gauge choice $a$, the quadratic \eqref{eq:recon-quadratic} yields at most two candidate pairs $(b,c)$ and hence at most two reconstructed pairs $(M^{(0)},M^{(1)})$.
        Physical realizability additionally requires the entrywise conditional-probability constraints $0\le M^{(o)}_{s',s}\le 1$ and the TP constraints already enforced above; these constraints typically rule out one candidate and restrict the admissible gauge choices to a physically allowed set (constructed in Subsec.~\ref{app:physicality:admissibility}).
    \end{enumerate}
\end{proposition}

\begin{proof}
    Since $\bm{\pi}_{\mathrm{mix}}=\tfrac12\bm{1}$, the single-shot probability $S_1=p(0)$ satisfies
    \begin{equation}
        \label{eq:S1-sum-entries}
        2S_1
        \;=\;
        \bm{1}^{\mathsf{T}}M^{(0)}\bm{1}
        \;=\;
        a+b+c+d.
    \end{equation}
    Fixing $a$ and using $\tau_0=a+d$ gives $d=\tau_0-a$.
    Subtracting $\tau_0$ from \eqref{eq:S1-sum-entries} yields $b+c=2S_1-\tau_0$, proving the first two relations in \eqref{eq:recon-M0-constraints}.
    The determinant condition $\delta_0=\det(M^{(0)})=ad-bc$ rearranges to $bc=ad-\delta_0$, giving the third relation.
    Eliminating $b$ via $b=(2S_1-\tau_0)-c$ yields the quadratic \eqref{eq:recon-quadratic} and hence the two-root structure \eqref{eq:recon-bc-roots}.

    Next, the TP constraints \eqref{eq:z-twirled-mcm-constraints} imply column-stochasticity of $M^{(0)}+M^{(1)}$, i.e. \( (a+c)+(e+g)=1 \) and \( (b+d)+(f+h)=1, \) establishing \eqref{eq:uv-def}.
    With $e=\tau_1-h$, $g=u-e=u-\tau_1+h$, and $f=v-h$, the determinant equation $\delta_1=eh-fg$ becomes
    \begin{equation}
        \delta_1
        \;=\;
        (\tau_1-h)h - (v-h)(u-\tau_1+h)
        \;=\;
        (u-v)h - v(u-\tau_1),
    \end{equation}
    which is linear in $h$ whenever $u\neq v$, giving \eqref{eq:h-solution-generic}.

    In the degenerate case $u=v$, the above linear reduction fails.
    Identity \eqref{eq:trM0M1-identity} follows from the $2\times 2$ determinant expansion stated in the proposition and the fact that any column-stochastic $2\times 2$ matrix has eigenvalues $\{1,\det(\cdot)\}$ and hence $\det(\cdot)=\Tr (\cdot)-1$.
    Finally, \eqref{eq:M1-linear-system} is exactly the collection of the three column/trace constraints and the scalar constraint $\Tr (M^{(0)}M^{(1)})=X$, written entrywise, and thus determines $(e,f,g,h)$ whenever the corresponding linear system is invertible.
\end{proof}

\refstepcounter{suppalgorithm}
\begin{modernalgobox}[unbreakable,title={Algorithm~\thesuppalgorithm: Representative reconstruction from invariants}]
    \begin{algorithmic}[1]
        \Require $S_1$, $\tau_0,\delta_0$, $\tau_1,\delta_1$ as in \eqref{eq:recon-invariant-inputs}, and a gauge choice $a$ in
        \eqref{eq:gauge-choice-a}.
        \Ensure A list of candidate pairs $(M^{(0)},M^{(1)})$ satisfying the invariant constraints and TP constraints.
        \State $d \gets \tau_0-a$, \quad $S \gets 2S_1-\tau_0$, \quad $K \gets a d - \delta_0$, \quad $\Delta \gets S^2-4K$.
        \If{$\Delta<0$}
        \State \Return empty list (no real-valued reconstruction for this gauge choice).
        \EndIf
        \State Compute the two candidate roots $c_{\pm}\gets \frac{S\pm\sqrt{\Delta}}{2}$ and set $b_{\pm}\gets S-c_{\pm}$.
        \For{each candidate pair $(b,c)\in\{(b_{+},c_{+}), (b_{-},c_{-})\}$}
        \State $M^{(0)} \gets \begin{pmatrix} a & b \\ c & d \end{pmatrix}$.
        \State $u \gets 1-(a+c)$, \quad $v \gets 1-(b+d)$.
        \If{$u\neq v$}
        \State $h \gets \frac{\delta_1 + v(u-\tau_1)}{u-v}$; \quad $e\gets \tau_1-h$; \quad $g\gets u-e$; \quad $f\gets v-h$.
        \Else
        \State $X\gets \delta_0 + \delta_1 + (\tau_0-1)(\tau_1-1)$.
        \State Solve the linear system \eqref{eq:M1-linear-system} for $(e,f,g,h)$ (if solvable).
        \EndIf
        \State $M^{(1)} \gets \begin{pmatrix} e & f \\ g & h \end{pmatrix}$.
        \State Append $(M^{(0)},M^{(1)})$ to the candidate list.
        \EndFor
        \State \Return candidate list (to be filtered by admissibility constraints in Subsec.~\ref{app:physicality:admissibility}).
    \end{algorithmic}
\end{modernalgobox}

\noindent\textbf{Algorithm~\thesuppalgorithm.}
Algebraic reconstruction of a representative instrument from invariants (two-candidate output).
\label{alg:reconstruct-representative-branch}

\paragraph{Implementation notes.}
In finite-shot data, the discriminant $\Delta$ in \eqref{eq:recon-bc-roots} and the denominator $u-v$ in \eqref{eq:h-solution-generic} can become small or slightly negative due to sampling noise.
In practice we (i) clamp $\Delta$ to zero when it is within numerical tolerance of zero, (ii) treat $|u-v|$ below a tolerance as the degenerate case and switch to the linear system \eqref{eq:M1-linear-system}, and (iii) postpone any physicality enforcement to the gauge-admissibility analysis in Subsec.~\ref{app:physicality:admissibility}, where one can robustly compute the physically allowed gauge set and select a representative according to the reporting policy.

\paragraph{Relating an entry-fixed gauge choice to the parameter $t$.}
\label{sec:reconstruction:a-vs-t}

The reconstruction pipeline can \emph{fix a representative} of the gauge-equivalent family by prescribing a specific matrix entry (e.g., $a:=M^{(0)}_{0,0}$).
Since this entry is \emph{not} gauge-invariant, we relate this entry-fixed choice to the parameter $t$ appearing in the gauge matrices $R(t)$ (Definition~\ref{def:gauge-action}).

\begin{lemma}[Entry value along the $R(t)$ gauge family]
    \label{lem:a-vs-t}
    Fix a representative instrument $\{M^{(0)}(0), M^{(1)}(0)\}$ as reference and define the associated gauge-equivalent family by
    \begin{equation}
        M^{(o)}(t)\;:=\;R(t)^{-1} M^{(o)}(0)\,R(t),
        \qquad o\in\{0,1\},\qquad t\in\mathbb{R}\setminus\{\tfrac12\}.
    \end{equation}
    Write
    \begin{equation}
        M^{(0)}(0)
        \;=\;
        \begin{pmatrix}
            m_{00} & m_{01} \\
            m_{10} & m_{11}
        \end{pmatrix},
        \qquad
        m_{ij}:=\bigl[M^{(0)}(0)\bigr]_{i,j},
    \end{equation}
    and define the entry function
    \begin{equation}
        a(t)\;:=\;\bigl[M^{(0)}(t)\bigr]_{0,0}.
    \end{equation}
    Then $a(t)$ is a rational function of $t$ given explicitly by
    \begin{equation}
        a(t)
        \;=\;
        \frac{
        (m_{01}+m_{11}-m_{00}-m_{10})\,t^{2}
        + (2m_{00}+m_{10}-m_{01})\,t
        - m_{00}
        }{2t-1}.
        \label{eq:a-of-t-rational}
    \end{equation}

    Conversely, prescribing an entry-fixed gauge choice $a\in\mathbb{R}$ and requiring $a=a(t)$ forces $t$ to satisfy the quadratic constraint
    \begin{equation}
        (m_{01}+m_{11}-m_{00}-m_{10})\,t^{2}
        + (2m_{00}+m_{10}-m_{01}-2a)\,t
        + (a-m_{00})
        \;=\;0.
        \label{eq:t-from-a-quadratic}
    \end{equation}
    In the special (nongeneric) case $m_{01}+m_{11}=m_{00}+m_{10}$ (i.e., the two columns of $M^{(0)}(0)$ have equal sums), Eq.~\eqref{eq:t-from-a-quadratic} collapses to a linear-fractional relation
    \begin{equation}
        t
        \;=\;
        \frac{m_{00}-a}{\,2m_{00}+m_{10}-m_{01}-2a\,},
        \qquad
        \text{provided} \quad 2m_{00}+m_{10}-m_{01}-2a\neq 0.
        \label{eq:t-from-a-rational-special}
    \end{equation}
\end{lemma}

\begin{proof}
    Using $R(t)=\begin{psmallmatrix}1-t&t\\ t&1-t\end{psmallmatrix}$ and $R(t)^{-1}=\frac{1}{1-2t}\begin{psmallmatrix}1-t&-t\\ -t&1-t\end{psmallmatrix}$, direct multiplication gives \( M^{(0)}(t)=R(t)^{-1}M^{(0)}(0)R(t) \) and the $(0,0)$ entry simplifies to Eq.~\eqref{eq:a-of-t-rational}.
    Rearranging $a=a(t)$ after multiplying both sides by $(2t-1)$ yields Eq.~\eqref{eq:t-from-a-quadratic}, and the degenerate case $m_{01}+m_{11}=m_{00}+m_{10}$ gives the linear-fractional form \eqref{eq:t-from-a-rational-special}.
\end{proof}

\paragraph{Practical use in gauge selection.}
Equation~\eqref{eq:a-of-t-rational} provides an explicit conversion between two common ways of parameterizing the same one-dimensional gauge-equivalent family, either by the parameter $t$ (used throughout SM Sec.~\ref{app:gauge} and the admissibility analysis in Subsec.~\ref{app:physicality:admissibility}) or by an entry-fixed gauge choice $a:=M^{(0)}_{0,0}$.
In general Eq.~\eqref{eq:t-from-a-quadratic} admits up to two real solutions for $t$, reflecting that a single matrix entry is not a globally one-to-one parameterization of the family; in applications the physically allowed gauge set selects the admissible representative(s).

\subsection{Admissibility and physically allowed gauge-set computation}
\label{app:physicality:admissibility}
\label{app:physicality:union-intervals}

Given a reconstructed candidate representative (and, in general, multiple algebraic candidates) from Subsec.~\ref{app:reconstruction:algebraic}, the admissibility analysis characterizes which similarity transforms in the $R(t)$ gauge family preserve physicality, i.e., keep every entry within $[0,1]$ while preserving instrument normalization.
This yields the \emph{physically allowed gauge set}, the set of representatives used for candidate filtering and for reporting gauge-dependent quantities as gauge bands.

\begin{definition}[Admissible representative and physically allowed gauge set]
    \label{def:feasible-gauge-set}
    Fix a reference representative $(M^{(0)},M^{(1)})$ satisfying the physical constraints of Definition~\ref{def:z-twirled-mcm-instrument}.
    For each $t\in\mathbb{R}\setminus\{\tfrac12\}$ define
    \begin{equation}
        \label{eq:R-t-def-physicality}
        R(t)\;\coloneqq\;\begin{pmatrix}1-t & t\\ t & 1-t\end{pmatrix},
        \qquad
        \widetilde{M}^{(o)}(t)\;\coloneqq\;R(t)^{-1}M^{(o)}R(t),
        \qquad o\in\{0,1\}.
    \end{equation}
    We say that $t$ is \emph{admissible} (for the chosen reference representative) if
    \begin{equation}
        \label{eq:admissible-t-entrywise}
        0 \le \widetilde{M}^{(o)}_{s',s}(t) \le 1
        \qquad\text{for all }o,s,s'\in\{0,1\}.
    \end{equation}
    The corresponding \emph{physically allowed gauge set} is
    \begin{equation}
        \label{eq:Tfeas-def}
        \mathcal{T}_\mathrm{phys}
        \;\coloneqq\;
        \Bigl\{t\in\mathbb{R}\setminus\{\tfrac12\}\ :\ 0\le \widetilde{M}^{(o)}_{s',s}(t)\le 1
        \ \text{for all }o,s,s'\in\{0,1\}\Bigr\},
    \end{equation}
    where inequalities are interpreted entrywise.
    Because normalization is preserved along this family, the upper bound is implied by entrywise nonnegativity; we state the full $[0,1]$ condition to emphasize the conditional-probability interpretation.
\end{definition}

\begin{lemma}[Normalization is preserved along the gauge family]
    \label{lem:gauge-preserves-normalization}
    Let $(M^{(0)},M^{(1)})$ satisfy the instrument normalization $\bm{1}^{\mathsf{T}}(M^{(0)}+M^{(1)})=\bm{1}^{\mathsf{T}}$.
    Then for every $t\neq \tfrac12$,
    \begin{equation}
        \label{eq:gauge-preserves-normalization}
        \bm{1}^{\mathsf{T}}\bigl(\widetilde{M}^{(0)}(t)+\widetilde{M}^{(1)}(t)\bigr)
        \;=\;
        \bm{1}^{\mathsf{T}}.
    \end{equation}
    Consequently, for such a reference representative, admissibility of $t$ is equivalent to the ability to physically realize $(\widetilde{M}^{(0)}(t),\widetilde{M}^{(1)}(t))$ in the sense of Definition~\ref{def:z-twirled-mcm-instrument}.
\end{lemma}

\begin{proof}
    Since $R(t)$ has column sums equal to $1$, we have $R(t)\bm{1}=\bm{1}$, and one checks directly that $\bm{1}^{\mathsf{T}}R(t)^{-1}=\bm{1}^{\mathsf{T}}$ for all $t\neq \tfrac12$.
    Therefore,
    \[
        \bm{1}^{\mathsf{T}}\bigl(\widetilde{M}^{(0)}(t)+\widetilde{M}^{(1)}(t)\bigr)
        =
        \bm{1}^{\mathsf{T}}R(t)^{-1}\bigl(M^{(0)}+M^{(1)}\bigr)R(t)
        =
        \bm{1}^{\mathsf{T}}\bigl(M^{(0)}+M^{(1)}\bigr)R(t)
        =
        \bm{1}^{\mathsf{T}}R(t)
        =
        \bm{1}^{\mathsf{T}}.
    \]
    The final claim follows because Definition~\ref{def:z-twirled-mcm-instrument} requires only entrywise nonnegativity together with the above normalization, and the latter holds automatically along this family.
\end{proof}

\paragraph{Optional relaxation under finite data.}
In finite-shot reconstructions, a numerically reconstructed reference representative may violate exact entrywise nonnegativity by a small tolerance.
When convenient, one may replace \eqref{eq:admissible-t-entrywise} by the relaxed bounds $\widetilde{M}^{(o)}_{s',s}(t)\in[-\varepsilon_\mathrm{phys},\,1+\varepsilon_\mathrm{phys}]$ for all $o,s,s'$, with a user-chosen $\varepsilon_\mathrm{phys}\ge 0$.
The analytic interval characterization below extends verbatim by simple coefficient shifts.

\paragraph{Analytic interval computation.}

The set $\mathcal{T}_\mathrm{phys}$ in \eqref{eq:Tfeas-def} is a finite union of intervals, and the following construction computes it analytically.

\begin{lemma}[Diagonal form of the gauge action]
    \label{lem:diagonal-form-gauge-action}
    Define the scalar
    \begin{equation}
        \label{eq:D-def}
        D\;\coloneqq\;1-2t,
    \end{equation}
    and the orthogonal change-of-basis matrix
    \begin{equation}
        \label{eq:U-def}
        U\;\coloneqq\;\frac{1}{\sqrt{2}}
        \begin{pmatrix}
            1 & 1  \\
            1 & -1
        \end{pmatrix},
        \qquad U^{\mathsf{T}}=U^{-1}.
    \end{equation}
    Then for all $t\neq \tfrac12$,
    \begin{equation}
        \label{eq:R-diagonalized}
        U^{\mathsf{T}}R(t)U=\mathrm{diag}(1,D),
        \qquad
        U^{\mathsf{T}}R(t)^{-1}U=\mathrm{diag}\!\left(1,\frac{1}{D}\right).
    \end{equation}
    For any real $2\times 2$ matrix $P$, writing
    \begin{equation}
        \label{eq:P-tilde-abgd}
        \widetilde{P}\;\coloneqq\;U^{\mathsf{T}}PU
        \;=\;
        \begin{pmatrix}
            \alpha & \beta  \\
            \gamma & \delta
        \end{pmatrix},
    \end{equation}
    the conjugated matrix $P(t)\coloneqq R(t)^{-1}PR(t)$ obeys
    \begin{equation}
        \label{eq:abgd-update}
        U^{\mathsf{T}}P(t)U
        \;=\;
        \begin{pmatrix}
            \alpha   & D\,\beta \\
            \gamma/D & \delta
        \end{pmatrix}.
    \end{equation}
    In particular, $\alpha$ and $\delta$ are invariant along the gauge family, while $\beta\gamma$ is invariant.
\end{lemma}

\begin{proof}
    The identities \eqref{eq:R-diagonalized} are verified by direct multiplication using $R(t)=(1-t)I+tX$ and the fact that $U$ diagonalizes $X$.
    Equation \eqref{eq:abgd-update} follows from
    \[
        U^{\mathsf{T}}P(t)U
        =
        \bigl(U^{\mathsf{T}}R(t)^{-1}U\bigr)\,\bigl(U^{\mathsf{T}}PU\bigr)\,\bigl(U^{\mathsf{T}}R(t)U\bigr)
        =
        \mathrm{diag}\!\left(1,\frac{1}{D}\right)\widetilde{P}\,\mathrm{diag}(1,D).
    \]
\end{proof}

\begin{lemma}[Entry formulas and quadratic constraints]
    \label{lem:entry-formulas-quadratic}
    Let $P\in\mathbb{R}^{2\times 2}$ and let $\alpha,\beta,\gamma,\delta$ be defined by \eqref{eq:P-tilde-abgd}.
    Write $P(t)=R(t)^{-1}PR(t)=\bigl(p_{ij}(t)\bigr)_{i,j\in\{0,1\}}$.
    Then, with $D=1-2t$,
    \begin{align}
        \label{eq:entry-formulas}
        p_{00}(t) & = \frac{1}{2}\!\left(\alpha+\delta + D\beta + \frac{\gamma}{D}\right),
                  &
        p_{01}(t) & = \frac{1}{2}\!\left(\alpha-\delta - D\beta + \frac{\gamma}{D}\right),
        \nonumber                                                                          \\
        p_{10}(t) & = \frac{1}{2}\!\left(\alpha-\delta + D\beta - \frac{\gamma}{D}\right),
                  &
        p_{11}(t) & = \frac{1}{2}\!\left(\alpha+\delta - D\beta - \frac{\gamma}{D}\right).
    \end{align}
    Consequently, for each fixed entry $p_{ij}(t)$ the bound $p_{ij}(t)\ge 0$ is equivalent to a quadratic inequality in $D$.
    \begin{equation}
        \label{eq:quadratic-ineq-template}
        p_{ij}(t)\ge 0
        \quad\Longleftrightarrow\quad
        q_{ij}(D)\ \begin{cases}
            \ge 0, & D>0, \\
            \le 0, & D<0,
        \end{cases}
    \end{equation}
    where $q_{ij}(D)\coloneqq 2D\,p_{ij}(t)$ is a quadratic polynomial.
    Explicitly,
    \begin{align}
        \label{eq:qij-explicit}
        q_{00}(D) & = \beta D^2 + (\alpha+\delta)D + \gamma,
                  &
        q_{01}(D) & = -\beta D^2 + (\alpha-\delta)D + \gamma,
        \nonumber                                             \\
        q_{10}(D) & = \beta D^2 + (\alpha-\delta)D - \gamma,
                  &
        q_{11}(D) & = -\beta D^2 + (\alpha+\delta)D - \gamma.
    \end{align}
\end{lemma}

\begin{proof}
    The entry formulas \eqref{eq:entry-formulas} follow by inserting \eqref{eq:abgd-update} into $P(t)=U\bigl(U^{\mathsf{T}}P(t)U\bigr)U^{\mathsf{T}}$ and multiplying out.
    Multiplying each expression in \eqref{eq:entry-formulas} by $2D$ yields \eqref{eq:qij-explicit}.
    If $D>0$ then $p_{ij}(t)\ge 0$ is equivalent to $2Dp_{ij}(t)\ge 0$; if $D<0$ then the inequality reverses, giving \eqref{eq:quadratic-ineq-template}.
\end{proof}

\begin{proposition}[Computing $\mathcal{T}_\mathrm{phys}$ analytically]
    \label{prop:compute-t-feas}
    Fix a reference representative $(M^{(0)},M^{(1)})$ satisfying Definition~\ref{def:z-twirled-mcm-instrument}.
    Then the physically allowed gauge set $\mathcal{T}_\mathrm{phys}$ in \eqref{eq:Tfeas-def} is a (finite) union of intervals.
    Moreover, it can be computed as follows.

    For each outcome $o\in\{0,1\}$, perform the following steps.
    \begin{enumerate}[label=\textup{(\roman*)}, leftmargin=2.2em]
        \item Form $\widetilde{M}^{(o)}\coloneqq U^{\mathsf{T}}M^{(o)}U$ and extract $(\alpha_o,\beta_o,\gamma_o,\delta_o)$ as in \eqref{eq:P-tilde-abgd}.

        \item For each entry index $(i,j)\in\{0,1\}^2$, define the quadratic $q_{ij}^{(o)}(D)$ by \eqref{eq:qij-explicit} with $(\alpha,\beta,\gamma,\delta)=(\alpha_o,\beta_o,\gamma_o,\delta_o)$.

        \item Solve the four quadratic inequalities
              \begin{equation}
                  \label{eq:D-feas-o}
                  q_{ij}^{(o)}(D)\ \begin{cases}
                      \ge 0, & D>0, \\
                      \le 0, & D<0,
                  \end{cases}
                  \qquad (i,j)\in\{0,1\}^2,
              \end{equation}
              obtaining an admissible set $\mathcal{D}^{(o)}_\mathrm{feas}\subset\mathbb{R}\setminus\{0\}$ that is a union of intervals on the $D$-axis.

        \item Intersect the outcome-wise admissible sets.
              \begin{equation}
                  \label{eq:D-feas-intersect}
                  \mathcal{D}_\mathrm{feas}\;\coloneqq\;\mathcal{D}^{(0)}_\mathrm{feas}\cap \mathcal{D}^{(1)}_\mathrm{feas}.
              \end{equation}

        \item Map back to $t$ using the monotone relation $t=\tfrac{1-D}{2}$ from \eqref{eq:D-def}.
              \begin{equation}
                  \label{eq:D-to-t-map}
                  \mathcal{T}_\mathrm{phys}
                  \;=\;
                  \left\{\frac{1-D}{2}\ :\ D\in\mathcal{D}_\mathrm{feas}\right\}.
              \end{equation}
              Each interval $[D_-,D_+]$ maps to the interval $\bigl[\tfrac{1-D_+}{2},\,\tfrac{1-D_-}{2}\bigr]$.
    \end{enumerate}
\end{proposition}

\begin{proof}
    By Lemma~\ref{lem:entry-formulas-quadratic}, each entrywise nonnegativity constraint is equivalent to a single quadratic inequality in $D$ on each of the two half-lines $(0,\infty)$ and $(-\infty,0)$.
    The solution set of a quadratic inequality on a half-line is a union of at most two intervals, hence the intersection over finitely many constraints and the two outcomes $o\in\{0,1\}$ yields a finite union of intervals $\mathcal{D}_\mathrm{feas}\subset\mathbb{R}\setminus\{0\}$.
    The linear change of variables \eqref{eq:D-to-t-map} preserves the “finite union of intervals” structure, and excludes $t=\tfrac12$ automatically because $D=0$ is excluded.
\end{proof}

\subsection{Representative choice and gauge-band reporting}
\label{app:physicality:representative}
\label{app:physicality:bands}

The physically allowed gauge set $\mathcal{T}_\mathrm{phys}$ is typically a union of intervals determined by the admissibility constraints (cf.\ Proposition~\ref{prop:compute-t-feas}).
Its components with opposite signs of $D=1-2t$ differ by a discrete exchange of the two hidden-state labels.

\begin{definition}[$D>0$ reporting branch]
    \label{def:canonical-gauge-component}
    For a fixed reconstruction reference, define
    \begin{equation}
        \label{eq:Tphys-plus-def}
        \mathcal{T}_\mathrm{phys}^{(+)}
        \;\coloneqq\;
        \bigl\{t\in\mathcal{T}_\mathrm{phys}:D=1-2t>0\bigr\}.
    \end{equation}
    We use $\mathcal{T}_\mathrm{phys}^{(+)}$ for representative matrices and gauge-dependent rate bands.
    This fixes the hidden-state labels relative to the reconstruction reference and carries no independent physical orientation information.
\end{definition}

Every $t\in\mathcal{T}_\mathrm{phys}^{(+)}$ yields a physical representative.
The convention below selects one such representative for plotting and explicit matrix reporting; scientific conclusions use bands over the full $D>0$ branch.

\begin{definition}[Max-margin (slack) gauge selection]
    \label{def:max-margin-gauge-selection}
    Let $\mathcal{T}_\mathrm{phys}^{(+)}$ be the $D>0$ reporting branch in Definition~\ref{def:canonical-gauge-component}.
    For each $t\in\mathcal{T}_\mathrm{phys}^{(+)}$ define the gauge-transformed representative matrices $\widetilde{M}^{(o)}(t)=R(t)^{-1}M^{(o)}R(t)$ as in \eqref{eq:R-t-def-physicality}.
    Define the \emph{entrywise slack} at $t$ by
    \begin{equation}
        \label{eq:slack-def}
        \mathrm{slack}(t)
        \;\coloneqq\;
        \min_{\substack{o\in\{0,1\}\\ i,j\in\{0,1\}}}
        \min\!\Bigl\{
        \widetilde{M}^{(o)}_{i,j}(t),\;
        1-\widetilde{M}^{(o)}_{i,j}(t)
        \Bigr\}.
    \end{equation}
    (The upper margin is included to avoid representatives with entries extremely close to $1$, which can also amplify finite-shot sensitivity in later reconstructions.)

    A \emph{max-margin representative} is any
    \begin{equation}
        \label{eq:t-rep-argmax}
        t_{\mathrm{rep}}
        \;\in\;
        \argmax_{t\in\mathcal{T}_\mathrm{phys}^{(+)}}\,\mathrm{slack}(t).
    \end{equation}
    If the maximizer is not unique, we break ties by choosing the minimizer of $|t|$ and then the smaller $t$.

\end{definition}

\paragraph{Existence of a maximizer.}
Under the nondegenerate conditions encountered in practice, $\mathcal{T}_\mathrm{phys}^{(+)}$ is a finite union of closed bounded intervals that avoid $t=\tfrac12$, hence is compact.
Since $\mathrm{slack}(t)$ is continuous on $\mathcal{T}_\mathrm{phys}^{(+)}$ (being the minimum of finitely many continuous functions), the argmax in \eqref{eq:t-rep-argmax} is nonempty.

\begin{proposition}[Efficient max-margin selection on interval unions]
    \label{prop:efficient-max-margin-selection}
    Assume $\mathcal{T}_\mathrm{phys}^{(+)}$ is given as a finite union of disjoint intervals $\mathcal{T}_\mathrm{phys}^{(+)}=\bigcup_{k=1}^K [t_k^-,t_k^+]$ with $t_k^+<\tfrac12$ for each $k$.
    Define $D=1-2t$ and set $\mathcal{D}_k=[D_k^-,D_k^+]$ as the image of $[t_k^-,t_k^+]$ under $D=1-2t$ (so $0\notin\mathcal{D}_k$).

    For each outcome $o\in\{0,1\}$ and entry index $(i,j)\in\{0,1\}^2$, the function $D\mapsto \widetilde{M}^{(o)}_{i,j}(t(D))$ is of the form
    \begin{equation}
        \label{eq:affine-D-invD-form}
        f(D)\;=\;u+vD+\frac{w}{D},
    \end{equation}
    with coefficients $(u,v,w)$ determined by $M^{(o)}$ (cf.\ Lemma~\ref{lem:entry-formulas-quadratic}).

    Fix one interval $\mathcal{D}_k$.
    Let $\mathcal{C}_k$ be the finite candidate set consisting of the following points.
    \begin{enumerate}[label=\textup{(\roman*)}, leftmargin=2.2em]
        \item Include the endpoints $D_k^-,D_k^+$.
        \item Include, for each of the $16$ functions in \eqref{eq:slack-def} written in the form \eqref{eq:affine-D-invD-form}, both real stationary roots $D_\ast=\pm\sqrt{w/v}$ whenever $vw>0$, retaining each root that lies in $\mathcal{D}_k$.
        \item Include, for each pair of such functions $f_1,f_2$ in \eqref{eq:affine-D-invD-form}, the intersection points in $\mathcal{D}_k$ obtained by solving $f_1(D)=f_2(D)$, i.e.\ the quadratic equation
              \begin{equation}
                  \label{eq:intersection-quadratic}
                  (v_1-v_2)D^2 + (u_1-u_2)D + (w_1-w_2)=0.
              \end{equation}
    \end{enumerate}
    Then any maximizer of $\mathrm{slack}(t)$ over $t\in[t_k^-,t_k^+]$ is attained at some $D\in\mathcal{C}_k$.
    Consequently, a global maximizer $t_{\mathrm{rep}}$ can be found by evaluating $\mathrm{slack}$ at the mapped candidate points $t=(1-D)/2$ over all $k$ and selecting the best value (with the tie-breaking rule of Definition~\ref{def:max-margin-gauge-selection}).
\end{proposition}

\begin{proof}
    On a fixed interval $\mathcal{D}_k$ (which avoids $D=0$), each function appearing inside the minimum in \eqref{eq:slack-def} is smooth and has the form \eqref{eq:affine-D-invD-form}.
    The pointwise minimum of finitely many smooth functions is continuous and is differentiable except at points where at least two of the functions coincide (kinks).
    Therefore, a maximizer of the minimum on the compact interval $\mathcal{D}_k$ must occur either at an endpoint, at an interior point where the active function is differentiable and has zero derivative, or at an interior kink where two functions tie.
    Cases (ii) and (iii) correspond exactly to the stationary points in (ii) and the pairwise intersections in (iii), because for \eqref{eq:affine-D-invD-form} we have $f'(D)=v-w/D^2$ and $f'(D)=0$ iff $D=\pm\sqrt{w/v}$ when $vw>0$, and $f_1(D)=f_2(D)$ reduces to \eqref{eq:intersection-quadratic} after clearing denominators.
\end{proof}

\paragraph{Gauge bands for derived quantities.}

Gauge non-identifiability implies that many practically relevant quantities (e.g.\ readout and backaction error rates) depend on the chosen representative $t\in\mathcal{T}_\mathrm{phys}^{(+)}$.
Accordingly, we use two reporting modes.
(i) a single representative choice $t_{\mathrm{rep}}$ used only for plotting and for quoting one explicit matrix pair, and (ii) gauge-robust \emph{bands} obtained by optimizing over $\mathcal{T}_\mathrm{phys}^{(+)}$.

\begin{definition}[Gauge bands for derived quantities]
    \label{def:gauge-bands}
    Let $q$ be any scalar quantity computed from a representative instrument $(\widetilde{M}^{(0)}(t),\widetilde{M}^{(1)}(t))$ at gauge parameter $t\in\mathcal{T}_\mathrm{phys}^{(+)}$.
    Define the \emph{gauge band} of $q$ by
    \begin{equation}
        \label{eq:q-band-def}
        q_{\min}
        \;\coloneqq\;
        \min_{t\in\mathcal{T}_\mathrm{phys}^{(+)}} q(t),
        \qquad
        q_{\max}
        \;\coloneqq\;
        \max_{t\in\mathcal{T}_\mathrm{phys}^{(+)}} q(t),
        \qquad
        [q]_{\mathcal{T}_\mathrm{phys}^{(+)}}
        \;\coloneqq\;
        [q_{\min},q_{\max}].
    \end{equation}
    When $q$ is gauge-invariant, the band collapses to the point $q_{\min}=q_{\max}$.

    \smallskip

    In particular, the readout/backaction error rates of Definition~\ref{def:readout-backaction-error-rates} are promoted to functions of $t$ by substituting $\widetilde{M}^{(o)}(t)$ for $M^{(o)}$.
    \begin{align}
        \label{eq:error-rates-as-functions-of-t}
        \epsilon_{1\mid 0}(t)
         & \;\coloneqq\;
        \sum_{s'\in\{0,1\}} \widetilde{M}^{(1)}_{s',0}(t),
         &
        \epsilon_{0\mid 1}(t)
         & \;\coloneqq\;
        \sum_{s'\in\{0,1\}} \widetilde{M}^{(0)}_{s',1}(t),
        \nonumber        \\
        \eta_{\uparrow}(t)
         & \;\coloneqq\;
        \sum_{o\in\{0,1\}} \widetilde{M}^{(o)}_{1,0}(t),
         &
        \eta_{\downarrow}(t)
         & \;\coloneqq\;
        \sum_{o\in\{0,1\}} \widetilde{M}^{(o)}_{0,1}(t).
    \end{align}
\end{definition}

\begin{proposition}[Bands reduce to interval images on connected components]
    \label{prop:bands-are-intervals}
    Write $\mathcal{T}_\mathrm{phys}^{(+)}=\bigcup_{k=1}^K I_k$ as a disjoint union of intervals.
    Let $q(t)$ be any function that is continuous on $\mathcal{T}_\mathrm{phys}^{(+)}$ (e.g.\ any polynomial/rational expression in the entries of $\widetilde{M}^{(0)}(t),\widetilde{M}^{(1)}(t)$ with no singularities on $\mathcal{T}_\mathrm{phys}^{(+)}$).
    Then for each component $I_k$ the image $q(I_k)$ is an interval $[\min_{t\in I_k} q(t),\ \max_{t\in I_k} q(t)]$, and the overall gauge-allowed range is the union $\bigcup_{k=1}^K q(I_k)$.
    In particular, when one reports only the outer envelope $[q_{\min},q_{\max}]$ in \eqref{eq:q-band-def}, the extrema are attained on the interval endpoints and/or at interior stationary points of $q$ within each component.
\end{proposition}

\begin{proof}
    Continuity of $q$ implies $q(I_k)$ is connected for each connected set $I_k$, hence an interval.
    If each $I_k$ is closed and bounded, then $q$ attains its extrema on $I_k$ by compactness.
    The final claim follows from standard one-dimensional calculus.
    Extrema occur at endpoints or at interior stationary points where the derivative vanishes (or at nondifferentiable points, which do not arise when $q$ is smooth on $I_k$).
\end{proof}

\begin{lemma}[Extrema for linear combinations of entries]
    \label{lem:extrema-affine-D-invD}
    Let $I\subset\mathbb{R}\setminus\{\tfrac12\}$ be an interval and let $D=1-2t$ map $I$ to an interval $\mathcal{D}\subset\mathbb{R}\setminus\{0\}$.
    Suppose $q(t)$ can be written as
    \begin{equation}
        \label{eq:q-affine-D-invD}
        q(t(D)) \;=\; a + bD + \frac{c}{D}
    \end{equation}
    on $\mathcal{D}$ (this includes all individual entries and any fixed linear combination of entries, hence includes the error rates in \eqref{eq:error-rates-as-functions-of-t}).
    Then on $\mathcal{D}$ the extrema of $q$ are attained at the endpoints of $\mathcal{D}$ and, if $bc>0$, possibly at either stationary point
    \begin{equation}
        \label{eq:D-star}
        D_\ast \;=\; \pm\sqrt{\frac{c}{b}},
    \end{equation}
    provided $D_\ast\in\mathcal{D}$.
\end{lemma}

\begin{proof}
    Differentiating \eqref{eq:q-affine-D-invD} yields $dq/dD=b-c/D^2$.
    Thus interior critical points exist iff $bc>0$, in which case $dq/dD=0$ iff $D=\pm\sqrt{c/b}$.
    The remaining extrema occur at the interval endpoints.
\end{proof}

\begin{remark}[Recommended reporting convention]
    \label{rem:recommended-reporting-convention}
    For each reconstructed dataset we recommend reporting the following quantities.
    \begin{enumerate}[label=\textup{(\roman*)}, leftmargin=2.2em]
        \item Report the physically allowed gauge set $\mathcal{T}_\mathrm{phys}$ as a union of intervals.
        \item Report a representative choice $t_{\mathrm{rep}}$ (Definition~\ref{def:max-margin-gauge-selection}) used only for plotting a single representative instrument.
        \item Report gauge bands $[q]_{\mathcal{T}_\mathrm{phys}^{(+)}}$ (Definition~\ref{def:gauge-bands}) for all gauge-dependent derived metrics of interest, notably $\epsilon_{1\mid 0},\epsilon_{0\mid 1},\eta_{\uparrow},\eta_{\downarrow}$.
    \end{enumerate}
    The gauge bands quantify an identifiability-limited systematic uncertainty that is conceptually distinct from finite-shot statistical uncertainty, which is treated separately in SM Sec.~\ref{app:stats}.
\end{remark}

\paragraph{Visualization-only midpoint values.}
For device-level plots that require a single point value per qubit, we additionally use the midpoint
\begin{equation}
    \label{eq:q-mid-def}
    q_{\mathrm{mid}}
    \;\coloneqq\;
    \frac{q_{\min}+q_{\max}}{2}
\end{equation}
of any reported band $[q]_{\mathcal{T}}=[q_{\min},q_{\max}]$.
For \prlFigResultsRef{}(a), the domain is $\mathcal{T}=\mathcal{T}_{\mathrm{hw}}^{(+)}$, defined in SM Sec.~\ref{app:spam}.
This midpoint supplies the visualization values in the main text.
Scientific conclusions use the full gauge bands.

\section{Inference from the learned instrument}
\label{app:spam}

\subsection{State-preparation Bloch-vector inference}
\label{app:spam:protocol}

For the SP task, the measured qubit is modeled as an unknown single-qubit state $\rho_{\mathrm{SP}}$ prepared immediately before the first mid-circuit measurement (MCM) in a dynamic circuit, parameterized by its Bloch components
\begin{equation}
    \label{eq:spam:rho-sp-bloch}
    \rho_{\mathrm{SP}}
    \;=\;
    \frac{1}{2}\bigl(I + r_x X + r_y Y + r_z Z\bigr),
    \qquad (r_x,r_y,r_z)\in\mathbb{R}^3.
\end{equation}
Assume a $Z$-twirled single-qubit MCM instrument $(M^{(0)},M^{(1)})$ has already been learned (up to gauge) as in Definition~\ref{def:z-twirled-mcm-instrument} and the subsequent gauge analysis, so that the associated readout confusion matrix $A$ (Definition~\ref{def:readout-backaction-error-rates}) is available for any chosen gauge representative.
The directly inferred quantity is an operational Bloch component obtained by using this learned readout-confusion matrix as the readout transfer function.
In the ideal axis-isolation-layer model below this quantity is the bare component $r_a$; with imperfect inserted layers it should instead be interpreted as the effective component $r_a^{\mathrm{eff}}$ defined in Proposition~\ref{prop:spam-effective-component}, unless those layers are independently characterized.

\begin{protocol}[Estimating the SP Bloch vector with a $Z$-twirled learned MCM]
    \label{prot:spam-bloch-vector}
    Fix an axis $a\in\{x,y,z\}$ and define the corresponding Pauli $P_a\in\{X,Y,Z\}$.
    For each $a$, repeat the following circuit family to estimate the marginal readout probabilities $p_a(o)$ for $o\in\{0,1\}$.
    \begin{enumerate}[label=\textup{(\roman*)}, leftmargin=2.2em]
        \item \textbf{Prepare.}
        Run the SP routine of interest, producing $\rho_{\mathrm{SP}}$ on the measured qubit.

        \item \textbf{Pauli-axis dephasing map.}
        Apply the Pauli-axis dephasing map $\mathcal{D}_a$ (Definition~\ref{def:pauli-axis-dephasing-map}) by sampling a random bit $b\sim\mathrm{Unif}\{0,1\}$ and applying $P_a^b$ (i.e., apply $P_a$ iff $b=1$), then average over $b$ in post-processing.
        By Lemma~\ref{lem:axis-dephasing-keeps-one-component}, the averaged state is
              \begin{equation}
                  \label{eq:spam:after-axis-dephasing}
                  \mathcal{D}_a(\rho_{\mathrm{SP}})
                  \;=\;
                  \frac{1}{2}\bigl(I + r_a P_a\bigr),
                  \qquad
                  r_a \in \{r_x,r_y,r_z\}\ \text{according to }a.
              \end{equation}

        \item \textbf{Rotate the retained axis to $Z$.}
        Apply a fixed single-qubit unitary that maps the retained Pauli component $P_a$ to $Z$ by conjugation.
        Concretely, one may choose
              \begin{align}
                  \label{eq:spam:axis-rotations}
                  a=z: & \ \text{apply } I, \nonumber                                                         \\
                  a=x: & \ \text{apply } e^{+i(\pi/4)Y} \quad (\text{a $-\pi/2$ rotation about }Y), \nonumber \\
                  a=y: & \ \text{apply } e^{-i(\pi/4)X} \quad (\text{a $+\pi/2$ rotation about }X),
              \end{align}
              which satisfy $e^{+i(\pi/4)Y} X e^{-i(\pi/4)Y}=Z$ and $e^{-i(\pi/4)X} Y e^{+i(\pi/4)X}=Z$.

        \item \textbf{Single MCM.}
        Apply one instance of the \emph{same} $Z$-twirled MCM instrument used throughout the learning dataset, and record the readout bit $o\in\{0,1\}$.

        \item \textbf{Data reduction and reporting.}
        From repeated shots, estimate the marginal readout probabilities $p_a(o)$ for each axis $a$.
        Proposition~\ref{prop:spam-bloch-inversion} gives the exact ideal-layer mapping $r_a \mapsto p_a(o)$ and an explicit inversion formula for the operational component in terms of $p_a(o)$ and the readout-confusion probabilities $A_{o\mid s}$.
        Because $A$ is gauge-dependent under the similarity action $M^{(o)}\mapsto R(t)^{-1}M^{(o)}R(t)$, the resulting reconstructed Bloch components should be reported as gauge bands over $\mathcal{T}_{\mathrm{phys}}^{(+)}$ (Definition~\ref{def:gauge-bands}).
    \end{enumerate}
\end{protocol}

\paragraph{Estimator inversion.}
\label{app:spam:estimators}

\begin{proposition}[Mapping from $(r_a)$ to observed readout probabilities and inversion]
    \label{prop:spam-bloch-inversion}
    Let $\rho_{\mathrm{SP}}$ be as in \eqref{eq:spam:rho-sp-bloch}.
    Fix $a\in\{x,y,z\}$ and consider the circuit family in Protocol~\ref{prot:spam-bloch-vector} that (i) applies $\mathcal{D}_a$, (ii) applies a unitary $U_a$ satisfying $U_a P_a U_a^\dagger = Z$, and (iii) performs a single application of the learned $Z$-twirled MCM instrument $(M^{(0)},M^{(1)})$ (Definition~\ref{def:z-twirled-mcm-instrument}).
    Denote by $p_a(o)$ the resulting marginal probability of observing readout $o\in\{0,1\}$.

    Let $A$ be the associated readout confusion matrix (Definition~\ref{def:readout-backaction-error-rates}), $A_{o\mid s}=\Pr(o\mid s)=\sum_{s'} M^{(o)}_{s',s}$.
    Then
    \begin{equation}
        \label{eq:spam:pa-linear}
        p_a(o)
        \;=\;
        A_{o\mid 0}\,\frac{1+r_a}{2}
        \;+\;
        A_{o\mid 1}\,\frac{1-r_a}{2}
        \;=\;
        \frac{A_{o\mid 0}+A_{o\mid 1}}{2}
        \;+\;
        r_a\,\frac{A_{o\mid 0}-A_{o\mid 1}}{2}.
    \end{equation}
    In particular, if $A_{o\mid 0}\neq A_{o\mid 1}$ for some $o$ (equivalently, the readout has nonzero discrimination), then $r_a$ is uniquely determined by $p_a(o)$ and $A$ via
    \begin{equation}
        \label{eq:spam:ra-inversion-A}
        r_a
        \;=\;
        \frac{2p_a(o) - (A_{o\mid 0}+A_{o\mid 1})}{A_{o\mid 0}-A_{o\mid 1}}.
    \end{equation}
    Choosing $o=0$ and using Definition~\ref{def:readout-backaction-error-rates}, $A_{0\mid 0}=1-\epsilon_{1\mid 0}$ and $A_{0\mid 1}=\epsilon_{0\mid 1}$, this can be written as
    \begin{equation}
        \label{eq:spam:ra-inversion-eps}
        r_a
        \;=\;
        \frac{2p_a(0) - (1-\epsilon_{1\mid 0}+\epsilon_{0\mid 1})}{1-\epsilon_{1\mid 0}-\epsilon_{0\mid 1}},
        \qquad\text{provided }1-\epsilon_{1\mid 0}-\epsilon_{0\mid 1}\neq 0.
    \end{equation}

    Under a gauge representative $t$ with $M^{(o)}(t)=R(t)^{-1}M^{(o)}R(t)$ (SM Sec.~\ref{app:physicality}), one obtains a corresponding readout confusion matrix $A(t)$ and hence a gauge-dependent value $r_a(t)$ by replacing $A$ with $A(t)$ in \eqref{eq:spam:ra-inversion-A}.
    The identifiability-limited uncertainty on $r_a$ is reported as the gauge band $[r_a]_{\mathcal{T}_{\mathrm{phys}}^{(+)}}$ (Definition~\ref{def:gauge-bands}).
\end{proposition}

\begin{proof}
    Fix $a\in\{x,y,z\}$.
    By Lemma~\ref{lem:axis-dephasing-keeps-one-component}, the state after Pauli-axis dephasing is $\mathcal{D}_a(\rho_{\mathrm{SP}})=\frac12(I+r_a P_a)$.
    Let $U_a$ satisfy $U_a P_a U_a^\dagger = Z$ (as in \eqref{eq:spam:axis-rotations}); then the rotated state is
    \begin{equation}
        \label{eq:spam:rotated-state}
        \rho_a'
        \;\coloneqq\;
        U_a\,\mathcal{D}_a(\rho_{\mathrm{SP}})\,U_a^\dagger
        \;=\;
        \frac{1}{2}\bigl(I + r_a\,U_a P_a U_a^\dagger\bigr)
        \;=\;
        \frac{1}{2}\bigl(I + r_a Z\bigr).
    \end{equation}
    Hence $\rho_a'$ is diagonal in the computational basis with populations
    \begin{equation}
        \label{eq:spam:populations}
        \pi_{0\mid a} \;=\; \bra{0}\rho_a'\ket{0}=\frac{1+r_a}{2},
        \qquad
        \pi_{1\mid a} \;=\; \bra{1}\rho_a'\ket{1}=\frac{1-r_a}{2}.
    \end{equation}

    The $Z$-twirled MCM model is classical on computational-basis populations.
    By definition of the readout confusion matrix \eqref{eq:assignment-matrix-def}, the marginal probability of observing readout $o$ depends on the pre-MCM basis state $s$ only through $A_{o\mid s}=\Pr(o\mid s)$.
    Therefore, averaging over the diagonal pre-state distribution \eqref{eq:spam:populations} yields
    \begin{equation}
        \label{eq:spam:pa-mixture}
        p_a(o)
        \;=\;
        A_{o\mid 0}\,\pi_{0\mid a} + A_{o\mid 1}\,\pi_{1\mid a}
        \;=\;
        A_{o\mid 0}\,\frac{1+r_a}{2} + A_{o\mid 1}\,\frac{1-r_a}{2},
    \end{equation}
    which is \eqref{eq:spam:pa-linear}.
    If $A_{o\mid 0}\neq A_{o\mid 1}$, then \eqref{eq:spam:pa-linear} has nonzero slope, and solving for $r_a$ gives \eqref{eq:spam:ra-inversion-A}.
    Substituting $A_{0\mid 0}=1-\epsilon_{1\mid 0}$ and $A_{0\mid 1}=\epsilon_{0\mid 1}$ from \eqref{eq:readout-error-rates-def} gives \eqref{eq:spam:ra-inversion-eps}.

    Finally, under the gauge similarity action $M^{(o)}\mapsto R(t)^{-1}M^{(o)}R(t)$, the induced readout confusion matrix $A(t)$ is obtained by the same marginalization $A_{o\mid s}(t)=\sum_{s'} M^{(o)}_{s',s}(t)$.
    Replacing $A$ by $A(t)$ in \eqref{eq:spam:ra-inversion-A} defines $r_a(t)$, and the band $[r_a]_{\mathcal{T}_{\mathrm{phys}}^{(+)}}$ follows by Definition~\ref{def:gauge-bands}.
\end{proof}

\paragraph{Impact of the extra single-qubit gates used in the SP protocol.}
\label{subsec:spam-impact-extra-1q-gates}

The SP protocol in SM Sec.~\ref{app:spam} isolates $r_z,r_x,r_y$ by inserting (i) the Pauli-axis dephasing map (see SM Sec.~\ref{app:randomization:axis-dephasing}) and (ii) for $a\in\{x,y\}$, a basis rotation $U_a$ satisfying $U_a\,\sigma_a\,U_a^\dagger = Z$ (where $\sigma_x=X$, $\sigma_y=Y$, $\sigma_z=Z$).
In the idealized analysis these are treated as perfect unitaries.
In practice they are implemented using additional single-qubit gate layers, whose imperfections bias the inferred Bloch components unless explicitly accounted for.
The following model makes this dependence explicit and supports the conservative reporting convention used throughout.

\paragraph{Implemented-gate model and effective pre-measurement maps.}
\label{subsubsec:spam-impact-extra-1q-gates:model}

We model each requested single-qubit gate layer $g$ by an arbitrary CPTP map $\mathcal{G}_g$ acting on the qubit.
In particular, the Pauli-axis dephasing step for axis $a\in\{x,y,z\}$ is realized by a classical random choice $b\in\{0,1\}$, where $b=0$ applies an ``identity layer'' and $b=1$ applies the Pauli-$a$ layer.
Averaging over this classical randomization yields the implemented Pauli-axis dephasing map
\begin{equation}
    \label{eq:spam-D-imp-def}
    \mathcal{D}_a^{\mathrm{imp}}(\rho)
    \;\coloneqq\;
    \frac12\Bigl(\mathcal{G}_I(\rho) + \mathcal{G}_{P_a}(\rho)\Bigr),
    \qquad
    P_x=X,\; P_y=Y,\; P_z=Z.
\end{equation}
For $a=z$ we take the intended basis rotation to be $U_z=I$.
For $a\in\{x,y\}$, $U_a$ is any fixed unitary such that $U_a\sigma_a U_a^\dagger=Z$; the actual circuit-level realization is again modeled by a CPTP map $\mathcal{G}_{U_a}$.
The total implemented pre-measurement map used by the SP protocol for axis $a$ is therefore
\begin{equation}
    \label{eq:spam-Fa-def}
    \mathcal{F}_a
    \;\coloneqq\;
    \mathcal{G}_{U_a}\circ \mathcal{D}_a^{\mathrm{imp}}.
\end{equation}
Given an underlying (unknown) prepared state
\begin{equation}
    \label{eq:spam-rhoSP-bloch}
    \rho_{\mathrm{SP}}
    \;=\;
    \frac12\bigl(I+r_x X+r_y Y+r_z Z\bigr),
\end{equation}
the actual state entering the learned MCM is $\rho_a^{\mathrm{in}}=\mathcal{F}_a(\rho_{\mathrm{SP}})$.

\paragraph{What the SP estimator returns in the presence of imperfect gates.}
\label{subsubsec:spam-impact-extra-1q-gates:what-estimator-returns}

The learned $Z$-twirled MCM model induces a classical readout map from the \emph{computational-basis populations} of the pre-measurement state to observed readout probabilities.
Consequently, any SP estimator that inverts the learned readout-confusion map returns the \emph{$Z$-basis population imbalance} of the \emph{actual} state entering the MCM, i.e.\ the $Z$-expectation of $\rho_a^{\mathrm{in}}$.

\begin{proposition}[Operational meaning of the inferred SP components with imperfect gate layers]
    \label{prop:spam-effective-component}
    Fix an axis $a\in\{x,y,z\}$ and let $\rho_a^{\mathrm{in}}=\mathcal{F}_a(\rho_{\mathrm{SP}})$ be the state entering the learned MCM in the SP protocol, with $\mathcal{F}_a$ defined in \eqref{eq:spam-Fa-def}.
    Any estimator that (i) uses the learned readout-confusion probabilities to reconstruct the pre-measurement computational-basis populations $(p_0,p_1)$ and then (ii) reports the corresponding Bloch component via $p_0-p_1=\Tr(Z\rho_a^{\mathrm{in}})$ returns
    \begin{equation}
        \label{eq:spam-ra-eff-def}
        r_a^{\mathrm{eff}}
        \;\coloneqq\;
        \Tr\!\bigl(Z\,\mathcal{F}_a(\rho_{\mathrm{SP}})\bigr),
    \end{equation}
    i.e.\ the $Z$-expectation value of the \emph{effective} pre-measurement state produced by the implemented axis-isolation-and-rotation layers.
    In the ideal-gate limit $\mathcal{F}_a=\mathcal{U}_{U_a}\circ \mathcal{D}_a$ (SM Sec.~\ref{app:randomization:axis-dephasing}), one has $r_a^{\mathrm{eff}}=r_a$.
\end{proposition}

\begin{proof}
    By definition, the reconstructed populations $(p_0,p_1)$ are the diagonal entries of the state that actually enters the MCM, here $\rho_a^{\mathrm{in}}$.
    For any single-qubit state $\rho$, writing $\rho=\sum_{s\in\{0,1\}} p_s \ket{s}\!\bra{s} + (\text{off-diagonal})$ yields $\Tr(Z\rho)=p_0-p_1$.
    Applying this to $\rho=\rho_a^{\mathrm{in}}=\mathcal{F}_a(\rho_{\mathrm{SP}})$ gives \eqref{eq:spam-ra-eff-def}.
    In the ideal limit, Pauli-axis dephasing produces $\mathcal{D}_a(\rho_{\mathrm{SP}})=\frac12(I+r_a\sigma_a)$, and the intended basis rotation satisfies $U_a\sigma_a U_a^\dagger=Z$, hence $\mathcal{U}_{U_a}\circ \mathcal{D}_a(\rho_{\mathrm{SP}})=\frac12(I+r_a Z)$ and $\Tr(Z\,\cdot)=r_a$.
\end{proof}

Proposition~\ref{prop:spam-effective-component} isolates the precise role of the additional single-qubit gates.
Gate imperfections do not invalidate the inversion step, but they change the \emph{quantity} being inferred from $r_a$ (the Bloch component of $\rho_{\mathrm{SP}}$) to $r_a^{\mathrm{eff}}$ (the Bloch component after the implemented pre-measurement map $\mathcal{F}_a$).

\paragraph{Dependence on the underlying Bloch vector, axis mixing, and nonunital bias.}
\label{subsubsec:spam-impact-extra-1q-gates:bloch-dependence}

Any single-qubit CPTP map $\mathcal{F}$ transforms the Bloch vector by a linear part plus a state-independent shift.
Concretely, for $\rho=\frac12(I+r_x X+r_y Y+r_z Z)$ and $\mathcal{F}(\rho)=\frac12(I+r'_x X+r'_y Y+r'_z Z)$ there exist a real matrix $T_{\mathcal{F}}\in\mathbb{R}^{3\times 3}$ and a real vector $\bm{c}_{\mathcal{F}}\in\mathbb{R}^3$ such that $\bm{r}'=T_{\mathcal{F}}\bm{r}+\bm{c}_{\mathcal{F}}$.
Applying this to $\mathcal{F}=\mathcal{F}_a$ and taking the $z$-component yields
\begin{equation}
    \label{eq:spam-ra-eff-bloch-map}
    r_a^{\mathrm{eff}}
    \;=\;
    \bigl(T_{\mathcal{F}_a}\bigr)_{z,x}\,r_x
    \;+\;
    \bigl(T_{\mathcal{F}_a}\bigr)_{z,y}\,r_y
    \;+\;
    \bigl(T_{\mathcal{F}_a}\bigr)_{z,z}\,r_z
    \;+\;
    \bigl(\bm{c}_{\mathcal{F}_a}\bigr)_z.
\end{equation}
In the ideal-gate limit, the intended structure implies $\bigl(T_{\mathcal{F}_a}\bigr)_{z,a}=1$ and the remaining terms vanish, recovering $r_a^{\mathrm{eff}}=r_a$.
For imperfect gates, \eqref{eq:spam-ra-eff-bloch-map} shows that the inferred ``$a$-component'' can contain (i) a multiplicative attenuation of $r_a$, (ii) axis-mixing contributions from the other Bloch components, and (iii) a state-independent offset from nonunitality (captured by $(\bm{c}_{\mathcal{F}_a})_z$).

\begin{assumption}[Gate-independent unital Pauli-layer noise for the SP isolation gates]
    \label{assump:spam-unital-pauli-layer-noise}
    Consider the additional single-qubit layers used by the SP protocol (Pauli-axis dephasing randomization by $P_a\in\{X,Y,Z\}$ and the basis rotation $U_a$ mapping $P_a$ to $Z$ by conjugation).
    Assume that each such \emph{inserted} layer $g$ is implemented as a fixed \emph{gate-independent} unital Pauli channel $\mathcal{N}$ applied after the intended unitary action.
    \begin{equation}
        \label{eq:spam-gate-independent-noise-model}
        \rho \;\longmapsto\; \mathcal{N}\!\bigl(g\,\rho\,g^\dagger\bigr),
        \qquad
        \mathcal{N}(I)=I,
        \qquad
        \mathcal{N}(X)=\lambda_x X,\ \mathcal{N}(Y)=\lambda_y Y,\ \mathcal{N}(Z)=\lambda_z Z.
    \end{equation}
    (Equivalently, $\mathcal{N}$ is diagonal in the Pauli basis with eigenvalues $(1,\lambda_x,\lambda_y,\lambda_z)$.)
\end{assumption}

\begin{proposition}[No axis mixing under gate-independent unital Pauli-layer noise]
    \label{prop:spam-no-axis-mixing-unital-pauli}
    Assume Assumption~\ref{assump:spam-unital-pauli-layer-noise}.
    Let $\rho_{\mathrm{SP}}=\frac12(I+r_x X+r_y Y+r_z Z)$ be the state produced by the SP routine.
    Fix $a\in\{x,y,z\}$, let $P_a\in\{X,Y,Z\}$ denote the corresponding Pauli, and let $U_a$ be any unitary such that $U_a P_a U_a^\dagger = Z$ (with $U_z=I$).

    In the SP axis-isolation circuit for axis $a$ (randomly apply $P_a^b$ with $b\sim\mathrm{Unif}\{0,1\}$, then apply $U_a$, then perform the learned $Z$-twirled MCM), the \emph{average} pre-measurement state is
    \begin{equation}
        \label{eq:spam-pre-measurement-state-no-axis-mixing}
        \rho_a^{\mathrm{in}}
        \;=\;
        \frac12\Bigl(I + (\lambda_z \lambda_a)\, r_a\, Z\Bigr),
        \qquad
        \lambda_a \in \{\lambda_x,\lambda_y,\lambda_z\}\ \text{according to } a.
    \end{equation}
    Consequently, any inferred ``$a$-component'' extracted from this circuit depends \emph{only} on $r_a$ via a multiplicative rescaling factor $\lambda_z\lambda_a$, with \emph{no} additive bias and \emph{no} mixing from the other Bloch components.
    Equivalently, in the Bloch-vector expression \eqref{eq:spam-ra-eff-bloch-map}, all axis-mixing coefficients and the state-independent nonunital offset vanish under this assumption.
\end{proposition}

\begin{proof}
    By linearity and the gate-independence in \eqref{eq:spam-gate-independent-noise-model}, averaging over the random bit $b\in\{0,1\}$ in the Pauli-axis dephasing step yields
    \[
        \frac12\Bigl(\mathcal{N}(\rho_{\mathrm{SP}}) + \mathcal{N}(P_a \rho_{\mathrm{SP}} P_a)\Bigr)
        \;=\;
        \mathcal{N}\!\left(\frac12(\rho_{\mathrm{SP}} + P_a \rho_{\mathrm{SP}} P_a)\right)
        \;=\;
        \mathcal{N}\!\bigl(\mathcal{D}_a(\rho_{\mathrm{SP}})\bigr).
    \]
    Using SM Sec.~\ref{app:randomization:axis-dephasing}, $\mathcal{D}_a(\rho_{\mathrm{SP}})=\frac12(I+r_a P_a)$.
    Because $\mathcal{N}$ is unital and Pauli-diagonal,
    \[
        \mathcal{N}\!\left(\tfrac12(I+r_a P_a)\right)
        \;=\;
        \tfrac12\bigl(\mathcal{N}(I) + r_a \mathcal{N}(P_a)\bigr)
        \;=\;
        \tfrac12\bigl(I + \lambda_a r_a P_a\bigr).
    \]
    Next apply the basis rotation layer $U_a$ with the same post-unitary noise channel.
    \[
        \rho_a^{\mathrm{in}}
        \;=\;
        \mathcal{N}\!\left(U_a \,\tfrac12(I+\lambda_a r_a P_a)\, U_a^\dagger\right)
        \;=\;
        \tfrac12\Bigl(I + \lambda_a r_a\,\mathcal{N}(U_a P_a U_a^\dagger)\Bigr)
        \;=\;
        \tfrac12\bigl(I + \lambda_a r_a\,\mathcal{N}(Z)\bigr).
    \]
    Finally $\mathcal{N}(Z)=\lambda_z Z$, giving \eqref{eq:spam-pre-measurement-state-no-axis-mixing}.
    The stated ``no mixing/no bias'' conclusions follow immediately because the pre-measurement Bloch vector points purely along $Z$ with magnitude $(\lambda_z\lambda_a)r_a$.
\end{proof}

\paragraph{Operational reporting convention and when gate-error correction is required.}
\label{rem:spam-impact-extra-1q-gates:convention}
In this work we adopt the following convention.

\emph{(i) Operational SP components.}
The reported SP components are taken to be the operational quantities $r_a^{\mathrm{eff}}$ defined in \eqref{eq:spam-ra-eff-def}, i.e.\ the Bloch components of the effective pre-measurement states produced by the same axis-isolation-and-rotation layers that are present in the SP protocol.
This is the appropriate notion whenever those same compilation layers (or statistically equivalent ones, e.g.\ under randomized compilation) appear throughout the experimental workflow.

\emph{(ii) Separation of SP from single-qubit gate error.}
If one instead wishes to report the ``bare'' Bloch vector $(r_x,r_y,r_z)$ of $\rho_{\mathrm{SP}}$ \emph{prior} to the added layers, then \eqref{eq:spam-ra-eff-bloch-map} shows that additional information is needed to identify (or bound) the coefficients of $T_{\mathcal{F}_a}$ and $\bm{c}_{\mathcal{F}_a}$.
In general this requires an independent characterization of the relevant single-qubit layers (e.g.\ via dedicated single-qubit benchmarking) or additional circuits that provide constraints sufficient to estimate the axis-mixing and bias terms.
Absent such characterization, reporting $r_a^{\mathrm{eff}}$ avoids introducing an uncontrolled modeling assumption.

\emph{(iii) Interaction with gauge bands.}
All quantities above should be understood at fixed gauge parameter $t$ for the learned MCM model.
When gauge non-identifiability is present, the recommended output is the gauge band of $r_a^{\mathrm{eff}}(t)$ over $t\in\mathcal{T}_{\mathrm{phys}}^{(+)}$ (SM Sec.~\ref{app:physicality} and Remark~\ref{rem:recommended-reporting-convention}), with any additional single-qubit gate characterization uncertainty (if used) treated as a separate systematic.

\paragraph{Noisy-$\ket{0}$ preparation error used in validation.}
\label{rem:spam-epsilon-sp-validation}
The preparation error $\epsilon_{\mathrm{SP}}$ is inferred after MCM instrument learning by combining a noisy-$\ket{0}$ readout probability with the learned readout matrix $A(t)$.
The validation predictors in SM Sec.~\ref{subsec:validation-fig2-predictors} use $\epsilon_{\mathrm{SP}}$ to write the noisy hardware preparation of $\ket{0}$ as
\begin{equation}
    \label{eq:spam-epsilon-sp-mixture}
    \rho_{\ket{0},\mathrm{noisy}}
    =
    (1-\epsilon_{\mathrm{SP}})\ket{0}\!\bra{0}
    +
    \epsilon_{\mathrm{SP}}\ket{1}\!\bra{1}.
\end{equation}
For a fixed allowed value of $t$, let $p_{\mathrm{prep}}(0)$ be the observed probability that the MCM reports $0$ after this noisy preparation.
Then
\[
    p_{\mathrm{prep}}(0)
    =
    A_{0\mid 0}(t)\bigl[1-\epsilon_{\mathrm{SP}}(t)\bigr]
    +
    A_{0\mid 1}(t)\epsilon_{\mathrm{SP}}(t).
\]
Solving this equation gives $\epsilon_{\mathrm{SP}}(t)$ for that value of $t$.
At each $t$, this equation uses the two readout probabilities in $A(t)$ to determine $\epsilon_{\mathrm{SP}}(t)$.
The four MCM rates are evaluated from the same $M^{(o)}(t)$, so all five quantities share one gauge-dependent family.
Requiring $0\leq\epsilon_{\mathrm{SP}}(t)\leq1$ can further restrict the allowed $t$ set.
For the hardware analysis, define the joint allowed set
\[
    \mathcal{T}_{\mathrm{hw}}^{(+)}
    \coloneqq
    \bigl\{t\in\mathcal{T}_{\mathrm{phys}}^{(+)}:
    0\leq\epsilon_{\mathrm{SP}}(t)\leq1\bigr\}.
\]
The preparation-error band and the four MCM-rate bands in \prlFigResultsRef{}(a) are all evaluated over $\mathcal{T}_{\mathrm{hw}}^{(+)}$.

In Bloch-vector notation the noisy preparation has $Z$ component $r_z=1-2\epsilon_{\mathrm{SP}}$, so
\begin{equation}
    \label{eq:epsilon-sp-from-rz}
    \epsilon_{\mathrm{SP}}=\frac{1-r_z}{2}.
\end{equation}
When $\epsilon_{\mathrm{SP}}$ is inferred using the full SP protocol above, $r_z$ in \eqref{eq:epsilon-sp-from-rz} should be read as the operational component $r_z^{\mathrm{eff}}$ unless the additional axis-isolation layers are ideal or independently characterized.

\subsection{Reset protocol evaluation}
\label{subsec:reset-application}

Learning the $Z$-twirled single-qubit mid-circuit measurement (MCM) instrument matrices $\{M^{(o)}\}_{o\in\{0,1\}}$ enables \emph{a priori} performance analysis of reset protocols built from MCMs and classical feedforward.
Here we compare two natural protocols that aim to prepare the computational ground state $\ket{0}$.

\smallskip

\noindent \textbf{Protocol~A (all-zero heralded postselection/retry).}
Perform a chain of MCMs and \emph{accept} an attempt iff all recorded readouts are $0$; if readout bit $1$ appears, the attempt is discarded and one restarts from a fresh qubit preparation.

\smallskip

\noindent \textbf{Protocol~B (deterministic measurement-based reset).}
Perform a measurement-based reset block consisting of one MCM followed by a feedforward conditional $X$ gate that is applied iff the preceding MCM readout is $1$.
We assume throughout this subsection that the feedforward $X$ gate is perfect.
If an externally characterized imperfect-$X$ transition matrix is available, it can be substituted in the same way as the separately characterized $X_{p_X}$ matrix used for the validation predictor in SM Sec.~\ref{subsec:validation-fig2-predictors}.

\smallskip

Throughout this subsection, reset quality is quantified by the standard fidelity with respect to $\ket{0}$,
\begin{equation}
    F(\rho)\coloneqq \bra{0}\rho\ket{0}.
\end{equation}
Because the first operation in both protocols is an MCM, the relevant initial datum is the pre-MCM $Z$-basis population vector
\begin{equation}
    \bm{\pi}=(\pi_0,\pi_1)^{\mathsf T},
    \qquad
    \pi_0=F(\rho_{\mathrm{in}}),
    \qquad
    \pi_1=1-\pi_0,
\end{equation}
consistent with the convention of Lemma~\ref{lem:matrix-product-string-probabilities}.

\paragraph{Protocol A. All-zero heralded postselection/retry.}
\label{subsubsec:reset-application:protA}

Fix an integer $k\ge 1$.
One \emph{attempt} of Protocol~A applies at most $k$ consecutive MCMs to the target qubit.
The attempt is accepted iff the first $k$ readouts are all $0$, i.e.\ iff the recorded string is $0^k$.
If readout bit $1$ appears before step $k$, the attempt is terminated, discarded, and restarted from a fresh input state.
Thus Protocol~A is a \emph{heralded} reset protocol.
It does not deterministically output a reset qubit on every attempt.

In the $Z$-twirled instrument model, conditioning on a fixed readout string corresponds to multiplying the associated instrument matrices and renormalizing.
For the all-zero readout string $0^k$, the unnormalized post-measurement population vector is $(M^{(0)})^k\bm{\pi}$, hence
\begin{equation}
    S_k
    \;\coloneqq\;
    p(0^k)
    \;=\;
    \bm{1}^{\mathsf T}(M^{(0)})^k\bm{\pi},
    \label{eq:reset-protA-yield}
\end{equation}
which is the \emph{per-attempt success probability} (equivalently, the yield) of Protocol~A.
Conditional on acceptance, the reset fidelity is
\begin{equation}
    F^{(k)}_{\mathrm A}(\bm{\pi})
    \;\coloneqq\;
    F\!\left(\rho_{\mathrm{out}}\mid 0^k\right)
    \;=\;
    \frac{\bigl[(M^{(0)})^k\bm{\pi}\bigr]_0}{\bm{1}^{\mathsf T}(M^{(0)})^k\bm{\pi}}
    \;=\;
    \frac{\bigl[(M^{(0)})^k\bm{\pi}\bigr]_0}{S_k},
    \label{eq:reset-protA-fidelity-def}
\end{equation}
where $[\cdot]_0$ denotes the first component (the post-measurement population of $\ket{0}$).
Equation~\eqref{eq:reset-protA-fidelity-def} makes explicit the core tradeoff of Protocol~A: one exchanges throughput for conditional fidelity.

To connect the accepted-readout-string analysis above to the operational ``discard and retry'' procedure, note that $S_k$ is the success probability of a single attempt.
Assuming independent retries, the expected number of attempts required to obtain one accepted output is
\begin{equation}
    N_{\mathrm{att},\mathrm A}^{(k)}
    \;=\;
    \frac{1}{S_k}.
    \label{eq:reset-protA-expected-attempts}
\end{equation}
If failed attempts are terminated immediately at the first occurrence of readout bit $1$, then the expected number of MCM calls consumed in a single attempt is
\begin{equation}
    C_{\mathrm{att},\mathrm A}^{(k)}
    \;=\;
    \sum_{j=0}^{k-1} p(0^j)
    \;=\;
    \sum_{j=0}^{k-1} S_j,
    \qquad S_0\coloneqq 1,
    \label{eq:reset-protA-expected-mcm-per-attempt}
\end{equation}
since the $(j+1)$-th MCM is reached iff the first $j$ readouts are all $0$.
Therefore the expected number of MCM calls required per \emph{accepted} output is
\begin{equation}
    C_{\mathrm{succ},\mathrm A}^{(k)}
    \;=\;
    \frac{C_{\mathrm{att},\mathrm A}^{(k)}}{S_k}
    \;=\;
    \frac{\sum_{j=0}^{k-1} S_j}{S_k}.
    \label{eq:reset-protA-expected-mcm-per-success}
\end{equation}

It is sometimes convenient to describe the accepted-state dynamics in terms of the scalar fidelity iterate $F_n\coloneqq F^{(n)}_{\mathrm A}(\bm{\pi})$ with $F_0=\pi_0$.
Writing
\begin{equation}
    M^{(0)}
    \;=\;
    \begin{pmatrix}
        m_{00} & m_{01} \\
        m_{10} & m_{11}
    \end{pmatrix},
    \qquad
    m_{s's}\coloneqq M^{(0)}_{s',s},
\end{equation}
a direct computation from Eq.~\eqref{eq:reset-protA-fidelity-def} yields the scalar recurrence
\begin{equation}
    F_{n+1}
    \;=\;
    f(F_n)
    \;\coloneqq\;
    \frac{(m_{00}-m_{01})F_n + m_{01}}{(A_0-A_1)F_n + A_1},
    \qquad
    A_s\coloneqq m_{0s}+m_{1s}=\Pr(o{=}0\mid s),
    \label{eq:reset-protA-mobius}
\end{equation}
which is a Möbius transformation mapping $[0,1]$ to itself.
The fixed points of $f$ are the roots of a quadratic polynomial, and the asymptotic accepted-state fidelity is determined by the Perron--Frobenius eigenvector of the nonnegative matrix $M^{(0)}$.

\begin{lemma}[Asymptotics of Protocol~A]
    \label{lem:reset-protA-asymptotics}
    Assume $M^{(0)}$ has strictly positive entries (equivalently, is a primitive nonnegative matrix).
    Let $\lambda_0>0$ be its spectral radius and let $\bm{v}_0\succ 0$ be the corresponding right Perron eigenvector, $M^{(0)}\bm{v}_0=\lambda_0\bm{v}_0$.
    Then for any initial population vector $\bm{\pi}$ with $\pi_0,\pi_1>0$,
    \begin{align}
        F^{(k)}_{\mathrm A}(\bm{\pi})
        \; & \longrightarrow\;
        \frac{(\bm{v}_0)_0}{(\bm{v}_0)_0+(\bm{v}_0)_1},
        \label{eq:reset-protA-limit-fid}
        \\
        S_k
        \; & =\;
        \bm{1}^{\mathsf T}(M^{(0)})^k\bm{\pi}
        \;=\;
        \Theta(\lambda_0^k),
        \label{eq:reset-protA-limit-yield}
    \end{align}
    as $k\to\infty$.
    In particular, unless $\lambda_0$ is extremely close to $1$, the per-attempt success probability $S_k$ decays exponentially with $k$.
\end{lemma}

\begin{proof}
    Since $M^{(0)}$ is a $2\times 2$ strictly positive matrix, the Perron--Frobenius theorem guarantees a simple dominant eigenvalue $\lambda_0>0$ with strictly positive right eigenvector $\bm{v}_0\succ 0$ and strictly positive left eigenvector $\bm{w}_0^{\mathsf T}$, together with a second eigenvalue $\lambda_1$ satisfying $|\lambda_1|<\lambda_0$.
    Expanding the initial vector in the eigenbasis,
    \[
        \bm{\pi}=c_0\bm{v}_0+c_1\bm{v}_1,
        \qquad
        c_0=\frac{\bm{w}_0^{\mathsf T}\bm{\pi}}{\bm{w}_0^{\mathsf T}\bm{v}_0},
    \]
    one obtains $(M^{(0)})^k\bm{\pi}=c_0\lambda_0^k\bm{v}_0+c_1\lambda_1^k\bm{v}_1$.
    Positivity of the entries of $\bm{\pi}$ and $\bm{w}_0$ ensures $c_0>0$.
    The accepted-state fidelity is
    \[
        F^{(k)}_{\mathrm A}(\bm{\pi})
        =
        \frac{\bigl[(M^{(0)})^k\bm{\pi}\bigr]_0}{\bm{1}^{\mathsf T}(M^{(0)})^k\bm{\pi}}
        =
        \frac{c_0\lambda_0^k(\bm{v}_0)_0+c_1\lambda_1^k(\bm{v}_1)_0}
        {c_0\lambda_0^k\bm{1}^{\mathsf T}\bm{v}_0+c_1\lambda_1^k\bm{1}^{\mathsf T}\bm{v}_1}.
    \]
    Dividing numerator and denominator by $c_0\lambda_0^k$ and using $(\lambda_1/\lambda_0)^k\to 0$ gives $F^{(k)}_{\mathrm A}\to (\bm{v}_0)_0/\bm{1}^{\mathsf T}\bm{v}_0$, which is \eqref{eq:reset-protA-limit-fid}.
    For the success probability, $S_k=\bm{1}^{\mathsf T}(M^{(0)})^k\bm{\pi}
    =c_0\lambda_0^k(\bm{1}^{\mathsf T}\bm{v}_0)+O(|\lambda_1|^k)
    =\Theta(\lambda_0^k)$.
\end{proof}

\paragraph{Protocol B. Deterministic measurement-based reset with conditional $X$.}
\label{subsubsec:reset-application:protB}

Protocol~B iterates a \emph{measurement-based reset block} that performs one MCM and then applies a single-qubit $X$ gate if the readout is $1$, or the identity otherwise.
Under our standing assumption that the feedforward $X$ gate is perfect, Protocol~B produces an output on every block and therefore has \emph{unit yield}.

At the level of the learned instrument matrices, one reset block is described by the column-stochastic $2\times 2$ matrix $R_{\mathrm{MB}}$ defined entrywise by
\begin{equation}
    (R_{\mathrm{MB}})_{s',s}
    \;\coloneqq\;
    M^{(0)}_{s',s} + M^{(1)}_{1-s',s},
    \qquad
    s,s'\in\{0,1\}.
    \label{eq:reset-Rmb-def}
\end{equation}
After $k$ blocks, the population vector is
\begin{equation}
    \bm{\pi}^{(k)}_{\mathrm B}
    \;=\;
    (R_{\mathrm{MB}})^k\bm{\pi},
    \qquad
    F^{(k)}_{\mathrm B}(\bm{\pi})
    \;=\;
    \bigl[\bm{\pi}^{(k)}_{\mathrm B}\bigr]_0.
    \label{eq:reset-protB-evolution}
\end{equation}
In particular, for an initially excited input $\bm{\pi}=(0,1)^{\mathsf T}$, a single block achieves fidelity
\begin{equation}
    F^{(1)}_{\mathrm B}\bigl((0,1)^{\mathsf T}\bigr)
    \;=\;
    (R_{\mathrm{MB}})_{0,1}
    \;=\;
    M^{(0)}_{0,1}+M^{(1)}_{1,1}.
    \label{eq:reset-protB-one-shot-excited}
\end{equation}

Writing the two transition probabilities
\begin{equation}
    a \;\coloneqq\; (R_{\mathrm{MB}})_{1,0},
    \qquad
    b \;\coloneqq\; (R_{\mathrm{MB}})_{0,1},
    \label{eq:reset-protB-ab}
\end{equation}
This gives the following closed scalar recurrence for the reset fidelity.
\begin{equation}
    F^{(k+1)}_{\mathrm B}
    \;=\;
    b + (1-a-b)F^{(k)}_{\mathrm B},
    \qquad
    F^{(0)}_{\mathrm B}=\pi_0.
    \label{eq:reset-protB-linear-rec}
\end{equation}
If $a+b>0$ (the generic case), the chain is ergodic and converges to a unique stationary point
\begin{equation}
    F^{(\infty)}_{\mathrm B}
    \;=\;
    \lim_{k\to\infty}F^{(k)}_{\mathrm B}
    \;=\;
    \frac{b}{a+b},
    \qquad
    F^{(k)}_{\mathrm B}
    \;=\;
    F^{(\infty)}_{\mathrm B}
    + \bigl(\pi_0-F^{(\infty)}_{\mathrm B}\bigr)(1-a-b)^k.
    \label{eq:reset-protB-solution}
\end{equation}
The convergence rate is governed by the second eigenvalue $1-a-b$ of $R_{\mathrm{MB}}$.

\paragraph{Comparison. Heralded vs deterministic reset.}
\label{subsubsec:reset-application:comparison}

The two protocols optimize different operational objectives.

\begin{enumerate}[label=\textup{(\roman*)}, leftmargin=2.2em]
    \item \textbf{Protocol~A is heralded.}
    Its natural figures of merit are the accepted-state fidelity $F^{(k)}_{\mathrm A}$, the per-attempt success probability $S_k$, and the retry overhead metrics \eqref{eq:reset-protA-expected-attempts}--\eqref{eq:reset-protA-expected-mcm-per-success}.
    It can achieve a high \emph{conditional} fidelity on accepted outputs, but only at the cost of discarding failed attempts.

    \item \textbf{Protocol~B is deterministic.}
    Its natural figure of merit is the output fidelity $F^{(k)}_{\mathrm B}$, with unit yield by construction.
    It is therefore the relevant reset operation when an online reset must always produce an output qubit rather than a heralded success/failure flag.
\end{enumerate}

Consequently, the appropriate ``winner'' depends on the operational constraint.
If reset must be deterministic---for example, inside a live feedback loop or a repeated syndrome-extraction cycle---then Protocol~B is the relevant operation.
If one can afford postselection and retry overhead---for example, in an offline heralded-state-preparation setting---then Protocol~A can trade throughput for higher accepted-state fidelity.

In all cases, once $\{M^{(o)}\}$ has been learned, the quantities $S_k$, $F^{(k)}_{\mathrm A}$, $N_{\mathrm{att},\mathrm A}^{(k)}$, $C_{\mathrm{succ},\mathrm A}^{(k)}$, and $F^{(k)}_{\mathrm B}$ are straightforward to evaluate numerically for any chosen $k$.
When the learned instrument is identified only up to a gauge (SM Sec.~\ref{app:gauge}), both Protocol~A and Protocol~B probe post-measurement state information and therefore depend on the chosen representative.
In that setting, we recommend reporting \emph{gauge bands} for the reset metrics by evaluating the above expressions over $\mathcal{T}_{\mathrm{phys}}^{(+)}$ (SM Sec.~\ref{app:physicality}), following the reporting convention of Remark~\ref{rem:recommended-reporting-convention}.

\section{Finite-shot uncertainty treatment}
\label{app:stats}

We model finite-shot uncertainty in the measured readout-string distributions as i.i.d.\ multinomial shot noise.
Additional non-i.i.d.\ effects such as drift and non-Markovian memory are treated instead as \emph{model violations} and handled by the diagnostics and cross-validation procedures of SM Secs.~\ref{app:invariants} and~\ref{app:validation}.

\paragraph{Shot-noise model for readout-string distributions.}
\label{subsec:stats-shot-noise}

\begin{assumption}[Multinomial sampling of outcomes]
    \label{assump:stats-multinomial}
    Fix a circuit family (e.g., a length-$L$ repeated-MCM circuit) whose single-shot record is a bit string $w\in\{0,1\}^L$, distributed according to the true probabilities $\{p(w)\}_{w\in\{0,1\}^L}$.
    Assume that repeated shots of this circuit are independent and identically distributed.

    For a total of $N$ shots, let $N(w)$ be the number of occurrences of outcome $w$, and define the empirical frequency $\widehat{p}(w)=N(w)/N$ as in SM Sec.~\ref{app:prob-dict}.
\end{assumption}

\begin{lemma}[Second moments of empirical frequencies and marginals]
    \label{lem:stats-frequency-second-moments}
    Under Assumption~\ref{assump:stats-multinomial} one has, for all $w,w'\in\{0,1\}^L$,
    \begin{align}
        \label{eq:stats-mean}
        \mathbb{E}\!\big[\widehat{p}(w)\big]
        \; & =\;
        p(w),
        \\[2pt]
        \label{eq:stats-cov-cases}
        \mathbb{E}\!\Big[\big(\widehat{p}(w)-p(w)\big)\big(\widehat{p}(w')-p(w')\big)\Big]
        \; & =\;
        \begin{cases}
            \dfrac{1}{N}\,p(w)\bigl(1-p(w)\bigr), & w=w',     \\[8pt]
            -\dfrac{1}{N}\,p(w)\,p(w'),           & w\neq w'.
        \end{cases}
    \end{align}

    More generally, for any two events $A,B\subseteq\{0,1\}^L$ define the marginal probabilities $p(A)=\sum_{w\in A}p(w)$ and empirical estimates $\widehat{p}(A)=\sum_{w\in A}\widehat{p}(w)$.
    Then
    \begin{equation}
        \label{eq:stats-marginal-cov}
        \mathbb{E}\!\Big[\big(\widehat{p}(A)-p(A)\big)\big(\widehat{p}(B)-p(B)\big)\Big]
        \;=\;
        \frac{1}{N}\,\Bigl(p(A\cap B)-p(A)\,p(B)\Bigr).
    \end{equation}
\end{lemma}

\begin{proof}
    The vector $(N(w))_{w\in\{0,1\}^L}$ is multinomial with parameters $(N,\{p(w)\})$.
    The mean \eqref{eq:stats-mean} is immediate.
    For the second moments, one has $\mathbb{E}[N(w)]=Np(w)$, $\mathbb{E}\big[(N(w)-Np(w))^2\big]=Np(w)(1-p(w))$, and for $w\neq w'$, $\mathbb{E}\big[(N(w)-Np(w))(N(w')-Np(w'))\big]=-Np(w)p(w')$.
    Dividing by $N^2$ yields \eqref{eq:stats-cov-cases}.

    For \eqref{eq:stats-marginal-cov}, expand $\widehat{p}(A)$ and $\widehat{p}(B)$ as sums of the $\widehat{p}(w)$ and use \eqref{eq:stats-cov-cases}, noting that the diagonal terms contribute $\frac{1}{N}\sum_{w\in A\cap B}p(w)$ and the off-diagonal terms contribute $-\frac{1}{N}\sum_{w\in A}\sum_{w'\in B}p(w)p(w')$.
\end{proof}

\paragraph{Nested event covariances for $p(0^n)$ and $p(1^n)$.}
\label{rem:stats-prefix-cov}
The all-zero/all-one readout-string probabilities $S_n=p(0^n)$ and $T_n=p(1^n)$ (SM Sec.~\ref{app:invariants}) correspond to nested events in the length-$L$ outcome space.
For example, for $m\ge n\ge 1$, the event ``outcome begins with $0^m$'' is a subset of ``outcome begins with $0^n$'', hence $A_m\subseteq A_n$ and $A_m\cap A_n=A_m$.
Equation~\eqref{eq:stats-marginal-cov} therefore yields the closed form
\begin{equation}
    \label{eq:stats-nested-prefix-cov}
    \mathbb{E}\!\Big[\big(\widehat{p}(0^n)-S_n\big)\big(\widehat{p}(0^m)-S_m\big)\Big]
    \;=\;
    \frac{1}{N}\,\Bigl(S_m - S_n S_m\Bigr),
    \qquad m\ge n\ge 1,
\end{equation}
and analogously for the $T_n$ sequence.
These correlations should be retained when propagating uncertainty to invariants and derived parameters.

\paragraph{Propagation to invariants and reconstructed parameters.}
\label{subsec:stats-uncertainty}

Many quantities reported in this work are obtained by applying a smooth map to a finite set of estimated probabilities $\widehat{p}(w)$ and/or marginals obtained by summing over later readouts.
A generic first-order propagation rule follows from a multivariate Taylor expansion and Lemma~\ref{lem:stats-frequency-second-moments}.

\begin{proposition}[First-order uncertainty propagation for smooth functionals]
    \label{prop:stats-first-order-propagation}
    Let $W=\{w_1,\dots,w_d\}$ be a finite set of readout bit strings (possibly from multiple circuit families), and let $f:\mathbb{R}^d\to\mathbb{R}$ be continuously differentiable in a neighborhood of the true probability vector $\bigl(p(w_1),\dots,p(w_d)\bigr)$.

    Assume that the vector $\bigl(\widehat{p}(w_1),\dots,\widehat{p}(w_d)\bigr)$ is obtained from a single multinomial experiment with $N$ shots as in Assumption~\ref{assump:stats-multinomial}.
    Then, to first order in $1/N$,
    \begin{align}
        \label{eq:stats-first-order-var}
         & \mathbb{E}\!\Big[\big(f\bigl(\widehat{p}(w_1),\dots,\widehat{p}(w_d)\bigr)-f\bigl(p(w_1),\dots,p(w_d)\bigr)\big)^2\Big] \nonumber \\
         & \approx\;
        \sum_{i,j=1}^{d}
        \frac{\partial f}{\partial p(w_i)}\Bigl(p(w_1),\dots,p(w_d)\Bigr)\,
        \frac{\partial f}{\partial p(w_j)}\Bigl(p(w_1),\dots,p(w_d)\Bigr)\,
        \mathbb{E}\!\Big[\big(\widehat{p}(w_i)-p(w_i)\big)\big(\widehat{p}(w_j)-p(w_j)\big)\Big],
    \end{align}
    where the second moments on the right-hand side are given explicitly by \eqref{eq:stats-cov-cases}.

    If the probabilities entering $f$ are estimated from multiple independent circuit families, the same formula holds with the second moments taken blockwise (i.e., cross-family second moments are zero under independence).
\end{proposition}

\paragraph{Application to the $\Tr/\det$ estimators and feasibility boundaries.}
\label{rem:stats-trdet-feasibility}
For the $\Tr/\det$ estimators in Corollary~\ref{cor:tr-det-from-length3}, the map $f$ is a rational function of the three initial-string marginals $(p(0),p(00),p(000))$ (and similarly for the all-one strings).
Equation~\eqref{eq:stats-first-order-var} applies directly by treating $(\widehat{p}(0),\widehat{p}(00),\widehat{p}(000))$ as the random input vector and using the nested-event covariance formulas \eqref{eq:stats-nested-prefix-cov}.

In particular, writing $U=S_1S_2-S_3$ and $D=S_1^2-S_2$, the estimator of $\Tr(M^{(0)})$ is $U/D$.
A direct differentiation gives
\begin{equation}
    \label{eq:stats-derivatives-tr}
    \frac{\partial}{\partial S_1}\!\left(\frac{U}{D}\right)
    =
    \frac{S_2 D - 2S_1 U}{D^2},
    \qquad
    \frac{\partial}{\partial S_2}\!\left(\frac{U}{D}\right)
    =
    \frac{S_1 D + U}{D^2},
    \qquad
    \frac{\partial}{\partial S_3}\!\left(\frac{U}{D}\right)
    =
    -\frac{1}{D},
\end{equation}
which may be substituted into \eqref{eq:stats-first-order-var}.
The appearance of $D^{-1}$ and $D^{-2}$ makes explicit the instability mechanism discussed in Subsec.~\ref{subsec:invariants-failure-modes}.

More generally, feasibility boundaries for the gauge set $\mathcal{T}_\mathrm{phys}$ are defined by root conditions of (nonnegative) entrywise constraints after gauge transformation (SM Sec.~\ref{app:physicality}).
When these boundary points are smooth functions of a finite set of probabilities, one may apply Proposition~\ref{prop:stats-first-order-propagation} locally; in practice, bootstrap resampling is often simpler and more robust within SM Sec.~\ref{app:stats}.

\paragraph{Bootstrap workflow for uncertainty bands.}
\label{subsec:stats-bootstrap}

In the regimes of interest here, quantities of practical relevance are often \emph{nonlinear} functions of the readout-string distribution (ratios, roots, interval endpoints, and gauge extrema).
A bootstrap approach avoids fragile analytic linearization, handles disconnected physically allowed gauge sets, and naturally integrates with gauge-band reporting.

\begin{protocol}[Bootstrap procedure for uncertainty bands]
    \label{prot:stats-bootstrap-uncertainty-bands}
    Given raw shot counts for the circuit families used to construct a measured readout-string distribution, proceed as follows.

    \begin{enumerate}[label=(\arabic*), leftmargin=2.2em]
        \item Form the empirical readout-string distribution $\widehat{p}$ by normalizing counts.

        \item Generate many resampled readout-string distributions by drawing simulated counts from the multinomial model with parameters $(N,\widehat{p})$ for each circuit family independently, and re-normalizing to obtain a resampled distribution.

        \item For each resampled distribution, rerun the full post-processing pipeline used in the main analysis by computing invariants, reconstructing representative candidates (SM Sec.~\ref{app:reconstruction}), and determining the physically allowed gauge set $\mathcal{T}_\mathrm{phys}$ (SM Sec.~\ref{app:physicality}).

        \item For each derived quantity of interest that is gauge-dependent, compute its gauge band over $t\in\mathcal{T}_\mathrm{phys}^{(+)}$ by gauge-transforming the learned MCM model.
        Externally characterized inputs, final-state marginalizations, or interleaved maps are held fixed, so any result using them is conditional on those external characterizations (Remark~\ref{rem:recommended-reporting-convention}).

        \item Summarize the resulting empirical distribution of bands (e.g., by reporting percentiles), and report the fraction of resamples in which feasibility fails as a diagnostic of statistical proximity to the physicality boundary.
    \end{enumerate}
\end{protocol}

\paragraph{Statistical versus model uncertainty.}
\label{rem:stats-stat-vs-model}
Bootstrap bands quantify \emph{finite-shot} uncertainty under Assumption~\ref{assump:stats-multinomial}.
If repeated independent data acquisitions produce band variations significantly larger than bootstrap predictions, this indicates model violations (e.g., drift, non-stationarity, or context dependence) rather than shot noise.
In that setting, longer-readout-string overcomplete constraints (Proposition~\ref{prop:overcomplete-long-string-constraints}) provide a convenient internal diagnostic and can guide targeted data collection to isolate the dominant violation.

\section{Finite-shot and model-violation simulations}
\label{app:synthetic-study}

We use numerical simulations to answer three questions.
First, how quickly do the reconstruction errors decrease as the number of shots increases?
Second, what happens when a denominator in the closed-form formulas is close to zero?
Third, can longer readout strings reveal departures from the stationary two-state model?
For each simulated dataset, we draw readout counts from the probabilities defined by the chosen instrument.

\subsection{Finite-shot design and quantities evaluated}
\label{subsec:synthetic-finite-shot}

Write $D_0=S_1^2-S_2$ and $D_1=T_1^2-T_2$ for the denominators in Corollary~\ref{cor:tr-det-from-length3}.
We first test six physical instruments.
Two have each of the basis-averaged joint success probabilities $F_{\mathrm{MCM}}^{\mathrm{pop}}\in\{0.70,0.95,0.99\}$.
Each instrument has a nonempty physically allowed gauge set and satisfies $\min(|D_0|,|D_1|)\ge0.02$.

We then test additional instrument families in which either $D_0$ or $D_1$ is brought progressively closer to zero, as shown in Fig.~\ref{fig:synthetic-study}(b).
These cases show how a small denominator can amplify ordinary fluctuations in the measured probabilities.

For every instrument, we simulate length-$3$ data with $10^3$ to $10^6$ shots per Pauli input setting.
Figure~\ref{fig:synthetic-study}(c) compares the closed-form length-$3$ reconstruction with a second estimate that uses strings through length $10$ and accounts for correlations among the measured probabilities.
We call a reconstruction successful only when all reported estimates are finite and at least one physically valid instrument is found.
The failure probability counts every simulated run.
The error summaries include only successful runs because a failed run has no physical estimate whose error can be evaluated.
For each reconstructed invariant, the bias is the mean signed error, the variance is the sample variance of the estimates, and the RMSE is the square root of the mean squared error.
All three quantities are evaluated over the successful runs.

\subsection{Controlled model violations and diagnostics}
\label{subsec:synthetic-mismatch}

We test three specific departures from the stationary two-state model.
In the first, only the first MCM differs from the later MCMs.
In the second, the MCM parameters alternate between two values along the sequence.
In the third, population can leak to an unobserved state outside the two-state model.
For the first two tests, the departure strength is the weight of a second physically valid MCM instrument at each affected step.
For example, a strength of $0.1$ uses 90\% of the starting instrument and 10\% of the second instrument.
In the leakage test, the strength is the probability that population in the computational space leaks out during one MCM.
We vary this quantity from zero to $0.1$.

We look for these departures in four ways.
We check the longer-string recurrence relations, test whether the readout probabilities require more than two effective states, compare predictions with separate length-$6$ data, and measure the disagreement with the length-$10$ data.
The four checks are combined into one detection test.

A run counts as a reconstruction failure whenever the closed-form method produces no physical instrument.
We still generate a comparison curve by fitting the simulated probabilities to a physically allowed stationary two-state model.
This auxiliary fit is used solely for that comparison, and the run remains in the failure count.

\subsection{Results and limitations}
\label{subsec:synthetic-results}

We evaluate the bias, variance, and root-mean-square error (RMSE) in $\Tr(M^{(0)})$ among successful reconstructions.
At $N=1000$ shots per Pauli input setting, the observed fraction of failed reconstructions across the 6 test instruments ranges from $0.069$ to $0.998$.
Among successful reconstructions, the median RMSE is $0.0112$, the median magnitude of the bias is $0.00447$, and the median variance is $0.000105$.

At $N=1000000$, the observed fraction of failed reconstructions for the same six instruments ranges from $0$ to $0.9997$.
The median RMSE among successful reconstructions is $0.000411$, with values from $0.000103$ to $0.0013$.
The median magnitude of the bias and the median variance decrease to $0.00000554$ and $0.000000169$, respectively.

For the additional instrument families in Fig.~\ref{fig:synthetic-study}(b), the most difficult near-zero-denominator case has a median denominator-to-noise ratio of only $0.65$ and a failure probability of $1.000$.
Increasing the shot count generally improves the reconstruction.
An instrument whose denominator is close to zero can require substantially more data.

We next examine the checks for departures from the model.
For a weak departure of strength $0.001$ at $N=1000$ shots per Pauli input setting, the median detection probability is $0.048$.
Depending on the departure and starting instrument, it ranges from $0.040$ to $0.072$.

For a stronger departure of strength $0.1$ at $N=1000000$, changing only the first MCM gives a median detection probability of $1.000$, with range $0.739$--$1.000$.
Alternating MCM parameters give detection probabilities from $1.000$ to $1.000$.
Leakage to an unobserved state gives detection probabilities from $1.000$ to $1.000$.
Larger departures are generally easier to detect.
A sufficiently weak departure can blend with finite-shot fluctuations.

These checks measure compatibility with the stationary two-state model.
A detected disagreement signals a departure from the assumed model.
Identifying its physical origin requires further tests, and departures below the finite-shot sensitivity can remain undetected.

Gauge bands quantify the ambiguity that remains with exact probabilities under the fitted model.
Bootstrap intervals quantify how finite counting statistics move the endpoints of those bands.
The required shot count therefore depends on the instrument, the estimator denominators, and the acceptable failure probability.

\begin{figure*}[!t]
    \centering
    
        \includegraphics[width=\textwidth,trim=0 0 0 38pt,clip]{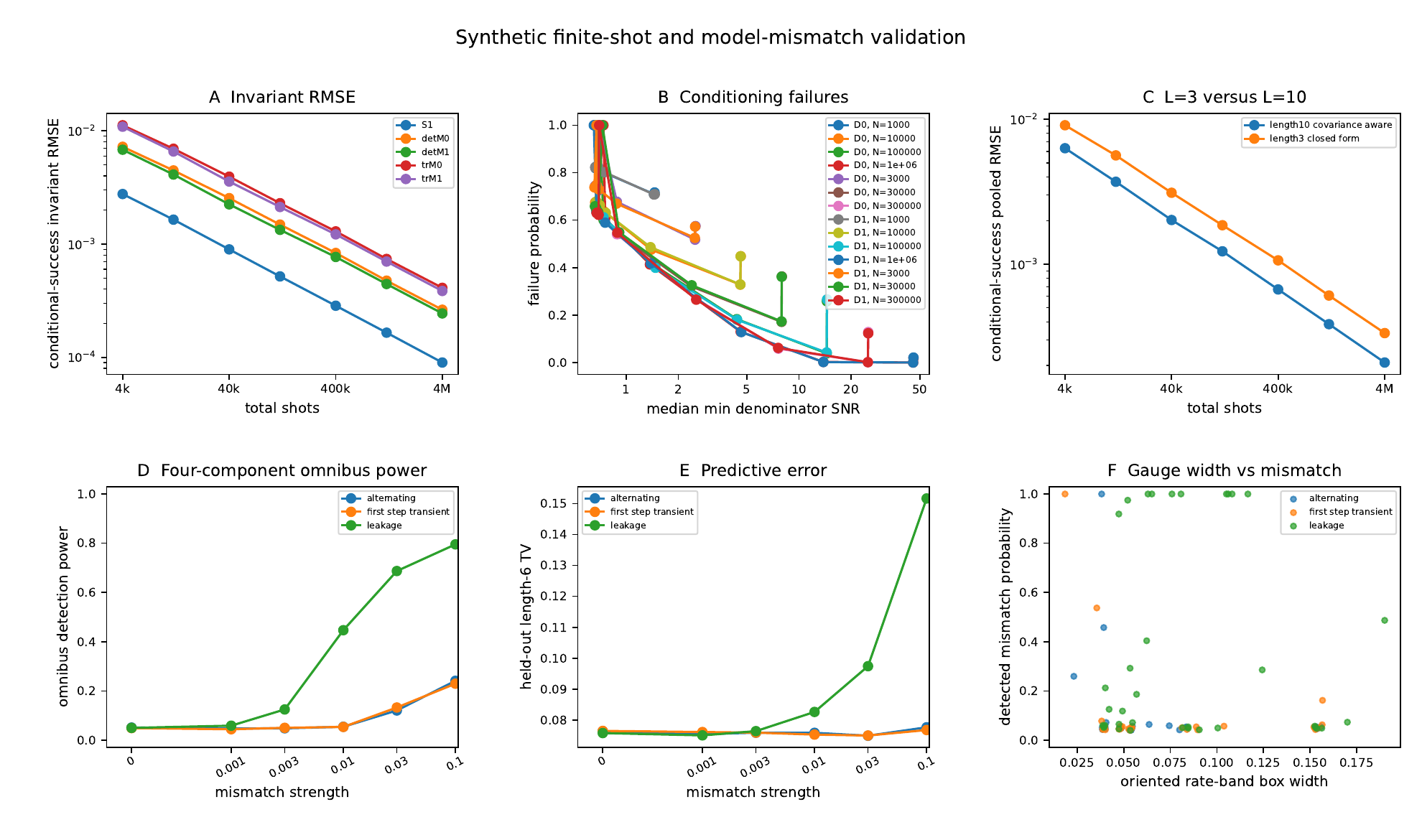}
    
    \caption{
        \textbf{Finite-shot and model-violation simulations.}
        \textbf{(a)}~Reconstruction error (RMSE) among successful reconstructions versus the total number of learning shots across the four Pauli input settings, with the corresponding failure probabilities for the six test instruments reported in Sec.~\ref{subsec:synthetic-results}.
        \textbf{(b)}~Fraction of failed reconstructions for additional instrument families with small estimator denominators, plotted versus the size of the smaller denominator relative to its typical shot-noise fluctuation.
        \textbf{(c)}~RMSE among successful reconstructions for the length-$3$ closed-form reconstruction and the length-$10$ estimate as the number of learning shots increases; both use the same definition of success.
        \textbf{(d)}~Detection probability for the three tested model departures.
        \textbf{(e)}~Prediction error on separate length-$6$ data versus departure strength.
        \textbf{(f)}~Width of the physically allowed gauge range on the $t<\tfrac12$ reporting branch.
        Gauge widths represent ambiguity within the fitted two-state model.
        Statistical uncertainty is reported separately through the bootstrap.
    }
    \label{fig:synthetic-study}
\end{figure*}

\section{IBM hardware implementation and validation}
\label{app:ibm}
\label{app:validation}

This section describes the IBM hardware study in four steps.
We first specify the devices and execution procedure.
We then describe the circuits used to learn the instrument and the separate circuits used to test its predictions.
Finally, we define the two predictors shown in \prlFigResultsRef{}(b) and the checks used to identify departures from the stationary two-state model.

\subsection{Hardware setup, selected qubits, and execution metadata}
\label{subsec:ibm-devices}

We run on three Heron-family IBM Quantum processors: \texttt{ibm\_kingston} (Heron r2), \texttt{ibm\_pittsburgh} (Heron r3), and \texttt{ibm\_boston} (Heron r3).
Each processor contains 156 transmon qubits, of which we select 20 for this study.
The devices use a heavy-hex lattice with tunable couplers~\cite{bravyiHighthresholdLowoverheadFaulttolerant2024}.

We select nonadjacent qubits to reduce measurement-induced crosstalk.
Within this constraint, the selected set spans a range of backend-reported readout-error values.

\paragraph{Execution, transpilation, and recorded metadata.}
\label{subsec:ibm-transpilation}
All circuits are submitted through the IBM Qiskit Runtime service.
We transpile them at optimization level~0 with dynamical decoupling and other error-suppression options disabled.
Because every circuit acts on only one qubit, no SWAP routing is needed.
The transpiled circuit therefore has the same gate sequence as the logical circuit, expressed in the native Heron basis $\{SX,RZ,X,\mathrm{measure}\}$.

On each processor, the 20 selected qubits are measured in parallel within the same jobs.
The learning stage uses four jobs, one for each Pauli input setting.
The validation stage uses four additional jobs.
Thus each processor requires eight jobs, independent of the number of selected qubits.

For every job, we record the backend revision, submission and completion times, software versions, transpiler settings, and the benchmark values reported when the job was submitted.
The reported Pauli-$X$ error is used in the validation predictors below.
We acquire all jobs for a processor within one hardware session, typically lasting less than one hour, to reduce changes caused by longer-timescale drift.
The checks below can reveal only departures that are larger than the corresponding finite-shot fluctuations.

\subsection{Learning and validation circuit families}
\label{subsec:ibm-circuits}

For each selected qubit, we learn the instrument from the distribution of three-bit readout strings produced by three consecutive MCMs (SM Sec.~\ref{app:invariants}, Corollary~\ref{cor:tr-det-from-length3}).
Each learning circuit acts on a single qubit and contains only single-qubit gates and MCMs.
It uses neither entangling gates nor ancilla qubits.

\emph{Learning circuits.}---The learning family consists of four Pauli input settings, one for each gate $P \in \{I, X, Y, Z\}$ (Definition~\ref{def:pauli-input-randomization-depolarizing-map}).
Within each setting, batched randomized circuit instances independently sample the pre-MCM twirling bits:

\begin{enumerate}[label=(\arabic*), leftmargin=2.2em]
    \item Apply gate $P$ to the initial state $\ket{0}$ (prepared by the hardware reset).
    \item Before each of the $L = 3$ consecutive MCMs on the same qubit, apply $Z^{c_\ell}$ with $c_\ell$ sampled independently and uniformly.
    \item Record the full length-3 readout string $w = w_1 w_2 w_3 \in \{0,1\}^3$.
\end{enumerate}

\noindent Each Pauli input setting pools $N = 10^6$ shots across its randomized instances.
The four settings therefore use $4 \times 10^6$ learning shots per qubit in total.
They are weighted equally.
When the four Pauli gates are treated as ideal, this uniform average maps the reset state to the maximally mixed input $I/2$, whose computational-basis probability vector is $\bm{\pi} = \tfrac{1}{2}\bm{1}$ (Lemma~\ref{lem:pauli-randomization-maximally-mixed}).
This is the input assumed by the algebraic reconstruction in SM Sec.~\ref{app:reconstruction}.
The learning result is therefore conditional on treating the four input Pauli gates as ideal.

The formal $Z$ twirl places the same random gate, either $I$ or $Z$, before and after each MCM.
For consecutive MCMs, the gate after one MCM is adjacent to the gate before the next MCM.
Because two adjacent $Z$ gates combine according to addition modulo two, the hardware can instead place one independently random $I/Z$ gate immediately before each MCM.

To show this explicitly, let $b_\ell$ be the random bit used for the formal pair around the $\ell$th MCM, and let $c_\ell$ be the bit for the compiled gate immediately before that MCM.
The adjacent layers combine as
\begin{equation}
    c_1=b_1,
    \qquad
    c_\ell=b_{\ell-1}\oplus b_\ell\quad (\ell=2,\ldots,L).
    \label{eq:compiled-z-twirl-binary-map}
\end{equation}
The state after the final MCM is not measured, so the last gate $Z^{b_L}$ cannot change any reported readout probability and can be omitted.
The relations in Eq.~\eqref{eq:compiled-z-twirl-binary-map} can also be inverted: $b_1=c_1$ and $b_\ell=c_\ell\oplus b_{\ell-1}$.
Thus every sequence of formal sandwich bits corresponds to exactly one sequence of pre-MCM bits, and vice versa.
A uniform choice of the $b_\ell$ bits therefore produces independent uniform $c_\ell$ bits and the same readout-string probabilities.
Physical errors in the compiled pre-MCM gates are included in the learned effective instrument.

This simplification applies when the MCMs are consecutive, as in the learning circuits above.
The $X$-interleaved validation circuits retain the original sandwich: the same random $I/Z$ choice is applied before and after each MCM.
Each intervening $X$ keeps the neighboring $Z$ layers separate.

\begin{table}[t]
    \centering
    \caption{How the operations most directly relevant to the learning, validation, and reset analyses are treated. The learned object is the effective $Z$-twirled instrument pair $(M^{(0)},M^{(1)})$.}
    \label{tab:operation-ledger}
    \small
    \begin{tabular}{@{}lll@{}}
        \toprule
        \parbox[t]{0.20\textwidth}{Operation} & \parbox[t]{0.18\textwidth}{Treatment} & \parbox[t]{0.55\textwidth}{Role in the analysis} \\
        \midrule
        \parbox[t]{0.20\textwidth}{Initial $I,X,Y,Z$ gates} & \parbox[t]{0.18\textwidth}{Assumed ideal} & \parbox[t]{0.55\textwidth}{Their uniform average supplies the $I/2$ input used for learning.} \\[0.5ex]
        \parbox[t]{0.20\textwidth}{Pre-MCM $Z$ gates} & \parbox[t]{0.18\textwidth}{Part of the learned instrument} & \parbox[t]{0.55\textwidth}{Their implemented effects are included in $(M^{(0)},M^{(1)})$.} \\[0.5ex]
        \parbox[t]{0.20\textwidth}{Mid-circuit measurement} & \parbox[t]{0.18\textwidth}{Learned} & \parbox[t]{0.55\textwidth}{Its readout and population transitions are learned together in $M^{(o)}$.} \\[0.5ex]
        \parbox[t]{0.20\textwidth}{Validation $X$ gates} & \parbox[t]{0.18\textwidth}{Characterized separately} & \parbox[t]{0.55\textwidth}{The fixed noisy-$X$ matrix $X_{p_X}$ is used in the validation calculation.} \\[0.5ex]
        \parbox[t]{0.20\textwidth}{Feedforward reset $X$} & \parbox[t]{0.18\textwidth}{Assumed ideal} & \parbox[t]{0.55\textwidth}{It is taken to be perfect in the reported reset analysis.} \\
        \bottomrule
    \end{tabular}
\end{table}

\paragraph{Independent lower-shot hardware runs.}
We also test the reconstruction with independent lower-shot IBM runs on the same 60 qubits and 3 processors.
This hardware test complements the numerical study of small estimator denominators in SM Sec.~\ref{app:synthetic-study}.

At each shot budget, we analyze 30000 reconstruction attempts collected over 500 repetitions.
Every attempt is counted as either a success or a failure.
Each successful four-band reconstruction is compared with its million-shot reference.
Table~\ref{tab:ibm-downsampling} defines $d$ as the root-mean-square difference between these endpoints.

At $N=1000$ shots per Pauli input setting, 94.4\% of the reconstructions succeed and 5.6\% fail.
Among the successes, the median $d$ is 0.0072.
At $N=10000$, the success rate rises to 99.4\%, the failure rate falls to 0.6\%, and the median $d$ among successful reconstructions is 0.00225.

These independent IBM results are separate from the numerical simulations in Fig.~\ref{fig:synthetic-study}, which start from known model parameters.
Table~\ref{tab:ibm-downsampling} gives the hardware values.
Lower shot counts worked well on the measured devices.
Required counts elsewhere depend on the instrument and hardware platform.

\begin{table}[t]
    \centering
    \caption{Empirical lower-shot reconstruction results from independent IBM hardware runs. For the four reported gauge bands $[\ell_q,u_q]$, the endpoint error relative to the $10^6$-shot reference is $d=\left[\frac{1}{8}\sum_q\left\{(\ell_q-\ell_q^{\mathrm{ref}})^2+(u_q-u_q^{\mathrm{ref}})^2\right\}\right]^{1/2}$. The table reports the median over successful qubit--run pairs; brackets give the central 95\% interval. Both successful and failed reconstructions remain in the counts.}
    \label{tab:ibm-downsampling}
    \small
    \begin{tabular}{@{}rrrrrr@{}}
        \toprule
        Shots/setting & Attempts & Successes & Failures & Failure (\%) & \shortstack{Median $d$ [interval]\\(successful reconstructions)} \\
        \midrule
        1{,}000 & 30{,}000 & 28{,}320 & 1{,}680 & 5.6 & 0.0072 [0.0019, 0.0205] \\
        10{,}000 & 30{,}000 & 29{,}820 & 180 & 0.6 & 0.00225 [0.00062, 0.0066] \\
        100{,}000 & 30{,}000 & 29{,}994 & 6 & 0.02 & 0.00072 [0.0002, 0.0021] \\
        \bottomrule
    \end{tabular}
\end{table}

\paragraph{Why length-3 chains suffice.}
\label{rem:ibm-length-3-sufficient}
The measured length-$3$ distribution contains the probabilities of all eight three-bit readout strings.
By summing appropriate strings, we obtain $S_n=p(0^n)$ and $T_n=p(1^n)$ for $n=1,2,3$.
With the maximally mixed input, Corollary~\ref{cor:tr-det-from-length3} then gives $p(0)=S_1$, $\Tr(M^{(0)})$, $\det(M^{(0)})$, $\Tr(M^{(1)})$, and $\det(M^{(1)})$, provided $S_1^2\neq S_2$ and $T_1^2\neq T_2$.
Under the two spanning conditions of Theorem~\ref{thm:one-dof-nonidentifiability}, these five quantities determine the instrument up to one gauge parameter.
No circuit longer than three consecutive MCMs is therefore needed for learning.

The length-$10$ $X$-interleaved circuits below have a different purpose: they test predictions made by the learned instrument.
Optional longer chains containing only MCMs are used only to check the recurrence relations.

\emph{Validation circuits.}---To test the learned instrument on a circuit family reserved for prediction, we perform a separate validation run.
Each validation circuit repeats the sequence ``$X$ gate followed by MCM'' ten times:
\begin{equation}
    \label{eq:ibm-validation-circuit}
    \ket{0} \;\xrightarrow{\;\text{noisy prep}\;}\;
    \bigl[\, X \to \mathrm{MCM}\,\bigr]^{L_{\mathrm{val}}},
\end{equation}
producing a length-$L_{\mathrm{val}}$ readout string $w \in \{0,1\}^{L_{\mathrm{val}}}$.
The $X$ gate exchanges the two computational-basis states before every MCM.
Successive MCMs therefore probe the $\ket{1}$-like and $\ket{0}$-like states in alternation.
If the MCM-induced excitation and decay probabilities differ, this repeated alternation produces a measurable change in the readout-string distribution.
The validation circuit is therefore particularly sensitive to the distinction between $\eta_\uparrow$ and $\eta_\downarrow$.

These circuits differ from the learning circuits, which contain consecutive MCMs without interleaved $X$ gates.
They consequently test whether the learned instrument predicts a new circuit family.

For each selected qubit, we execute the validation circuit with $N=10^6$ shots, starting from the hardware's noisy $\ket{0}$ preparation.
Both predictors use the same preparation error $\epsilon_{\mathrm{SP}}$, inferred after instrument learning as explained in SM Sec.~\ref{app:spam}, and the same backend-reported Pauli-$X$ error parameter $p_X$.
The MCM instrument is learned entirely from the length-$3$ learning circuits.
The validation data are reserved for the prediction test.
The precise prediction rules are given below.

\paragraph{Separation of learning and validation data.}
\label{rem:ibm-learning-validation-separation}
Because reconstruction uses the length-$3$ learning circuits, agreement with the validation data tests prediction on the separate $X$-interleaved circuit family.

\subsection{Validation metrics, \prlFigResultsRef{}(b) predictors, and diagnostics}
\label{subsec:validation-metrics}

We compare each predictor with the validation data at two levels.
First, at every MCM position we compare the predicted and measured Pauli-$Z$ readout observable; this is the quantity shown in \prlFigResultsRef{}(b).
Second, we compare the predicted and measured probabilities of all length-$10$ readout strings.
The first comparison shows how the error changes with circuit depth, while the second checks the complete readout distribution.

\paragraph{Predictors used in \prlFigResultsRef{}(b).}
\label{subsec:validation-fig2-predictors}
For each characterized qubit, both \prlFigResultsRef{}(b) predictors use the same noisy initial probability vector and the same backend-reported Pauli-$X$ error parameter $p_X$.
\begin{equation}
    \label{eq:validation-initial-x-map}
    \bm{\pi}_0
    =
    \begin{pmatrix}
        1-\epsilon_{\mathrm{SP}} \\
        \epsilon_{\mathrm{SP}}
    \end{pmatrix},
    \qquad
    X_{p_X}
    =
    (1-p_X)
    \begin{pmatrix}
        0 & 1 \\
        1 & 0
    \end{pmatrix}
    +
    p_X I.
\end{equation}
Here $\epsilon_{\mathrm{SP}}$ is the noisy-$\ket{0}$ preparation error defined by Eq.~\eqref{eq:epsilon-sp-from-rz} under the mixture convention \eqref{eq:spam-epsilon-sp-mixture}.
It is inferred after instrument learning from the noisy-$\ket{0}$ data and the learned readout matrix $A(t)$, as explained in SM Sec.~\ref{app:spam}.
The symbol $I$ denotes the $2\times2$ identity.

The five scalar markers in \prlFigResultsRef{}(a) use the visualization-only midpoint convention of Eq.~\eqref{eq:q-mid-def}.
Each marker is the midpoint of that quantity's own band over $\mathcal{T}_{\mathrm{hw}}^{(+)}$.
These five scalar midpoints are computed independently and may correspond to different values of $t$.
For \prlFigResultsRef{}(b), the dashed point curve is the midpoint of the prediction envelope obtained by carrying the separately reported parameter bands through the validation sequence.
Scientific reporting uses the full gauge bands in both cases.

\emph{Dashed-line full-instrument predictor.}---The dashed curves use the reconstructed MCM matrix pair $(M^{(0)},M^{(1)})$, so readout confusion, measurement-induced basis-state transitions, and their outcome-conditioned correlations remain separate.
For the first $\ell$ validation readouts $w_{1:\ell}=w_1\cdots w_\ell$, define the unnormalized conditional population vector by
\begin{align}
    \label{eq:validation-full-branch-update}
    \bm{u}^{\mathrm{full}}_0
    &\coloneqq
    \bm{\pi}_0, \nonumber \\
    \bm{u}^{\mathrm{full}}_{\ell}(w_{1:\ell})
    &\coloneqq
    M^{(w_\ell)} X_{p_X}\,
    \bm{u}^{\mathrm{full}}_{\ell-1}(w_{1:\ell-1}),
    \qquad
    \ell=1,\dots,L_{\mathrm{val}}, \nonumber \\
    p_{\mathrm{full}}(w)
    &\coloneqq
    \bm{1}^{\mathsf{T}}\bm{u}^{\mathrm{full}}_{L_{\mathrm{val}}}(w).
\end{align}
The full-string expression above follows every possible readout history.
For \prlFigResultsRef{}(b), the required quantity is the probability that each MCM reports $0$.
We calculate it directly in three steps: apply the noisy $X$ gate, calculate the probability that the next MCM reports $0$, and sum over both MCM outcomes to obtain the state probabilities used at the next position.
This gives
\begin{align}
    \label{eq:validation-full-marginal-update}
    \bm{v}^{\mathrm{full}}_0
    &\coloneqq
    \bm{\pi}_0, \nonumber \\
    \bm{v}^{\mathrm{full},-}_{\ell}
    &\coloneqq
    X_{p_X}\bm{v}^{\mathrm{full}}_{\ell-1}, \nonumber \\
    p^{\mathrm{full}}_{\ell}(0)
    &\coloneqq
    \bm{1}^{\mathsf{T}}M^{(0)}\bm{v}^{\mathrm{full},-}_{\ell}, \nonumber \\
    \bm{v}^{\mathrm{full}}_{\ell}
    &\coloneqq
    \bigl(M^{(0)}+M^{(1)}\bigr)\bm{v}^{\mathrm{full},-}_{\ell}.
\end{align}

\emph{Solid-line confusion-matrix-based predictor.}---The solid curves use a simpler MCM model constructed from the same learned point values.
First, it combines same-step state changes with the readout errors to form two effective readout-error probabilities.
It then removes all MCM-induced changes to the state that would carry forward to the next step.
With the learned rates of Definition~\ref{def:readout-backaction-error-rates}, define
\begin{align}
    \label{eq:validation-spam-collapse}
    \bar{\epsilon}_{1\mid 0}^{\mathrm{cm}}
    &\coloneqq
    (1-\eta_{\uparrow})\,\epsilon_{1\mid 0}
    +
    \eta_{\uparrow}\,(1-\epsilon_{0\mid 1}),
    \nonumber \\
    \bar{\epsilon}_{0\mid 1}^{\mathrm{cm}}
    &\coloneqq
    (1-\eta_{\downarrow})\,\epsilon_{0\mid 1}
    +
    \eta_{\downarrow}\,(1-\epsilon_{1\mid 0}).
\end{align}
These rates are then inserted into diagonal observation matrices
\begin{align}
    \label{eq:validation-solid-matrices}
    M^{(0)}_{\mathrm{solid}}
    &=
    \begin{pmatrix}
        1-\bar{\epsilon}_{1\mid 0}^{\mathrm{cm}} & 0 \\
        0 & \bar{\epsilon}_{0\mid 1}^{\mathrm{cm}}
    \end{pmatrix}, \nonumber \\
    M^{(1)}_{\mathrm{solid}}
    &=
    \begin{pmatrix}
        \bar{\epsilon}_{1\mid 0}^{\mathrm{cm}} & 0 \\
        0 & 1-\bar{\epsilon}_{0\mid 1}^{\mathrm{cm}}
    \end{pmatrix}.
\end{align}
Thus $M^{(0)}_{\mathrm{solid}}+M^{(1)}_{\mathrm{solid}}=I$.
The model can still report the wrong readout.
However, after we sum over the two possible readouts, the MCM leaves the computational-basis state probabilities unchanged.
The full readout-string probabilities are computed as
\begin{align}
    \label{eq:validation-solid-branch-update}
    \bm{u}^{\mathrm{solid}}_0
    &\coloneqq
    \bm{\pi}_0, \nonumber \\
    \bm{u}^{\mathrm{solid}}_{\ell}(w_{1:\ell})
    &\coloneqq
    M^{(w_\ell)}_{\mathrm{solid}} X_{p_X}\,
    \bm{u}^{\mathrm{solid}}_{\ell-1}(w_{1:\ell-1}),
    \qquad
    \ell=1,\dots,L_{\mathrm{val}}, \nonumber \\
    p_{\mathrm{solid}}(w)
    &\coloneqq
    \bm{1}^{\mathsf{T}}\bm{u}^{\mathrm{solid}}_{L_{\mathrm{val}}}(w),
\end{align}
and the corresponding per-depth marginal update is
\begin{align}
    \label{eq:validation-solid-marginal-update}
    \bm{v}^{\mathrm{solid}}_0
    &\coloneqq
    \bm{\pi}_0, \nonumber \\
    \bm{v}^{\mathrm{solid},-}_{\ell}
    &\coloneqq
    X_{p_X}\bm{v}^{\mathrm{solid}}_{\ell-1}, \nonumber \\
    p^{\mathrm{solid}}_{\ell}(0)
    &\coloneqq
    \bm{1}^{\mathsf{T}}M^{(0)}_{\mathrm{solid}}\bm{v}^{\mathrm{solid},-}_{\ell}, \nonumber \\
    \bm{v}^{\mathrm{solid}}_{\ell}
    &\coloneqq
    \bigl(M^{(0)}_{\mathrm{solid}}+M^{(1)}_{\mathrm{solid}}\bigr)
    \bm{v}^{\mathrm{solid},-}_{\ell}
    =
    \bm{v}^{\mathrm{solid},-}_{\ell}.
\end{align}

\paragraph{Metric plotted in \prlFigResultsRef{}(b): depth-resolved Pauli-$Z$ absolute deviation.}
At each position $\ell$, we first calculate the probability that the $\ell$th readout is $0$, irrespective of the earlier readouts.
This is obtained by summing the probabilities of all length-$\ell$ prefixes that end in $0$.
We then convert that probability into the Pauli-$Z$ readout observable using $z_\ell=2p_\ell(0)-1$.
For each depth $\ell\in\{1,\dots,L_{\mathrm{val}}\}$, define
\begin{equation}
    \label{eq:validation-marginal-zero}
    \widehat{p}_{\mathrm{val},\ell}(0)
    \;\coloneqq\;
    \sum_{u\in\{0,1\}^{\ell-1}} \widehat{p}_{\mathrm{val}}(u0),
    \qquad
    p_{\mathrm{full},\ell}(0)
    \;\coloneqq\;
    \sum_{u\in\{0,1\}^{\ell-1}} p_{\mathrm{full}}(u0),
    \qquad
    p_{\mathrm{solid},\ell}(0)
    \;\coloneqq\;
    \sum_{u\in\{0,1\}^{\ell-1}} p_{\mathrm{solid}}(u0),
\end{equation}
where $u0$ denotes a length-$\ell$ initial string ending in readout bit $0$, with the corresponding marginals defined in SM Sec.~\ref{app:prob-dict}.
The corresponding Pauli-$Z$ observables are
\begin{equation}
    \label{eq:validation-z-observable}
    \widehat{z}_{\mathrm{val},\ell}
    \;\coloneqq\;
    2\widehat{p}_{\mathrm{val},\ell}(0)-1,
    \qquad
    z_{\mathrm{full},\ell}
    \;\coloneqq\;
    2p_{\mathrm{full},\ell}(0)-1,
    \qquad
    z_{\mathrm{solid},\ell}
    \;\coloneqq\;
    2p_{\mathrm{solid},\ell}(0)-1.
\end{equation}

For a device with selected qubit set $Q_{\mathrm{dev}}$ (of size $20$ in the hardware runs reported here), the plotted dashed-line quantity is the average absolute deviation
\begin{equation}
    \label{eq:validation-aad-full}
    \Delta_{\mathrm{full}}(\ell)
    \;\coloneqq\;
    \frac{1}{|Q_{\mathrm{dev}}|}
    \sum_{q\in Q_{\mathrm{dev}}}
    \Bigl|
    \widehat{z}^{(q)}_{\mathrm{val},\ell}
    -
    z^{(q)}_{\mathrm{full},\ell}
    \Bigr|,
\end{equation}
and the plotted solid-line quantity is
\begin{equation}
    \label{eq:validation-aad-solid}
    \Delta_{\mathrm{solid}}(\ell)
    \;\coloneqq\;
    \frac{1}{|Q_{\mathrm{dev}}|}
    \sum_{q\in Q_{\mathrm{dev}}}
    \Bigl|
    \widehat{z}^{(q)}_{\mathrm{val},\ell}
    -
    z^{(q)}_{\mathrm{solid},\ell}
    \Bigr|.
\end{equation}
Because $z_\ell=2p_\ell(0)-1$ for binary readout, these Pauli-$Z$ deviations are exactly twice the corresponding absolute deviations in the marginals for readout bit $0$.
With $N=10^6$ shots per circuit, the expected fluctuations from finite counting statistics are much smaller than the deviations of order $10^{-2}$ shown in the main text.
The plotted deviations therefore primarily measure the remaining disagreement between the hardware data and the two predictors.
Ordinary shot noise contributes much less at this scale.

\paragraph{Additional whole-distribution check.}
As a second check, we compare the measured and predicted probabilities of all $2^{10}$ possible length-$10$ readout strings.
We use the total-variation (TV) distance
\begin{equation}
    \label{eq:validation-tv}
    \mathcal{D}_{\mathrm{TV}}\bigl(\widehat{p}_{\mathrm{val}}, p_{\mathrm{full}}\bigr)
    = \frac{1}{2} \sum_{w \in \{0,1\}^{10}} \bigl| \widehat{p}_{\mathrm{val}}(w) - p_{\mathrm{full}}(w) \bigr|.
\end{equation}
This quantity is zero when the two distributions agree exactly and increases as their probabilities differ.
We calculate it separately for the full-instrument predictor and the confusion-matrix-based predictor.
With $N=10^6$ shots per circuit, the expected contribution from finite counting statistics is well below $10^{-3}$ for these datasets.
Distances of order $10^{-2}$ therefore mainly reflect disagreement between the measured distribution and the corresponding predictor.

\paragraph{Interpretation of the solid-line confusion-matrix-based baseline.}
\label{subsec:validation-baselines}

The solid-line predictor uses the effective readout errors in Eq.~\eqref{eq:validation-spam-collapse}, derived from the MCM rates inferred by our protocol.
The IBM backend's terminal-readout benchmark plays no role in this construction.
Both predictors also use the same inferred $\epsilon_{\mathrm{SP}}$ in the initial state.

After summing over outcomes, the solid model leaves the computational-basis populations unchanged.
The dashed predictor carries the learned MCM state changes and readout--state correlations through the sequence.
The difference between the curves measures the resulting improvement in predictive accuracy.

\paragraph{Diagnostics for model violations.}
\label{subsec:validation-diagnostics}

The reconstruction assumes that every MCM is described by the same stationary two-state model and that no population leaks outside the computational subspace (Assumption~\ref{ass:stationarity-markovianity}).
We first learn this model from the length-$3$ circuits.
The learned model then makes definite predictions for longer MCM chains.
By comparing those predictions with probabilities measured from the longer chains, we can test whether the assumptions remain consistent with the hardware data.
A disagreement larger than the finite-shot uncertainty indicates that at least one assumption is inaccurate.
Further diagnostics are needed to identify the cause.

\begin{proposition}[Consistency diagnostics via overcomplete relations]
    \label{prop:consistency-diagnostics}
    Let $\widehat{p}$ be an empirical readout-string distribution measured for bare MCM chains of length $L > 3$.
    If the underlying process perfectly satisfies the two-state structural assumptions (Definition~\ref{def:z-twirled-mcm-instrument}), then the following statements hold.
    \begin{enumerate}[label=(\roman*), leftmargin=2.2em]
        \item \textbf{Cayley--Hamilton recurrence tightness.}
        The all-zero/all-one readout-string probabilities must satisfy the overcomplete recurrences of Proposition~\ref{prop:overcomplete-long-string-constraints}; explicitly, for every $n=0,\dots,L-2$,
              \begin{align}
                  \label{eq:diagnostic-ch-recurrence}
                  S_{n+2} - \Tr\!\big(M^{(0)}\big)S_{n+1} + \det\!\big(M^{(0)}\big)S_n &= 0, \nonumber \\
                  T_{n+2} - \Tr\!\big(M^{(1)}\big)T_{n+1} + \det\!\big(M^{(1)}\big)T_n &= 0.
              \end{align}
              Here the trace and determinant invariants are those extracted from the length-$3$ readout-string distribution (or from an equivalent overdetermined fit).
              Empirical residuals that exceed finite-shot uncertainty are evidence for drift, leakage, non-Markovianity, or another violation of the base model.
        \item \textbf{Hankel matrix rank.}
              Any finite Hankel matrix $H_{u,v}=p(uv)$ formed from the bare-MCM string probabilities must have rank at most $2$.
              A statistically stable third singular value beyond bootstrap uncertainty is evidence that the observed process requires an effective classical state dimension greater than two, as would occur with leakage from the computational subspace or when other unmodeled degrees of freedom affect the measurement record.
    \end{enumerate}
\end{proposition}

\noindent Finite shot counts produce small nonzero recurrence residuals and third singular values even when the model is correct.
We use the bootstrap procedure in Protocol~\ref{prot:stats-bootstrap-uncertainty-bands} to estimate the range expected from these counting fluctuations.
If an observed value lies outside the corresponding $99\%$ bootstrap interval, we report a statistically detectable disagreement with the two-state model.
Drift, leakage, memory effects, and other omitted processes can each produce a detectable disagreement.
Identifying the cause requires additional evidence.
The calculation uses only the optional longer bare-MCM chains.

\end{document}